\documentclass[reqno]{amsart}
\usepackage{graphicx}
\graphicspath{ {./figures/} }
\usepackage[left=3cm,right=3cm,top=3cm,bottom=3cm]{geometry}
\usepackage{amsmath, amssymb,amsthm}
\usepackage[scr=rsfs]{mathalpha}
\usepackage[shortlabels]{enumitem}
\usepackage{mlmodern}
\usepackage{mathtools} 
\usepackage{stmaryrd} 
\usepackage{esint} 
\usepackage[export]{adjustbox}
\usepackage{tikz}
\usetikzlibrary{decorations.pathreplacing,patterns,external}
\usepackage{diagbox}
\usepackage{mlmodern}
\usepackage{upref} 
\usepackage[colorlinks]{hyperref}
\hypersetup{citecolor=blue,filecolor=blue,linkcolor=blue,urlcolor=navyblue}
\definecolor{navyblue}{rgb}{0.0, 0.0, 0.5}

\newtheorem{thm}{Theorem}[section]
\newtheorem{prop}[thm]{Proposition}
\newtheorem{lem}[thm]{Lemma}
\newtheorem{cor}[thm]{Corollary}
\theoremstyle{definition}

\newtheorem{rmk}{Remark}[section]

\newtheorem*{claim*}{Claim}

\newcommand{\N}{\mathbb{N}}
\newcommand{\Z}{\mathbb{Z}}
\newcommand{\R}{\mathbb{R}}
\newcommand{\C}{\mathbb{C}}

\newcommand{\T}{\mathbb{T}}

\renewcommand{\Re}{\operatorname{\mathrm{Re}}}

\newcommand{\intbrr}[1]{\llbracket#1\rrbracket}

\newcommand{\E}{\mathbb{E}}
\newcommand{\oneb}{\mathbf{1}}
\renewcommand{\P}{\mathbb{P}}
\newcommand{\chf}{\operatorname{chf}}
\newcommand{\CN}{\mathcal{CN}}
\newcommand{\RN}{\mathcal{N}}

\newcommand{\x}{\mathbf{x}}

\newcommand{\jj}{\alpha}
\newcommand{\jjtwo}{\beta}

\newcommand{\cs}[1]{#1'\!\cdot #1}
\newcommand{\csn}[1]{#1}
\newcommand{\Lip}{\mathrm{Lip}}
\newcommand{\Tr}{\operatorname{Tr}}

\newcommand{\Wg}{\operatorname{Wg}}
\newcommand{\ol}[1]{\overline{#1}}

\newcommand{\Opkl}{\operatorname{Op}_{k,\ell}}
\newcommand{\Ba}{\hat{B}}
\DeclareMathAlphabet{\mathcalboondox}{U}{BOONDOX-calo}{m}{n}

\newcommand{\fmat}{(*)}

\newcommand{\dsim}{\overset{d}{\sim}}

\newcommand{\dfoura}{\langle a_0\rangle_{\mathcal{B}^{(4)}_{0,2}}}
\newcommand{\mz}{M}
\newcommand{\frozen}{\mathcal{F}}

\newcommand{\jeven}{J_{\mathrm{even}\;\jj}}
\newcommand{\jodd}{J_{\mathrm{odd}\;\jj}}

\newcommand\numberthis{\stepcounter{equation}\tag{\theequation}}
\numberwithin{equation}{section}

\newcounter{tikznumber}
\newcommand{\tp}[1]{\includegraphics[valign=c]{walsh_fluctuations-figure\thetikznumber.pdf}\stepcounter{tikznumber}}

\begin{document}

\title{Quantum variance and fluctuations for Walsh-quantized baker's maps}

\author[L. Shou]{Laura Shou}
\address[L. Shou]{Joint Quantum Institute, Department of Physics, University of Maryland, College Park, MD 20742, USA}
\email{lshou@umd.edu}

\begin{abstract}
The Walsh-quantized baker's maps are models for quantum chaos on the torus. We show that for all baker's map scaling factors $D\ge2$ except for $D=4$, typically (in the sense of Haar measure on the eigenspaces, which are degenerate) the empirical distribution of the scaled matrix element fluctuations $\sqrt{N}\{\langle \varphi^{(j)}|\Opkl(a)|\varphi^{(j)}\rangle-\int_{\T^2}a\}_{j=1}^{N}$ for a random eigenbasis $\{\varphi^{(j)}\}_{j=1}^{N}$ is asymptotically Gaussian in the semiclassical limit $N\to\infty$, with variance given in terms of classical baker's map correlations. This determines the precise rate of convergence in the quantum ergodic theorem for these eigenbases. We obtain a version of the Eigenstate Thermalization Hypothesis (ETH) for these eigenstates, including a limiting complex Gaussian distribution for the off-diagonal matrix elements, with variances also given in terms of classical correlations. The presence of the classical correlations highlights that these eigenstates, while random, have microscopic correlations that differentiate them from Haar random vectors. For the single value $D=4$, the Gaussianity of the matrix element fluctuations depends on the values of the classical observable on a fractal subset of the torus.
\end{abstract}

\maketitle

\section{Introduction}
Quantum ergodicity refers to the property where a limiting density one subset of eigenfunctions of a quantum operator equidistribute throughout phase space in the semiclassical limit.
It was first proved for eigenfunctions of the Laplacian (=quantum Hamiltonian in this case) corresponding to ergodic geodesic flow on the unit tangent bundle of a manifold \cite{Shnirelman1974,Zelditch1987,CdV1985}, and later extended to other models, including for quantizations of various chaotic maps on the torus $\T^2=\R^2/\Z^2$.
Typically, to establish a quantum ergodic theorem one shows that a quantum variance tends to zero in the semiclassical limit. For torus maps this is typically the statement
\begin{align}\label{eqn:qv}
\frac{1}{N}\sum_{j=1}^N\left|\langle\varphi^{(j)}|\operatorname{Op}_N(a)|\varphi^{(j)}\rangle-\int_{\T^2}a(\x)\,d\x\right|^2\to0,\quad\text{as }N\to\infty,
\end{align}
where $\{\varphi^{(j)}\}_{j=1}^N$ is an orthonormal basis of the quantum operator $\hat{U}_N$ and $\operatorname{Op}_N(a)$ is a quantum observable associated with a classical observable $a:\T^2\to\R$.
The \emph{rate} at which the quantum variance tends to zero has been of significant interest in the quantum chaos literature. 

The expectation from the physics literature
\cite{FeingoldPeres1986,Eckhardtetal1995} 
is that for generic chaotic systems, the scaled fluctuations $\sqrt{N}\{\langle\varphi^{(j)}|\operatorname{Op}_N(a)|\varphi^{(j)}\rangle-\int_{\T^2}a(\x)\,d\x\}_{j=1}^N$ should look Gaussian with variance $v_a$ given by the classical correlations, in the case of torus maps as
\begin{align}\label{eqn:classical-var}
v_a=\sum_{t=-\infty}^\infty\int_{\T^2}a_0(\x)a_0(S^t\x)\,d\x,
\end{align}
where $S:\T^2\to\T^2$ is the classical map associated with $\hat{U}_N$ and $a_0:=a-\int_{\T^2}a$ is the centered observable.
In particular, the quantum variance in \eqref{eqn:qv} should have decay of order $N^{-1}$.

There has been much work in the past few decades concerning the matrix element fluctuations for \emph{arithmetic} systems, for which the generic expectations have been shown to not always hold. 
In \cite{KurlbergRudnick2005}, the authors determined the variance and fourth moment for matrix element fluctuations for Hecke eigenfunctions of quantum cat maps on the torus in prime dimensions, and 
conjectured a general distribution for the fluctuations. While the proved $O(N^{-1})$ decay of the quantum variance agrees with the order of the generic expectation, the distribution is shown to not be Gaussian.
The fifth moment was determined in \cite{Rosenzweig2011} and shown to agree with the (non-Gaussian) conjectured distribution.
Concerning the quantum variance in general systems, a logarithmic upper bound
\cite{Zelditch1994,Schubert2006} on the rate of decay is known to hold, but is expected to be far from optimal in many cases.
However, \cite{Schubert2008} showed the logarithmic bound is sharp for quantum cat maps, utilizing the large eigenspace degeneracies of the quantum cap maps for certain dimensions $N_k\to\infty$.
Additionally, \cite{Kelmer2008} considered quantum cat maps in composite dimensions and showed there are Hecke eigenstates with slow quantum variance decay $\gtrsim N^{-1/3}$, and also proved a limiting distribution for the fluctuations along prime power dimensions, which also is not Gaussian.
Besides cat maps on the torus $\T^2$, the quantum variance has been studied for higher dimensional cat maps \cite{Kelmer2005,Kelmer2010} and on arithmetic manifolds \cite{LuoSarnak2004,Zhao2010,Nelson2019,SarnakZhao2019,Huang2021,HuangLester2023}, where it typically exhibits number-theoretic deviations from the generic prediction \eqref{eqn:classical-var}.

Another motivation for studying matrix element fluctuations is in view of the Eigenstate Thermalization Hypothesis (ETH) \cite{Deutsch1991, Srednicki1994} from quantum dynamics.
ETH is expected to govern thermalization in generic chaotic quantum systems, and loosely states that individual eigenstates behave thermally, with a specific ansatz for the matrix elements of physical observables $\hat{O}$ in the energy eigenbasis $|E_i\rangle$ of the form,
\begin{align}\label{eqn:eth-phys}
\langle E_i|\hat{O}|E_j\rangle &= \langle \hat{O}\rangle_E \delta_{ij} + F(E_i,E_j)\Delta_{ij},
\end{align}
where \cite{DKPR2016, Deutsch-rpp}
\begin{itemize}
\item $\langle \hat{O}\rangle_E$ is a smooth function of $E$ and is a microcanonical or statistical mechanical average at energy $E$.
\item $F$ is a smooth function of order unity, which may tend to zero for large $|E_i-E_j|$.
\item $\Delta_{ij}$ are random variables with zero mean and variance of order $\Omega^{-1/2}(E)$, where $\Omega(E)$ is the density of states.
\item The randomness in $\Delta_{ij}$ refers to considering the fluctuations over a set of nearby eigenstates.
\end{itemize}
For single particle systems, the on-diagonal prediction of ETH is closely related to ``quantum unique ergodicity'' (QUE) \cite{RudnickSarnak1994}, which is the property that {all} eigenfunctions equidistribute, rather than just requiring a limiting density one subset to as in quantum ergodicity.
Certain forms of ETH predict the individual fluctuations in the convergence in QUE, which is what we will be interested in here.
ETH is also closely related to random matrix theory (RMT), and there have been many recent results, e.g. \cite{CES2021,cipolloni2022normal,BL2022,ER2024,adhikari2024eigenstate} to list just a few, for random Wigner and generalized Wigner matrices in this context.

\vspace{2mm}
In this paper, we study the quantum variance and matrix element fluctuations for the Walsh-quantized baker's maps, a ``Walsh'' quantization of the classical $D$-baker's maps on the torus. 
For an integer $D\ge2$, the classical $D$-baker's map $B:\T^2\to\T^2$ is the map given by
\begin{align}\label{eqn:classical-baker}
B(q,p)&=\left(Dq-\lfloor Dq\rfloor, \frac{p+\lfloor Dq\rfloor}{D}\right).
\end{align}
Equivalently, for $q=0.q_1q_2\cdots$ and $p=0.p_1p_2\cdots$ written in base $D$, the $D$-baker's map is the left shift on sequences $\cdots p_2p_1.q_1q_2\cdots$, and is thus an ergodic and chaotic map.
The case $D=2$ is the ``standard'' baker's map. In this paper we may refer to any of the $D$-baker's maps as a baker's map, and will specify the value of $D$ when necessary.

The Walsh-quantized baker's map is a quantum, or quantized, version of the classical $D$-baker's map. 
For each $N=D^k$, $k\in\N$, it is an $N\times N$ unitary matrix $\Ba_k$ acting on a Hilbert space $\mathcal{H}_N\cong \C^N$, which recovers the classical $D$-baker's map (in a Walsh quantization sense, as described in Section~\ref{sec:background}) in the semiclassical limit $N\to\infty$. Explicit definitions for $\Ba_k$ are given in \eqref{eqn:walsh-action} and \eqref{eqn:wbaker}.
Quantum observables $\operatorname{Op}_{k,\ell}(a):\mathcal{H}_N\to\mathcal{H}_N$ are constructed from a classical Lipschitz observable $a\in \Lip(\T^2)$, with a chosen parameter $\ell=\ell(k)\in\intbrr{0:k}:=\{0,\ldots,k\}$.

While this Walsh-quantized baker's map quantization is not one of the more usual \emph{Weyl} quantizations \cite{BV1989}, it has still been of interest in both the physics and mathematics literature \cite{SchackCaves2000,TracyScott2002,NZ2005,AN2007,NZ2007}, in part due to its special properties. 
In particular, the Walsh-quantized baker's map has highly degenerate eigenspaces, so that there are many different choices of eigenbases. 
Our results concern ``generic'' eigenbases, or those chosen randomly according to Haar measure in each eigenspace.
As we explain later, even though these eigenstates are chosen randomly from highly degenerate eigenspaces, they behave distinctly differently at this scale compared to Haar random vectors.

\subsection{Main results}
We now state the main results of this paper on matrix element fluctuations.
Throughout this paper, $\RN(\mu,\sigma^2)$ will denote the real Gaussian distribution with mean $\mu$ and variance $\sigma^2$, and $\CN(0,\sigma^2)$ the complex Gaussian distribution consisting of independent real and imaginary parts which are each real Gaussian with mean 0 and variance $\sigma^2/2$.
\begin{thm}\label{thm:mat}
Let $D\ge2$ and consider the $N\times N$ Walsh-quantized $D$-baker's map $\Ba_k$ (defined in \eqref{eqn:walsh-action} or \eqref{eqn:wbaker}) for $N=D^k$, $k\in\N$.
Consider random orthonormal eigenbases $(\varphi^{(k,m)})_{m=1}^{D^k}$ of $\Ba_k$ chosen according to Haar measure within each eigenspace, and define the scaled matrix element fluctuations (as a function of the randomly chosen eigenbasis) for a real Lipschitz observable $a:\T^2\to\R$ as
\[
F_j^{(k)}:=\sqrt{N}\left(\langle \varphi^{(k,j)}|\operatorname{Op}_{k,\ell}(a)|\varphi^{(k,j)}\rangle-\int_{\T^2}a\right),
\]
where $\ell=\ell(k)\in\intbrr{0:k}$ is a quantization parameter as described in \eqref{eqn:opkl}, chosen so that $\min(\ell,k-\ell)\ge3\log_D k$. 
\begin{enumerate}[(i),leftmargin=*]
\item Let $D\ne4$. Then as $k\to\infty$, the scaled quantum variance $\frac{1}{N}\sum_{j=1}^N|F_j^{(k)}|^2$ converges in probability to a classical value $V(a)$, 
\begin{align}\label{eqn:var}
\frac{1}{N}\sum_{j=1}^N|F_j^{(k)}|^2\xrightarrow{p} V(a),
\end{align}
for
\begin{align}\label{eqn:Va}
V(a):=\sum_{t=-\infty}^\infty\left(\int_{\T^2}a_0(\x)a_0(B^t\x)\,d\x+\mathbf{1}_{D\ge3}\int_{\T^2}a_0(\x)a_0(B^tR\x)\,d\x\right),
\end{align}
where $a_0:=a-\int_{\T^2}a$ and $R:(q,p)\mapsto(1-q,1-p)$ is a reflection operator.
Additionally, in the same setting, the empirical distribution $\mu_k:=\frac{1}{N}\sum_{j=1}^N\delta_{F_j^{(k)}}$ converges weakly in probability to a centered Gaussian random variable with variance $V(a)$; that is, for $Z\sim \RN(0,V(a))$ and any $\epsilon>0$ and bounded continuous $f:\R\to\C$,
\begin{align}\label{eqn:conv}
\P_\omega\left[\left|\int_\R f\,d\mu_k-\E_Z[f]\right|>\epsilon\right]\xrightarrow{k\to\infty}0.
\end{align}

\item For $D=4$, we have the convergence in probability
\begin{align}\label{eqn:var4}
\frac{1}{N}\sum_{j=1}^N|F_j^{(k)}|^2\xrightarrow{p}V(a)+\langle a_0\rangle_{\mathcal B^{(4)}_{0,2}}^2,
\end{align}
where
\begin{align}\label{eqn:a0-avg}
\langle a_0\rangle_{\mathcal B^{(4)}_{0,2}}:=\lim_{k\to\infty}\frac{1}{2^k}\sum_{\substack{q\in \mathcal L^{(4)}_{0,2}(\ell)\\ p\in \mathcal L^{(4)}_{0,2}(k-\ell)}}a_0(q,p),
\end{align}
with $\mathcal L^{(4)}_{0,2}(m)$ the set of $x\in[0,1)$ whose base four expansion has length at most $m$ digits and consists only of $0$s and $2$s. In other words, $\langle a_0\rangle_{\mathcal B^{(4)}_{0,2}}$ is the limiting average value of $a_0$ on the fractal set $\mathcal{B}^{(4)}_{0,2}:=\{x\in[0,1)\text{ whose base four expansion has only 0s and 2s}\}$.

The empirical distribution $\mu_k:=\frac{1}{N}\sum_{j=1}^N\delta_{F_j^{(k)}}$ converges weakly in probability to a mixture of two Gaussians $\RN(\pm\langle a_0\rangle_{\mathcal B^{(4)}_{0,2}},V(a))$ defined by the probability density function
\begin{align}\label{eqn:d4pdf}
p(x)&=\frac{1}{2}\left[\frac{1}{\sqrt{2\pi V(a)}}e^{-\Big(x-\langle a_0\rangle_{\mathcal B^{(4)}_{0,2}}\Big)^2/(2V(a))}+\frac{1}{\sqrt{2\pi V(a)}}e^{-\Big(x+\langle a_0\rangle_{\mathcal B^{(4)}_{0,2}}\Big)^2/(2V(a))}\right].
\end{align}
\end{enumerate}
\end{thm}

\begin{rmk}\label{rmk:thm1}
\begin{enumerate}[leftmargin=*]
\item The difference for the $D=4$ case is because of an $a$-dependent term of order $O\left(\frac{2^k}{\sqrt{N}}\right)$ which appears at half the quantum period when calculating $\E_\omega[F_j^{(k)}]$ for even $D\ge3$; see Section~\ref{sec:tr-ind} and also Proposition~\ref{prop:averages}(i).
For even $D\ge6$, this term is always $o(1)$. 
For $D=4$ however, $\sqrt{N}=2^k$ and it turns out the term will be nonvanishing if $\langle a_0\rangle_{\mathcal B^{(4)}_{0,2}}>0$. 
An example of an observable $a_0$ where $\langle a_0\rangle_{\mathcal B^{(4)}_{0,2}}>0$ is $a_0(q,p)=\sin(4\pi q)$.
Thus for $D=4$, the Gaussianity or non-Gaussianity of the empirical distribution for the matrix element fluctuations is dependent on the observable $a$. Half of the eigenspaces will produce matrix element fluctuations that look like $\RN(\dfoura,V(a))$, while the other half will produce those that look like $\RN(-\dfoura,V(a))$.

\item The physics arguments in \cite{FeingoldPeres1986,Eckhardtetal1995} concerning the matrix fluctuations for generic chaotic systems do not strictly apply to discrete time maps, and additionally not to systems with large eigenspace degeneracies. 
Nevertheless, we can still compare the variance \eqref{eqn:Va} given for $D\ne4$ to the generic expectation \eqref{eqn:classical-var}. First we note that the decay rate and Gaussian behavior for $D\ne4$ are in agreement with the physics prediction. 
The variance \eqref{eqn:Va} relates the quantum fluctuations to classical correlations as expected, although for $D\ge3$ it looks different than \eqref{eqn:classical-var} (which is stated for systems without symmetries).
One could argue heuristically that the extra factor in \eqref{eqn:Va} for $D\ge3$ appears naturally due to the classical reflection symmetry $R:(q,p)\mapsto(1-q,1-p)$ of the $D$-baker's maps \cite{BV1989}, which is carried by the quantum map $\Ba_k$ for $D\ge3$, via commutation with $\tilde{R}=(\hat{F}_D^2)^{\otimes k}$, where $\hat{F}_D$ is the $D\times D$ discrete Fourier transform matrix. For $D=2$ this map $\tilde R$ is just the identity so does not represent a symmetry.
On the other hand, the anticanonical time-reversal symmetry of the classical maps \cite{BV1989} suggests we should maybe have a factor of 2 on $V(a)$ due to \cite{Eckhardtetal1995}, which we do not have. 
In general however, it is not actually clear how the symmetry sectors of the classical maps should manifest in the quantum system's spectra and eigenvectors, particularly for either arithmetic systems \cite{HannayBerry1980,LuoSarnakarithmetic,BGGSarithmetic} or torus maps \cite{bakernumerics}, and so we will not be so concerned with symmetry-related differences in \eqref{eqn:Va}.
The case $D=4$, however, differs significantly from the general expectation when $\dfoura\ne0$.

\item In Theorem~\ref{thm:mat}, we do not have to consider all $N$ of the $F_j^{(k)}$. As will be seen from the proof, in \eqref{eqn:var} for $D\ne4$, the average over any number $n(k)\le N$ with $n(k)\to\infty$ of the $F_j^{(k)}$ converges in probability to $V(a)$. Similarly, \eqref{eqn:conv} holds for empirical distributions $\mu_{k,n(k)}=\frac{1}{|J(k)|}\sum_{j\in J(k)}\delta_{F_j^{(k)}}$ with $|J(k)|\to\infty$. We will use this later in Theorem~\ref{thm:eth}. However, for $D=4$, if $\dfoura\ne0$, the empirical distribution convergence depends on which eigenspaces the $F_j$, $j\in J(k)$, are drawn from.

\end{enumerate}
\end{rmk}

To prove Theorem~\ref{thm:mat}, we will prove:
\begin{thm}\label{thm:ind}
Consider random eigenbases drawn from Haar measure in each eigenspace, and the corresponding scaled matrix element fluctuations $F_j^{(k)}$. 
In the same setting as the previous theorem, for any $j\in\intbrr{1:D^k}$ and as $k\to\infty$,
\begin{align}
\E_\omega[F_j^{(k)}]&=\begin{cases}o(1),&D\ne 4\\ (-1)^{\alpha_j}\langle a_0\rangle_{\mathcal B^{(4)}_{0,2}}+o(1),&D=4\end{cases}.\label{eqn:center}
\end{align}
Defining the ``recentered'' fluctuations as
\begin{align}\label{eqn:tilde}
\tilde F_j^{(k)}:=\begin{cases}F_j^{(k)},&D\ne4\\F_j^{(k)}-(-1)^{\jj_j}\langle a_0\rangle_{\mathcal{B}^{(4)}_{0,2}},&D=4\end{cases},
\end{align} 
where $\jj_j\in\intbrr{0:4k-1}$ denotes an eigenspace index\footnote{$\jj_j$ is the index for the eigenspace corresponding to eigenvalue $e^{2\pi i\jj_j/q(k)}$, where $q(k)$ is the quantum period which is $4k$ for $D\ge3$; see Section~\ref{subsec:qbakermap}.},
we also have the variance and higher moment asymptotics for any $D\ge2$, 
\begin{align}
\E_\omega[(\tilde F_j^{(k)})^2]
&= V(a)+o_a(1),\label{eqn:e-var}\\
\E_\omega[(\tilde F_j^{(k)})^p]&=V(a)^{p/2}\E[g^p]+o_{a,p}(1),\text{ any }p\ge 3,\label{eqn:moments}
\end{align}
where $g\sim\RN(0,1)$, and the $o(1)$ terms are uniform over $j\in\intbrr{1:D^k}$. 

As a consequence, we have Gaussian individual fluctuations in QUE: for $D\ne4$ and any sequence $j(k)\in\intbrr{1:N}$, as $k\to\infty$,
\begin{align}\label{eqn:clt}
\sqrt{\frac{N}{V(a)}}\left[\langle\varphi^{(k,j(k))}|\Opkl(a)|\varphi^{(k,j(k))}\rangle-\int_{\T^2}a(\x)\,d\x\right]\xrightarrow{d}\RN(0,1).
\end{align}
For $D=4$ and any sequence $j(k)$ such that the eigenspace index $\alpha_{j(k)}$ is eventually a fixed parity,
\begin{align}\label{eqn:clt4}
\sqrt{N}\left[\langle\varphi^{(k,j(k))}|\Opkl(a)|\varphi^{(k,j(k))}\rangle-\int_{\T^2}a(\x)\,d\x\right]\xrightarrow{d}\RN((-1)^{\alpha_{j(k)}}\langle a_0\rangle_{\mathcal B^{(4)}_{0,2}},V(a)).
\end{align}
\end{thm}
The estimates \eqref{eqn:center} and \eqref{eqn:e-var} will be used to prove the variance convergence \eqref{eqn:var}, and the higher moments \eqref{eqn:moments} will be used to prove the empirical distribution convergence \eqref{eqn:conv}.
In order to prove these estimates, we will need to consider the quantum time evolution $\Ba_k^t$ for $t$ at and near the Ehrenfest time $t\approx k=\log_DN$. For these times, we will use the specific structure of the Walsh-quantized baker's map dynamics to show a small but important amount of cancellation due to varying phases.

\begin{rmk}
We emphasize that while we are considering eigenstates drawn at random from within each eigenspace, they are crucially different from full Haar random orthonormal bases of $\C^N$. If one simply takes a Haar random unit vector $u\in\C^N$, then using Weingarten calculus (see e.g. \eqref{eqn:variance}) and the definition \eqref{eqn:opkl} of the Walsh-quantized observable $\Opkl(a_0)$ for Lipschitz $a_0$, one can show for $\min(\ell,k-\ell)\to\infty$, that the $F_j^{(k)}$ associated with a $\C^N$ Haar random vector have variance
\begin{align}
N\E |\langle u|\Opkl(a_0)|u\rangle|^2&=\frac{1+o(1)}{N}\Tr(\Opkl(a_0)^2)=\int_{\T^2}a_0^2(\x)\,d\x+o_a(1),
\end{align}
which is missing the rest of the classical correlations found in \eqref{eqn:Va}.
This is inevitable however, as a full Haar random basis does not know anything about the classical dynamics. 
While random eigenbases of quantized systems with large eigenspace degeneracies can sometimes be shown to typically behave in many respects like a Haar random basis, including for macroscopic properties like quantum ergodicity or QUE as shown in several systems \cite{Zelditch1992,VanderKam1997,Schwartz-thesis,wbaker}, this is not the case for the matrix element fluctuations. 
It will be the microscopic correlations \emph{between eigenstate coordinates} that will be necessary for producing the classical correlation terms in \eqref{eqn:Va}. 
Note also that here we need at least mixing classical dynamics if the variance \eqref{eqn:classical-var} or \eqref{eqn:Va} is to be a finite number, while quantum ergodicity or QUE can hold even when the corresponding classical system is not ergodic \cite{Zelditch1992,VanderKam1997}. 
\end{rmk}

Finally, we can also consider off-diagonal matrix entries $\sqrt{N}\langle\varphi^{(k,i)}|\Opkl(a)|\varphi^{(k,j)}\rangle$ for $j\ne i$ in view of the Eigenstate Thermalization Hypothesis (ETH). Technically, the argument for ETH (see e.g. \cite{Deutsch-rpp}) does not allow for degenerate eigenspaces (one should desymmetrize the system).
Nevertheless, other than having complex values, which are a result of considering the discrete-time unitary propagators instead of a Hamiltonian, we prove a version of ETH for randomly chosen eigenstates of the Walsh-quantized baker's maps.
Since the ``randomness'' in \eqref{eqn:eth-phys} refers to fluctuations over nearby eigenstates such as those in a small spectral window, and not to any inherent randomness in the system or eigenspace, we will consider an empirical distribution over some number of eigenstates from the same eigenspace(s) of the Walsh-quantized baker's maps.

\begin{thm}\label{thm:eth}
For $N=D^k$, let $\{\varphi^{(k,j)}\}_{j=1}^N$ be an orthonormal eigenbasis for the $D$-baker's map $\Ba_k$ chosen according to Haar measure in each eigenspace $E_\jj$, with corresponding eigenvalues $\{e^{2\pi i\lambda_j}\}_{j=1}^N$.
Let $a:\T^2\to\R$ be a real Lipschitz observable and choose quantization parameter $\ell(k)$ so that $\ell,k-\ell\ge3\log_D k$. 
For any sets of eigenvectors $\{\varphi^{(k,j)}\}_{j\in J_\jj(k)}\subseteq E_\jj$ and $\{\varphi^{(k,i)}\}_{i\in J_\jjtwo(k)}\subseteq E_\jjtwo$ (where $\alpha,\beta$ can depend on $k$, and we allow $\jj=\jjtwo$), consider the empirical spectral measures
\begin{align}
\nu^{(k)}_{\jj\jjtwo}&=\frac{1}{\#\{(j,i)\in J_\jj(k)\times J_\jjtwo(k):j\ne i\}}\sum_{\substack{j\in J_\jj(k),i\in J_\jjtwo(k)\\ j\ne i}}\delta_{\frac{\sqrt{N}}{\sqrt{V(a)}f_a(\jj,\jjtwo)}\langle\varphi^{(k,i)}|\Opkl(a_0)|\varphi^{(k,j)}\rangle},\\
\nu^{(k)}_{\jj}&=\frac{1}{|J_\jj(k)|}\sum_{j\in J_\jj(k)}\delta_{\frac{\sqrt{N}}{\sqrt{V(a)}}\langle\varphi^{(k,j)}|\Opkl(a_0)|\varphi^{(k,j)}\rangle},
\end{align}
where $V(a)$ is as in \eqref{eqn:Va} and
\begin{align}\label{eqn:f}
\begin{aligned}
f_a(\jj,\jjtwo)^2&= \frac{1}{V(a)}\sum_{t=-\infty}^\infty e^{2\pi it(\jj-\jjtwo)}\left(\int_{\T^2}a_0(\x)a_0(B^t\x)\,d\x+\mathbf{1}_{D\ge3}\int_{\T^2}a_0(\x)a_0(B^tR\x)\,d\x\right).
\end{aligned}
\end{align}

\begin{enumerate}[(i)]
\item Let $D\ne 4$. As $k\to\infty$, if $|J_\jj(k)|,|J_\jjtwo(k)|\to\infty$, then
\begin{align}\label{eqn:eth}
\nu^{(k)}_{\jj\jjtwo}\overset{w,\P}{\to}\CN(0,1),
\quad \text{ and }\quad\nu^{(k)}_\jj\overset{w,\P}{\to}\RN(0,1),
\end{align}
where the convergence is weakly in probability. 

\item Let $D=4$. As $k\to\infty$, if $|J_\jj(k)|,|J_\jjtwo(k)|\to\infty$, then
\begin{align}\label{eqn:eth4}
\nu^{(k)}_{\jj\jjtwo}\overset{w,\P}{\to}\CN(0,1).
\end{align}
For a sequence of eigenspace indices $\jj=\jj(k)$ which is eventually always even or always odd,
\begin{align}\label{eqn:eth4-diag}
\nu^{(k)}_\jj\overset{w,\P}{\to}\RN\Big((-1)^{\jj}\dfoura/\sqrt{V(a)},1\Big),
\end{align}
where $(-1)^\jj$ represents the eventual sign of $(-1)^{\jj(k)}$.
\end{enumerate}
\end{thm}
For comparison to the physics literature statement \eqref{eqn:eth-phys}, let $A_{ij}:=\langle\varphi^{(k,i)}|\Opkl(a)|\varphi^{(k,j)}\rangle$ and $\langle A\rangle:=\int_{\T^2}a(\x)\,d\x$. 
For $D\ne4$, informally, one could write \eqref{eqn:eth} as
\begin{align}\label{eqn:eth-ph}
A_{ij}\sim \langle A\rangle \delta_{ij} + \sqrt{\frac{V(a)}{N}} f_a(\jj,\jjtwo)R_{ij},
\end{align}
where $R_{ij}$ are $\CN(0,1)$ complex normal random variables if $i\ne j$, and $R_{ii}$ are standard real $\RN(0,1)$ random variables.
The randomness in \eqref{eqn:eth-ph} is understood to come from considering the distribution over nearby eigenvectors; in the setting of \eqref{eqn:eth} where there are degeneracies, we consider the empirical distribution of a set of eigenvectors chosen according to Haar measure in the eigenspaces.
Note that $f_a(\jj,\jjtwo)=1$ if $\jj=\jjtwo$, and that the definition \eqref{eqn:f} for $f_a(\jj,\jjtwo)^2$ is seen to be a real number for $\jj\ne\jjtwo$ by considering $t\leftrightarrow -t$, and will also be seen to be eventually $\ge0$ by its derivation in Proposition~\ref{prop:var-ind-2}.
Note also that the diagonal case of $\nu^{(k)}_\jj$ in \eqref{eqn:eth} is already stated in Remark~\ref{rmk:thm1}(3); the new statement here is the off-diagonal distributions.

The presence of $V(a)$ indicates the microscopic correlations in the eigenstates of $\Ba_k$ play a role in \eqref{eqn:eth} as well.
For contrast, the statement that
\begin{align}\label{eqn:eth-0}
\max_{i,j\in\intbrr{1:N}}\left|\langle \varphi^{(k,i)}|\Opkl(a)|\varphi^{(k,j)}\rangle-\delta_{ij}\int_{\T^2}a(\x)\,d\x\right|&\le N^{-1/2+\delta}
\end{align}
with high probability as $k\to\infty$ and any $\delta>0$, is easier to prove.
For the maximum over diagonal entries $i=j$, the above follows from the random eigenbasis QUE estimates in \cite{wbaker}, as we explain in \eqref{eqn:ind-que} and Appendix~\ref{sec:que}.
Also in Appendix~\ref{sec:que}, the off-diagonal estimates in \eqref{eqn:eth-0} follow from moment bounds with Weingarten calculus and Markov's inequality, and it is enough to use rougher absolute value estimates that do not consider the correlations between eigenstate coordinates.

\begin{rmk}
To the best of the author's knowledge, the results in Theorem~\ref{thm:mat} appear to be the only matrix element fluctuation proofs for a non-arithmetic quantized chaotic system. 
These results are however made possible by the high eigenspace degeneracies and periodicity, which are non-generic and allow for a probabilistic result.
\end{rmk}

\subsection{Outline and notation}

\begin{itemize}
\item In Section~\ref{sec:background}, we provide background for the classical $D$-baker's maps, Walsh-quantized $D$-baker's maps, and Weingarten calculus for random quadratic forms. 
We also give an outline of the proof method in Section~\ref{subsec:overview}.

\item In Section~\ref{sec:lemmas}, we prove lemmas concerning quantum vs classical time evolution, and apply these to calculate the Haar-expectation values of the mean and variance of the $\{F_j^{(k)}\}_j$, motivating how the classical value $V(a)$ can be extracted from the quantum variance.

\item In Section~\ref{sec:phases}, we prove several lemmas involving cancellation due to phases of the quantum propagator $\Ba_k^t$, which will be used in the remaining sections to prove Theorems~\ref{thm:mat}, \ref{thm:ind}, and \ref{thm:eth}.

\item In Section~\ref{sec:tr-ind}, we prove \eqref{eqn:center} and \eqref{eqn:e-var} on the individual expectation values $\E_\omega F_j^{(k)}$ and $\E_\omega (F_j^{(k)})^2$, and use this to prove the convergence in probability of the quantum variance stated in \eqref{eqn:var} and \eqref{eqn:var4} of Theorem~\ref{thm:mat}.

\item In Section~\ref{sec:exp}, we prove \eqref{eqn:moments} on higher moments, and use this to prove the weak convergence in probability of the empirical distribution to a Gaussian in Theorem~\ref{thm:mat}. This completes the proofs of Theorems~\ref{thm:mat} and \ref{thm:ind}.

\item In Section~\ref{sec:off-diag}, we prove the off-diagonal matrix element fluctuations for Theorem~\ref{thm:eth} on ETH.

\end{itemize}

Throughout this paper, when convenient and clear from context, we may drop superscripts $k$ and instead write, e.g. $\varphi^{(j)}$ for $\varphi^{(k,j)}$, and $F_j$ for $F_j^{(k)}$. We make use of standard little-o and big-O asymptotic notation $o,O$. We will use the interval notation $\intbrr{a:b}:=\{a,a+1,\ldots,b-1,b\}$ for $a,b\in\Z$. Throughout the paper, generic constants $C$ and $c$ may change from line to line.

\section{Background and preliminaries}\label{sec:background}

\subsection{Classical baker's map}
For the classical $D$-baker's map \eqref{eqn:classical-baker}, we will primarily use its representation using symbolic dynamics with $D$-ary expansions of $(q,p)\in\T^2$. 
We will always assume that $q$ and $p$ are written in their standard $D$-ary expansions without any trailing $(D-1)$s, i.e. an infinitely repeating sequence of $(D-1)$s is replaced by the $D$-ary expansion terminating in a zero.
For position coordinate $q=0.q_1q_2\cdots$ and momentum coordinate $p=0.p_1p_2\cdots$ written in base $D$, the action of the classical map $B$ is the (Bernoulli) left shift
\begin{align*}
\underbrace{\cdots p_3 p_{2}p_{1}}_{\longleftarrow p}\cdot \underbrace{q_1q_2q_3\cdots}_{q \longrightarrow}
\;\mapsto\;\cdots p_3p_{2}p_{1}q_1\cdot q_2q_3\cdots.
\end{align*}
Since the $D$-baker's map is equivalent to a Bernoulli shift, it is immediately seen to be ergodic, maximally chaotic, etc.
The $D$-baker's map possess two classical symmetries: a reflection symmetry $R:(q,p)\mapsto(1-q,1-p)$, and an anticanonical time-reversal symmetry $(q,p)\mapsto(p,q)$ \cite{BV1989}.

\subsection{Quantum baker's map}\label{subsec:qbakermap}
We briefly introduce the quantum Hilbert space, baker's maps, and observables, referring to \cite{AN2007} for further details.
For the Hilbert space for the Walsh-quantized baker's maps, one considers only dimensions $N=D^k$ for $k\in\N$, so that $\mathcal H_N$ can be viewed as the tensor product $(\C^D)^{\otimes k}$.
The semiclassical parameter is $\hbar\propto N^{-1}$, and the semiclassical limit $\hbar\to0$ is the large dimension limit $N\to\infty$.

\subsubsection{Walsh-quantized baker's maps} 
For $D\ge3$, the Walsh-quantized baker's map $\Ba_k$ can be defined through its action on tensor product states,
\begin{align}\label{eqn:walsh-action}
\Ba_k(v^{(1)}\otimes\cdots\otimes v^{(k)})=v^{(2)}\otimes v^{(3)}\otimes\cdots\otimes v^{(k)}\otimes \hat{F}_D^\dagger v^{(1)},
\end{align}
where $\hat{F}_D^\dagger$ is the inverse of the discrete $D\times D$ Fourier transform matrix,
\begin{align*}
\langle m|\hat{F}_D^\dagger|n\rangle &= \frac{1}{\sqrt{D}}e^{2\pi imn/D},\quad\text{for }m,n=0,\ldots,D-1.
\end{align*}
From the definition, $\hat B_k$ is immediately seen to be unitary and to have period $4k$ if $D\ge3$, and $2k$ if $D=2$. We will let $q(k)\in\{4k,2k\}$ denote this period throughout the paper, to make it easier to treat the cases $D\ge3$ and $D=2$ together. With this notation, $\hat B_k^{q(k)}=\operatorname{Id}$. 
We discuss how $\Ba_k$ is a (Walsh) quantization of the classical baker's map once we define quantum observables in Section~\ref{subsubsec:walsh}.
For now, we just note that $\Ba_k$ can be viewed as a left shift on a string of qudits \cite{SchackCaves2000}, which makes it an intuitive quantum analogue of the classical baker's map which acts as a left shift on a string of bits (or other base strings).

\subsubsection{Coherent state bases}
The Walsh quantization of a classical observable $a\in \Lip(\T^2)$ is defined in terms of coherent state bases. 
For $\eta\in\intbrr{0:D-1}$, let $|\eta\rangle\in\C^D$ denote the standard basis vector for the $\eta$th coordinate.  For $\ell=\ell(k)\in\intbrr{0:k}$, and for momentum coordinate $\varepsilon'=\varepsilon_{\ell+1}\cdots\varepsilon_k\in\intbrr{0:D-1}^{k-\ell}$ and position coordinate $\varepsilon=\varepsilon_{\ell}\cdots\varepsilon_1\in\intbrr{0:D-1}^\ell$, the \emph{$(k,\ell)$-coherent state} $|\cs{\varepsilon}\rangle$ is defined as
\begin{align}\label{eqn:cs}
|\cs{\varepsilon}\rangle &= |\varepsilon_\ell\rangle\otimes\cdots\otimes|\varepsilon_1\rangle\otimes\hat{F}_D^\dagger|\varepsilon_k\rangle\otimes\cdots\otimes\hat{F}_D^\dagger|\varepsilon_{\ell+1}\rangle.
\end{align}
We note that this indexing for the coordinates of $|\cs{\varepsilon}\rangle$ is very non-standard, differing from the more usual indexing in e.g. \cite{AN2007}, as well as from the indexing used in \cite{wbaker}. However, because we consider long time evolution, which by \eqref{eqn:walsh-action} cyclically permutes the $\varepsilon_j$, it is especially convenient to have all the indices be in descending order (modulo $k$). 
The corresponding classical representation of the coordinates of $|\cs{\varepsilon}\rangle$ given by writing $q=0.\varepsilon_\ell\cdots\varepsilon_1$ and $p=0.\varepsilon_{\ell+1}\cdots\varepsilon_k$, is
\[
\underbrace{\varepsilon_k\cdots\varepsilon_{\ell+1}}_{\longleftarrow p}\cdot \underbrace{\varepsilon_\ell\cdots\varepsilon_1}_{q \longrightarrow},
\]
which now also has all indices in descending order, making comparison between the two maps easier.
The set $\{|\cs{\varepsilon}\rangle:\varepsilon'\in\intbrr{0:D-1}^{k-\ell},\varepsilon\in\intbrr{0:D-1}^\ell\}$ of $(k,\ell)$-coherent states forms an orthonormal basis for $(\C^D)^{\otimes k}$.

A coherent state $|\cs{\varepsilon}\rangle$ is associated
with a $(k,\ell)$ quantum rectangle $[\cs{\varepsilon}]\subseteq\T^2$, consisting of points $(q,p)\in\T^2$ with 2-sided $D$-ary representation 
\[
\cdots**\,\varepsilon_k\cdots \varepsilon_{\ell+1}\bullet\varepsilon_\ell\cdots \varepsilon_1**\cdots,
\]
where the $*$'s indicate any digit.
Since we do not take representations with trailing $(D-1)$s, we have
\begin{align}\label{eqn:rectangle}
\begin{aligned}
[\cs{\varepsilon}]=\{(q,p)\in\T^2:&\,0.\varepsilon_\ell\cdots\varepsilon_1\le q<0.\varepsilon_\ell\cdots\varepsilon_1+D^{-\ell},\\
&0.\varepsilon_{\ell+1}\cdots\varepsilon_k\le p<0.\varepsilon_{\ell+1}\cdots\varepsilon_k+D^{-(k-\ell)}\},
\end{aligned}
\end{align}
where the values $0.\varepsilon_i\ldots\varepsilon_j$ are written in base $D$.
The quantum rectangle $[\cs{\varepsilon}]$ will also serve as a classical rectangle, since it can be viewed as a $D^{-\ell}\times D^{-(k-\ell)}$ rectangle in $\T^2$.
The set of all $(k,\ell)$ rectangles will be denoted by $\mathcal R_{k,\ell}$. If $\ell=\ell(k)$ is chosen so that
\begin{align}\label{eqn:l}
\ell\to\infty,\quad k-\ell\to\infty,
\end{align}
as $k\to\infty$, then the side lengths of the $(k,\ell)$-rectangles shrink to zero. We will always assume this for a semiclassical limit, though for the proofs we will often need a slightly stronger condition such as $\min(\ell,k-\ell)\ge c\log k$ anyway. This is always met for example by the symmetric choice $\ell=\lfloor k/2\rfloor$.

\subsubsection{Walsh quantization of observables}\label{subsubsec:walsh}
For a $(k,\ell)$-coherent state basis $|\cs{\varepsilon}\rangle$, the Walsh (or Walsh-anti-Wick) quantization of the observable $a\in\Lip(\T^2)$ is defined \cite{AN2007} as
\begin{align}\label{eqn:opkl}
\Opkl(a)&=\sum_{[\cs{\varepsilon}]\in\mathcal R_{k,\ell}}|\cs{\varepsilon}\rangle\langle\cs{\varepsilon}|\fint_{[\cs{\varepsilon}]}a(\x)\,d\x,
\end{align}
where $\fint_{R}a(\x)\,d\x$ denotes the average value $\frac{1}{|R|}\int_Ra(\x)\,d\x$.
We immediately have
\begin{align}\label{eqn:trace-walsh}
\frac{1}{D^k}\operatorname{Tr}\Opkl(a) &= \int_{\T^2}a(\x)\,d\x,
\end{align}
for any $\ell\in\intbrr{0:k}$, as well as $\|\Opkl(a)\|\le\|a\|_\infty$.
The Walsh quantization of observables differs from the more usual Weyl quantization on the torus, which is derived from periodicizing the Weyl quantization on $\R^2$.
The definition \eqref{eqn:opkl} is similar in structure to the anti-Wick quantization (which with Gaussian states can be associated with the Weyl quantization in the semiclassical limit), but with ``Walsh'' coherent states instead of Gaussian coherent states, which creates the differences.
In the Walsh quantization, position and momentum are related via the Walsh (or Walsh--Fourier) transform $W_{D^k}$, defined via
\[
W_{D^k}(v^{(1)}\otimes\cdots\otimes v^{(k)})=\hat{F}_D v^{(k)}\otimes\cdots\otimes\hat{F}_Dv^{(1)},
\]
instead of the usual Fourier transform.

The Walsh-quantized baker's maps $\Ba_k$ are quantizations in the sense that for Walsh-quantized observables $\Opkl(a)$, they recover the action of the classical observable $a$ in the semiclassical limit.
It was shown in \cite{TracyScott2002} for $D=2$, that $\Ba_k$ is not a Weyl quantization of the 2-baker's map; however $\Ba_k$ can instead still be viewed as a Weyl quantization of a multi-valued, or stochastic, version of the classical 2-baker's map.

While we will not use the following representation in the proofs, we note that in the position basis (which is the coherent state basis with $\ell=k$), the map $\Ba_k$ can be written as the matrix
\begin{align}\label{eqn:wbaker}
\Ba_k = W_{D^k}^{-1}\begin{pmatrix}W_{D^{k-1}}&0&0\\0&\ddots&0\\0&0&W_{D^{k-1}}\end{pmatrix},
\end{align}
for Walsh transforms $W_{D^k}$ and $W_{D^{k-1}}$.
This is the same structure as the (Weyl-quantized) baker's map constructed in \cite{BV1989}, but with the Walsh transform matrices replacing discrete Fourier transform (DFT) matrices.
We note that while the Weyl-quantized baker's map from \cite{BV1989} involves DFT matrices, a \emph{Walsh}--Hadamard transform was used in \cite{ML2005} to study multifractal behavior of some of its eigenstates.

\subsubsection{Useful properties}
We collect several useful properties for later use in the proofs.
\begin{enumerate}[(i)]
\item The classical baker's map is exponential mixing for H\"older continuous observables (``exponential decay of correlations'') \cite{Bowen2008}.
This can be seen using the symbolic dynamics\footnote{By considering $(k,\ell)$-rectangles $[\cs{\varepsilon}]\subset\T^2$, the classical symbolic dynamics show that $m([\cs{\varepsilon}]\cap B^t[\cs{\delta}])=m([\cs{\varepsilon}])m([\cs{\delta}])$ for $|t|\ge k$ and any $(k,\ell)$-rectangles $[\cs{\varepsilon}],[\cs{\delta}]$. Taking $k=|t|$ or $|t|-1$ even and $\ell=k/2$, we can then directly estimate the mixing rate in \eqref{eqn:expmix} by approximating $a(\x)$ by $\fint_{[\cs{\varepsilon}]}a$, using that $\int_{[\cs{\varepsilon}]}b(B^t\x)\,d\x=\sum_{[\cs{\delta}]}\int_{[\cs{\delta}]\cap B^{-t}[\cs{\varepsilon}]}b(\x)\,d\x$, and approximating $b(\x)$ by $\fint_{[\cs{\delta}]}b$.}, or from more general results involving Markov partitions such as \cite{Chernov1992}. In particular, for Lipschitz observables $a,b$ and any $t\in\Z$,
\begin{align}\label{eqn:expmix}
\left|\int_{\T^2}a(\x)b(B^{t}\x)\,d\x-\int_{\T^2}a(\x)\,d\x\int_{\T^2}b(\x)\,d\x\right|&\le C\|a\|_{\Lip}\|b\|_{\Lip}e^{-\Gamma|t|},
\end{align}
for some $C,\Gamma>0$ and where $\|a\|_\Lip$ and $\|b\|_\Lip$ denote the Lipschitz norm
\begin{align}
\|f\|_\Lip := \|f\|_\infty+\sup_{\x\ne \mathbf y}\frac{|f(\x)-f(\mathbf y)|}{\|\x-\mathbf y\|_2}.
\end{align}

\item Periodicity and spectral projection matrix formula:
Recalling that $q(k)=4k$ for $D\ge3$ and $q(k)=2k$ for $D=2$, the quantized map $\Ba_k$ satisfies $\Ba_k^{q(k)}=\operatorname{Id}$. As a result, considering polynomials like $(z^{q(k)}-1)/(z-1)$ which are zero on all eigenvalues except one, we see the spectral projection matrix $P_\jj$ onto the eigenspace of $e^{2\pi i\jj/q(k)}$, $\jj=0,\ldots,q(k)-1$, is given by
\begin{align}\label{eqn:ppoly0}
P_\jj &= \frac{1}{q(k)}\left(\sum_{t=-q(k)/2}^{q(k)/2-1}e^{2\pi i \jj t/q(k)}\Ba_k^t\right).
\end{align}

\item Leading order degeneracy for each eigenspace: As shown in \cite[Cor. 8.2]{wbaker}, 
each eigenspace $E_\jj$ of $\Ba_k$ has dimension
\begin{align}\label{eqn:dim}
d_\jj=\Tr P_\jj &= \frac{N}{q(k)}(1+o(1)),
\end{align}
with $o(1)$ remainder term uniform over the eigenspace index $\jj\in\intbrr{0:q(k)-1}$.
\end{enumerate}

We also discuss quantum ergodicity and QUE here, though we do not directly use the following results for the proofs.
A quantum ergodic theorem for the Walsh-quantized baker's maps was proved in \cite{AN2007}: 
For any sequence of orthonormal eigenbases $\{\psi_{k,j}\}_{j=1}^N$ of $\Ba_k$ over $k\ge1$, there is a subset $J_k\subset\intbrr{0:D^k-1}$ such that (a) $\lim_{k\to\infty}\frac{\# J_k}{D^k}=1$, and (b) if $\ell,k-\ell\to\infty$, then for any sequence $j(k)\in J_k$ and observable $a\in\Lip(\T^2)$,
\begin{align}
\lim_{k\to\infty}\langle \psi^{k,j(k)}|\Opkl(a)|\psi^{k,j(k)}\rangle =\int_{\T^2}a(\x)\,d\x.
\end{align}
Explicitly constructed eigenstates in \cite{AN2007} show that QUE does not hold for every eigenbasis.
However, as shown in \cite[\S8.5.2]{wbaker}, if one considers random eigenbases chosen according to Haar measure in each eigenspace, then a QUE statement holds with high probability: 
there is a sequence $\epsilon_k\to0$ so that
\begin{align}\label{eqn:ind-que}
\P\left[\max_{j\in\intbrr{1:N}}\Big|\langle\varphi^{(k,j)}|\Opkl(a)|\varphi^{(k,j)}\rangle-\int_{\T^2}a(\x)\,d\x\Big|>\epsilon_k\|a\|_\infty\right]&\le C_1 D^k\exp\left[-C_2\epsilon_k^2\frac{D^k}{q(k)}\right],
\end{align}
and the decay is fast enough to extend to other Lipschitz $a$ using a dense countable subset (with respect to $\|\cdot\|_\infty$ norm) of Lipschitz observables $(a_j)_j$.
While not stated there, one is permitted to take $\epsilon_k=c N^{-1/2+\delta}$ for any $\delta>0$,
which gives the expected near optimal QUE convergence rate, up to the factor of $N^{\delta}$.
One can also extend this to off-diagonal entries using moment bounds to show 
\begin{align}\label{eqn:therm}
\max_{i,j\in\intbrr{1:N}}\left|\langle \varphi^{(k,i)}|\Opkl(a)|\varphi^{(k,j)}\rangle-\delta_{ij}\int_{\T^2}a(\x)\,d\x\right|&\le N^{-1/2+\delta},
\end{align}
with high probability as $k\to\infty$. 
As discussed in the introduction, this is a much coarser statement than the fluctuation statements in Theorems~\ref{thm:mat} and \ref{thm:eth}.
We provide details for \eqref{eqn:therm} in Appendix~\ref{sec:que}.

\subsection{Random quadratic forms}\label{subsec:weingarten}

For a real-valued classical observable $a\in \Lip(\T^2)$, let $a_0=a-\int_{\T^2}a$ be the centered observable.
Since $\varphi^{(j)}$ is chosen randomly according to Haar measure in its eigenspace, a matrix element fluctuation $F_j=\sqrt{N}\left(\langle\varphi^{(k,j)}|\Opkl(a_0)|\varphi^{(k,j)}\rangle\right)$ is a quadratic form in random variables, and we can compute its mean and variance (as well as its higher moments later in Section~\ref{subsec:moments}).

To study $F_j$ coming from the eigenspace $E_\jj$ corresponding to eigenvalue $e^{2\pi i\jj/q(k)}$,
let the eigenspace dimension be $d_\jj=\Tr P_\jj=\frac{N}{q(k)}(1+o(1))$, and let $\Lambda_\jj$ be an $N\times d_\jj$ matrix with columns which form an orthonormal basis of $E_\jj$. 
Then for a Haar random unitary matrix $U=(u^{(1)}\cdots u^{(d_\jj)})\in\mathcal{U}(d_\jj)$ with columns $u^{(1)},\ldots u^{(d_\jj)}$, the matrix $\Lambda_\jj U$ is distributed according to Haar measure on $E_\jj$. Also, the spectral projection matrix is $P_\jj=\Lambda_\jj\Lambda_\jj^\dagger$. Letting $u=u^{(1)}$, the distribution of a single $F_j$ corresponding to the eigenspace $E_\jj$ thus coincides with that of
\begin{align}\label{eqn:F-unitary}
\sqrt{N}\langle u|\Lambda_\jj^\dagger \Opkl(a_0)\Lambda_\jj|u\rangle.
\end{align}
In order to compute expectation values, we use the Weingarten calculus \cite{Collins2003,CollinsSniady2006}:
\begin{thm}[{\cite[Corollaries 2.4, 2.7]{CollinsSniady2006}}]\label{thm:weingarten}
Let $\mathcal U(d)$ be the space of $d\times d$ unitary matrices equipped with Haar measure, and let $S_n$ be the permutation group on $n$ elements. 
Let $n$ be a positive integer and $i=(i_1,\ldots,i_n)$, $i'=(i_1',\ldots,i_n')$, $j=(j_1,\ldots,j_n)$, $j'=(j_1',\ldots,j_n')$ be $n$-tuples of positive integers. Then
\begin{align}\label{eqn:weingarten}
\int_{\mathcal{U}(d)}U_{i_1j_1}\cdots U_{i_nj_n}\overline{U_{i_1'j_1'}}\cdots\overline{U_{i_n'j_n'}}\,dU
&= \sum_{\sigma,\tau\in S_n}\delta_{i_1i'_{\sigma(1)}}\cdots\delta_{i_ni'_{\sigma(n)}}\delta_{j_1j'_{\tau(1)}}\cdots\delta_{j_nj'_{\tau(n)}} \Wg(\tau\sigma^{-1}),
\end{align}
where $\Wg$ denotes the Weingarten function defined in \cite[Eq.~(9)]{CollinsSniady2006}.
The Weingarten function satisfies the asymptotics, for fixed $n$ as $d\to\infty$,
\begin{align}\label{eqn:weingarten-asymptotics}
\Wg(\operatorname{Id})&=\frac{1}{d^n}\left(1+O\left(\frac{1}{d^2}\right)\right),\quad
\Wg(\sigma)=O\left(\frac{1}{d^{n+|\sigma|}}\right),
\end{align}
where $|\sigma|$ denotes the minimal number of factors needed to write $\sigma$ as a product of transpositions, or equivalently, $|\sigma|=n-(\#\text{ of disjoint cycles})$.

If $n\ne n'$, then
\begin{align}\label{eqn:weingarten-diff}
\int_{\mathcal{U}(d)}U_{i_1j_1}\cdots U_{i_nj_n}\overline{U_{i_1'j_1'}}\cdots\overline{U_{i_{n'}'j_{n'}'}}\,dU=0.
\end{align}
\end{thm}

Let $\mz:=\Lambda_\jj^\dagger \Opkl(a_0)\Lambda_\jj$, which is self-adjoint for $a_0$ real. Note that $\Tr(\mz^n)=\Tr((\Opkl(a_0)P_\jj)^n)$ since $P_\jj=\Lambda_\jj\Lambda_\jj^\dagger$. 
We will use a subscript $\omega$ on expectations and probabilities to emphasize it is taken over Haar measure in each eigenspace.
Using Theorem~\ref{thm:weingarten}, \eqref{eqn:F-unitary}, and \eqref{eqn:dim}, we obtain
\begin{align*}
\E_\omega F_j= \sqrt{N}\sum_{i_1,i_1'=1}^{d_\jj}\mz_{i_1'i_1}\E[\bar{u}_{i_1'}u_{i_1}] 
&=\sqrt{N}\frac{\Tr(\mz)}{d_\jj}
=\frac{q(k)}{\sqrt{N}}\Tr(\Opkl(a_0)P_\jj)(1+o(1)),\numberthis\label{eqn:mean}
\end{align*}
and
\begin{align*}
\E_\omega|F_j|^2 &= N\sum_{i_1,i_2,i_1',i_2'=1}^{d_\jj}\mz_{i_1'i_1}\mz_{i_2'i_2}\E[\bar{u}_{i_1'}u_{i_1}\bar{u}_{i_2'}u_{i_2'}]\\
&=N\sum_{\tau\in S_2}\left[(\Tr \mz)^2\Wg(\tau)+(\Tr \mz^2)\Wg(\tau \circ(1\,2))\right]\\
&= \frac{q(k)^2}{N}\left[(\Tr \Opkl(a_0) P_\jj)^2 + \Tr(\Opkl(a_0) P_\jj \Opkl(a_0)P_\jj)\right](1+o(1)),\numberthis\label{eqn:variance}
\end{align*}
where we used the Weingarten function asymptotics \eqref{eqn:weingarten-asymptotics} to evaluate $\Wg(\tau)$ and $\Wg(\tau\circ(1\;2))$ for the last line.

In Sections~\ref{subsec:gaussian} and \ref{subsec:eth-conv}, we will also use the fact that one can generate Haar-distributed unitary matrices using the Gram--Schmidt procedure on independent standard complex Gaussian random vectors; see e.g.  \cite[\S1.2]{Meckes2019-book}.

\subsection{Proof outline}\label{subsec:overview}

In this section, we give an outline of the proof method.
As can be seen using the previous section on random quadratic forms, the moments (with respect to Haar measure) of the $F_j$ are given in terms of traces such as  $\Tr(\Opkl(a_0)P_\jj)$ and $\Tr(\Opkl(a_0)P_\jj\Opkl(a_0)P_\jj)$. 
Using \eqref{eqn:ppoly0}, these traces can be written in terms of the quantum propagator $\Ba_k^t$ for $t$ along a full period $-q(k)/2$ to $q(k)/2-1$. 
This means we will have to consider $\Ba_k^t$ for times $t$ near and at the Ehrenfest time $t\approx k$. 
For these times, we will have to rely on cancellations due to the quantum phases.
Away from the Ehrenfest time, we will use a precise correspondence between the classical and quantum dynamics proved in Section~\ref{subsec:ctime}. 
 
Once we estimate all the above traces to obtain the behavior of the moments $\E_\omega F_j^p$ in Theorem~\ref{thm:ind}, we can use Weingarten calculus and some estimates (since the $F_j$ are not independent) to show the desired convergences in Theorem~\ref{thm:mat}. Letting $A_0:=\Opkl(a_0)$, the sequence in summary is:
\begin{center}
\stepcounter{tikznumber} 
\begin{tikzpicture} 
\node[matrix,right] at (0,0)
{
\node {$\E F_j^m$}; & \node {$\rightsquigarrow$}; 
& \node {$\Tr((A_0P_\jj)^p)$}; & \node{$\rightsquigarrow$}; 
& \node {\parbox{2.2cm}{$\Tr(A_0\Ba_k^{t})$,\\$\Tr(A_0\Ba_k^t A_0\Ba_k^s)$}}; & \node{$\rightsquigarrow$};
& \node{\parbox{2.9cm}{\flushleft classical dynamics,\\quantum phase cancellation}}; & \node{$+$};
& \node{\parbox{1.5cm}{some \\additional estimates}}; & \node{$\rightsquigarrow$};
& \node{Theorem~\ref{thm:mat}.};
\\
};
\end{tikzpicture}
\end{center}
Similar methods involving different eigenstates and eigenspaces are used to prove Theorem~\ref{thm:eth}.

Rather than starting with the moments $\E_\omega F_j^m$ for a single $F_j$, it is easiest to first compute the averaged expectation values $\E_\omega\frac{1}{N}\sum_{j=1}^NF_j$ and $\E_\omega\frac{1}{N}\sum_{j=1}^N|F_j|^2$, which we do in Section~\ref{sec:lemmas}.
Due to the averaging over $j$, it turns out there is a great deal of cancellation of terms in the trace expansions. 
Some of the remaining terms can be directly related to the classical dynamics using Lemmas~\ref{lem:classical} and \ref{lem:intsum}.

Computing the individual moments $\E_\omega F_j^p$ turns out to be trickier, because a simple triangle inequality/absolute value bound on the trace expansion does not work near the Ehrenfest time $t\approx k$; the remainder terms, some of which previously canceled when averaging over $j$, are too large.
In order to show the remainder terms are sufficiently small for a single $j$, we must consider the previously mentioned phases of entries of $\Ba_k^t$, which do not have a classical analogue, and demonstrate that there is cancellation (Section~\ref{sec:phases}).
Applying this, we determine the asymptotics of $\Tr(\Opkl(a_0)P_\jj)$ and $\Tr(\Opkl(a_0)P_\jj \Opkl(a_0)P_\jj)$ in Section~\ref{sec:tr-ind}, followed by some additional Weingarten calculus to show convergence in probability of the quantum variance.
To show the convergence of the empirical distribution $\mu_k$, we show that the higher moments $\E \tilde F_j^p$ converge to those of a Gaussian (Theorem~\ref{thm:ind}). With some additional estimates for non-independent $F_j,F_k$, this will show the desired convergence.

\section{Time evolution and average expectation values}\label{sec:lemmas}

In this section, we start with some ``warm-up'' calculations in Proposition~\ref{prop:averages}, which use averaging to obtain explicit cancellations. 
In later sections, we will have to obtain cancellations without the averaging, which will instead be obtained through phase cancellations in the quantum map.
\begin{prop}[average expectation values]\label{prop:averages}
Let $\ell,k-\ell\ge 2\log_D k$. Then as $k\to\infty$, we have the following.
\begin{enumerate}[(i)]
\item Expected mean values: For any $D\ge2$ and $N=D^k$,
\begin{align}
\E_\omega\left[\frac{1}{N}\sum_{j=1}^NF_j\right]=0.\label{eqn:avg-mean}
\end{align}
We can also identify a difference in the $D=4$ case which is consistent with \eqref{eqn:center}: For $D\ge3$, letting $\jeven\subset\intbrr{1:N}$ denote the eigenvector indices $j$ corresponding to eigenvalues $e^{2\pi i\jj/q(k)}$ with even $\jj$, then
\begin{align}\label{eqn:jeven}
\E_\omega\Bigg[\frac{1}{|\jeven|}\sum_{j\in\jeven}F_j\Bigg]&=\begin{cases}
o(1),&D\ne 4\\
\langle a_0\rangle_{\mathcal B^{(4)}_{0,2}}+o(1),&D=4
\end{cases},
\end{align}
where $\langle a_0\rangle_{\mathcal B^{(4)}_{0,2}}$ is defined as in \eqref{eqn:a0-avg}. Similarly, with $\jodd\subset\intbrr{1:N}$ denoting the indices corresponding to odd $\jj$,
\begin{align}
\E_\omega\Bigg[\frac{1}{|\jodd|}\sum_{j\in\jodd}F_j\Bigg]&=\begin{cases}
o(1),&D\ne 4\\
-\langle a_0\rangle_{\mathcal B^{(4)}_{0,2}}+o(1),&D=4
\end{cases}.
\end{align}

\item Expected variance: For any $D\ge2$,
\begin{multline}
\E_\omega\left[\frac{1}{N}\sum_{j=1}^N|F_j|^2\right] =\sum_{t=-\infty}^\infty \left(\int_{\T^2}a_0(\x)a_0(B^t\x)\,d\x+\mathbf{1}_{D\ge3}\int_{\T^2}a_0(\x)a_0(B^tR\x)\,d\x \right)\\
+\frac{1}{N}\sum_{\substack{t=-q(k)/2\\t\ne0}}^{q(k)/2-1}|\Tr(\Opkl(a_0)\hat B_k^t)|^2+o(1).\label{eqn:avg-var}
\end{multline}
\end{enumerate}
\end{prop}

It will turn out, due to later results in Sections~\ref{sec:phases} and \ref{sec:tr-ind}, that the term in the second line of \eqref{eqn:avg-var} is $o(1)$ for $D\ne4$, and is $\dfoura^2+o(1)$ for $D=4$.
Keeping this in mind, \eqref{eqn:avg-var} suggests that $V(a)+\oneb_{D=4}\dfoura^2$, where $V(a)$ appears in the first line of \eqref{eqn:avg-var}, is indeed the correct variance for the matrix element fluctuations.
However, to prove the evaluation of the above mentioned term, as well as the convergence of the quantum variance to $V(a)+\oneb_{D=4}\dfoura^2$ as stated in Theorem~\ref{thm:mat}, we will need to determine the individual matrix element means and variances, that is $\E_\omega F_j$ and $\E_\omega F_j^2$.
As discussed previously this is more complicated than the expectations computed above, since we will need a better understanding of the phases of the quantum propagator $\Ba_k^t$.

We note that, strictly speaking, we do not actually need the result of Proposition~\ref{prop:averages}, since it will be implied by the individual estimates on $\E_\omega F_j$ and $\E_\omega F_j^2$ that we prove later. However, the proof is useful for demonstrating several computational techniques and for motivating how the classical quantity $V(a)$ arises from the quantum dynamics.
Additionally, we will need Lemmas~\ref{lem:classical} and \ref{lem:intsum} proved below, which relate quantum and classical time evolution, to analyze terms in individual $\E_\omega F_j^2$ as well.
The terms that produce the leading order terms in Proposition~\ref{prop:averages} are the same ones that produce the end result in $\E_\omega F_j$ and $\E_\omega F_j^2$; the challenge for the latter expectations is to show that all the other terms (which magnitude-wise can be a similar order) are negligible after cancellations.

Before starting to analyze the time evolution of $\Ba_k^t$, which will be used for the expected variance in \eqref{eqn:avg-var}, we can prove part (i) of the proposition using just properties of the spectral projectors $P_\jj$ and the definition of $\Ba_k^t$. 
\begin{proof}[Proof of Proposition~\ref{prop:averages}(i)]
Recall from Section~\ref{subsec:weingarten} that
\begin{align}
\E_\omega F_j&=\frac{q(k)}{\sqrt{N}}\Tr(\Opkl(a_0)P_\jj)(1+o(1)),
\end{align}
where $P_\jj$ is the projection onto the eigenspace of $e^{2\pi i\jj/q(k)}$.
Using the eigenspace dimensionality \eqref{eqn:dim} and the fact that $\sum_{\jj=0}^{q(N)-1}P_\jj=\operatorname{Id}$, we see
\begin{align*}
\E_\omega\left[\frac{1}{N}\sum_{j=1}^NF_j\right]&=\frac{1}{N}\frac{N}{q(k)}\sum_{\jj=0}^{q(k)-1}\frac{q(k)}{\sqrt{N}}{\Tr(\Opkl(a_0)P_\jj)}(1+o(1))\\
&=\frac{1+o(1)}{\sqrt{N}}\Tr(\Opkl(a_0))=0,\numberthis\label{eqn:exp-average}
\end{align*}
proving \eqref{eqn:avg-mean}.

We next consider the averages over every other eigenspace for $D\ge3$.
Applying \eqref{eqn:dim} and \eqref{eqn:ppoly0}, we obtain
\begin{align*}
\E_\omega\Bigg[\frac{1}{|\jeven|}\sum_{j\in\jeven}F_j\Bigg]&=\frac{2}{N}\frac{N}{q(k)}(1+o(1))\sum_{\substack{\jj=0\\\jj\text{ even}}}^{q(k)-1}\frac{q(k)}{\sqrt{N}}\Tr(\Opkl(a_0)P_\jj)\\
&=\frac{2}{\sqrt{N}}(1+o(1))\frac{1}{q(k)}\sum_{t=-q(k)/2}^{q(k)/2-1}\sum_{\substack{\jj=0\\\jj\text{ even}}}^{q(k)-1}\Tr(\Opkl(a_0)\Ba_k^t)e^{2\pi i\jj t/q(k)}.
\end{align*}
The only term dependent on $\jj$ is the exponential phase term, leading to the sum
\begin{align*}
\sum_{\substack{\jj=0\\\jj\text{ even}}}^{q(k)-1}e^{2\pi i\jj t/q(k)}&=\sum_{\gamma=0}^{q(k)/2-1}e^{4\pi i\gamma t/q(k)}=\begin{cases}
\frac{q(k)}{2},&t=0\text{ or }-q(k)/2\\
0,&\text{otherwise}
\end{cases}.
\end{align*}
Thus
\begin{align}\label{eqn:jeven-2trace}
\E_\omega\Bigg[\frac{1}{|\jeven|}\sum_{j\in\jeven}F_j\Bigg]&=\frac{1+o(1)}{\sqrt{N}}\left[\Tr(\Opkl(a_0))+\Tr(\Opkl(a_0)\Ba_k^{-2k})\right].
\end{align}
As before $\Tr(\Opkl(a_0))=0$, while using the definition \eqref{eqn:walsh-action} gives
$\Ba_k^{-2k}=\big[(\hat F_D)^2\big]^{\otimes k}$. 
Defining $\hat{R}_D:=(\hat{F}_D^\dagger)^2=(\hat{F}_D)^2$, one can check by direct calculation that $\hat R_D$ is the map $|x\rangle\mapsto|-x\;\mathrm{mod}\;D\rangle$. Then expanding the trace $\Tr(\Opkl(a_0)\Ba_k^{-2k})$ in the $(k,\ell)$-coherent state basis gives
\begin{align*}
\Tr(\Opkl(a_0)\Ba_k^{-2k}) = \sum_{[\cs{\varepsilon}]\in\mathcal R_{k,\ell}}\langle \cs{\varepsilon}|\hat R_D^{\otimes k}|\cs{\varepsilon}\rangle\fint_{[\cs{\varepsilon}]}a_0
&=\sum_{\substack{[\cs{\varepsilon}]\in\mathcal R_{k,\ell}\\\varepsilon_i=-\varepsilon_i\;\mathrm{mod}\;D}}\fint_{[\cs{\varepsilon}]}a_0,\numberthis\label{eqn:tracehalf}
\end{align*}
with the last equality since $\langle\varepsilon_i|\hat R_D|\varepsilon_i\rangle=\langle\varepsilon_i|-\varepsilon_i\;\mathrm{mod}\;D\rangle$. 
For $D$ odd, the sum is thus only over the single term $[\cs{\varepsilon}]$ with $\varepsilon_1=\cdots=\varepsilon_k=0$, and \eqref{eqn:tracehalf} is $\le\|a_0\|_\infty=o(\sqrt{N})$. 
For $D$ even, the sum is over $\varepsilon_i\in\{0,D/2\}$, so there are $2^k$ terms in the sum, and \eqref{eqn:tracehalf} is $\le 2^k\|a_0\|_\infty$. For $D\ge6$ even, this is also $o(\sqrt{N})$.

However, for $D=4$, $2^k=\sqrt{N}$ and so \eqref{eqn:tracehalf} can be $\ge c\sqrt{N}$ for certain $a_0$. In summary we thus have
\begin{align}\label{eqn:trace2k}
\frac{1}{\sqrt{N}}\Tr(\Opkl(a_0)\Ba_k^{-2k})&=\begin{cases}
o(1),&D\ne4\\
\displaystyle\frac{1}{2^k}\sum_{\substack{[\cs{\varepsilon}]\in\mathcal R_{k,\ell}\\\varepsilon_j\in\{0,2\}}}\fint_{[\cs{\varepsilon}]}a_0,&D=4
\end{cases}.
\end{align}
The sequence $(s_k)_k$ defined by
\begin{align*}
s_k:=\frac{1}{2^k}\sum_{\substack{[\cs{\varepsilon}]\in\mathcal R_{k,\ell(k)}\\\varepsilon_j\in\{0,2\}}}\fint_{[\cs{\varepsilon}]}a_0
\end{align*}
converges as $k\to\infty$: For $\ell(k),k-\ell(k)\ge v$ and $(k,\ell)$-symbolic expression $\cs{\varepsilon}=\varepsilon_k\cdots\varepsilon_{\ell+1}\bullet\varepsilon_{\ell}\cdots\varepsilon_1$, let $(\cs{\varepsilon})_v:=\varepsilon_{2v}\cdots\varepsilon_{v+1}\bullet \varepsilon_v\cdots\varepsilon_1$ be its truncation to $v$ $D$-ary places on each side. Since $a_0$ is Lipschitz continuous, then for any $k$ such that $\ell(k),k-\ell(k)\ge v$,
\begin{align*}
s_k&=\frac{1}{2^k}\sum_{\substack{[\cs{\varepsilon}]\in\mathcal R_{k,\ell(k)}\\\varepsilon_j\in\{0,2\}}}\left[a_0((\cs{\varepsilon})_v)+O(\|a\|_\Lip D^{-v})\right]\\
&=\frac{1}{2^k}\sum_{\substack{[\cs{\tilde\varepsilon}]\in\mathcal R_{2v,v}\\\tilde\varepsilon_j\in\{0,2\}}}2^{\ell-v}2^{k-\ell-v}\left[a_0(\cs{\tilde\varepsilon})+O(\|a\|_\Lip D^{-v})\right]\\
&=\frac{1}{2^{2v}}\sum_{\substack{[\cs{\tilde\varepsilon}]\in\mathcal R_{2v,v}\\\tilde\varepsilon_j\in\{0,2\}}}a_0(\cs{\tilde\varepsilon})+O(\|a\|_\Lip D^{-v}).
\end{align*}
Since $\ell(k),k-\ell(k)\to\infty$ as $k\to\infty$, taking $v\to\infty$ implies $(s_k)_k$ is a Cauchy sequence, so
\begin{align}
\lim_{k\to\infty}s_k=\lim_{k\to\infty}\frac{1}{2^k}\sum_{\substack{[\cs{\varepsilon}]\in\mathcal R_{k,\ell(k)}\\\varepsilon_j\in\{0,2\}}}\fint_{[\cs{\varepsilon}]}a_0
\end{align}
exists and equals $\langle a_0\rangle_{\mathcal B^{(4)}_{0,2}}$ as defined in \eqref{eqn:a0-avg}. Thus \eqref{eqn:trace2k} becomes for $D\ge3$,
\begin{align}\label{eqn:trace2k-final}
\frac{1}{\sqrt{N}}\Tr(\Opkl(a_0)\Ba_k^{-2k})&=\begin{cases}
o(1),&D\ne4\\
\langle a_0\rangle_{\mathcal B^{(4)}_{0,2}}+o(1),&D=4
\end{cases}.
\end{align}
With \eqref{eqn:jeven-2trace}, this implies \eqref{eqn:jeven}. The case with $\jodd$ follows similarly.
\end{proof}

\subsection{Quantum vs classical time evolution}\label{subsec:ctime}
The powers $\Ba_k^t$ for $t=1,2,\ldots$ describe the quantum time evolution.
A key property of the $\Ba_k^t$ is that the pattern of nonzero matrix elements in a coherent state basis are directly related to the classical $D$-baker's map. They mimic the behavior of time evolution of the classical baker's map before an Ehrenfest time $t\approx k$, at which they appear to ``decohere''. Afterwards, instead of further decohering, they reverse and continue to imitate the classical map evolution with some twists.
This is clearest to see visually in the position basis, where for $t=1,2,\ldots,k-1$, the nonzero matrix elements of $\Ba_k^t$ resemble the expanding map $q\mapsto D^tq\;\mathrm{mod}\;1$, which is the action of $t$ compositions of the classical $D$-baker's map on the position coordinate. (See Figure~\ref{fig:matrix-powers-pos}.)
\begin{figure}[ht]
\includegraphics[width=\textwidth]{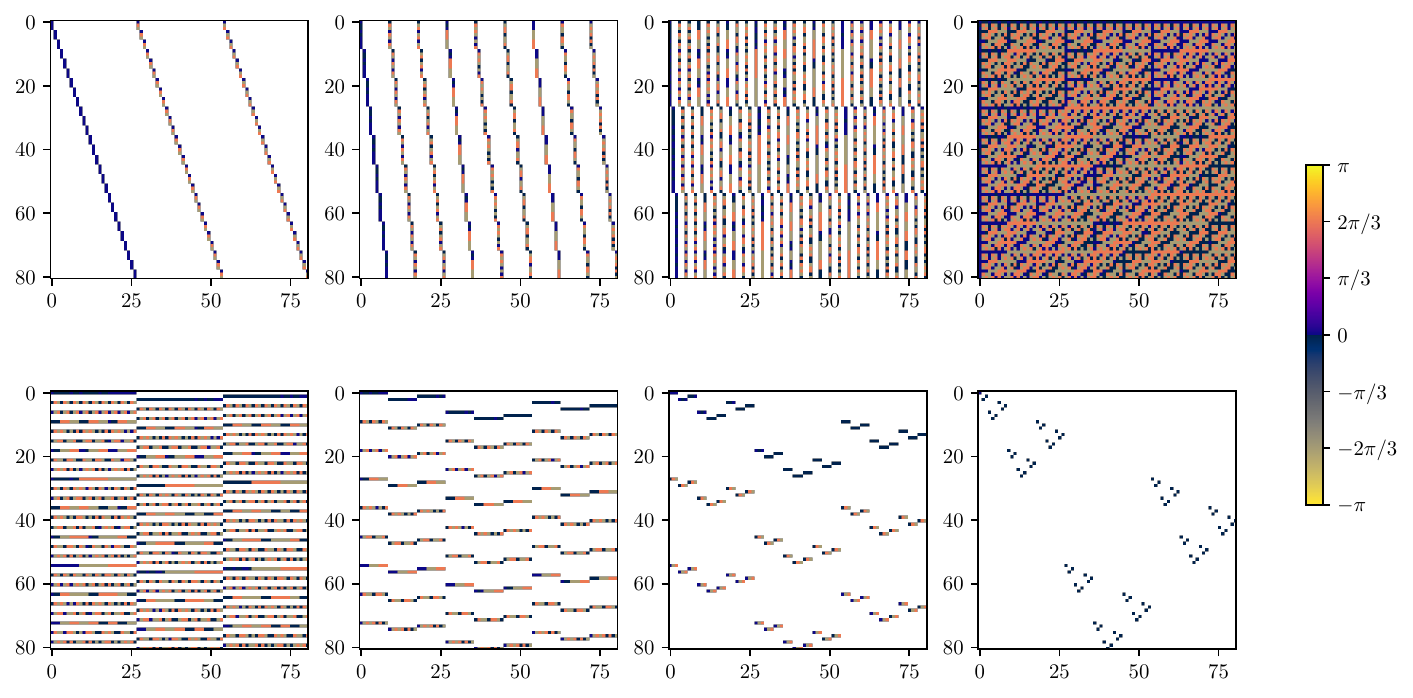}
\caption{Graphical plot of the matrix entries of $\Ba_k^t$ for $D=3$ and $k=4$ in the position basis according to \eqref{eqn:wbaker}, for $t=1,\ldots,8$ from left to right, top to bottom. The phases are plotted in color, with matrix entries that are zero shown in white. 
Noting that the $(0,0)$ entry of the matrix $\Ba_k^t$ is the top left, we see for early $t$ that the nonzero matrix entries of $\Ba_k^t$ resemble the graph of the classical map $q\mapsto D^tq\;\mathrm{mod}\;1$, which is the $D$-baker's map action on the position coordinate.
At $t$ approaches the Ehrenfest time $k=4$ (top right), the correspondence breaks down as the matrix $\Ba_k^t$ becomes fully dense. Afterwards, however, $\Ba_k^t$ ``reverses'' form and becomes sparse again.
}\label{fig:matrix-powers-pos}
\end{figure}

We first recall the following result from \cite{wbaker}, which counts the number of nonzero entries of $\hat B_k^t$.
\begin{prop}[{\cite[Prop. 8.4]{wbaker}}]\label{prop:walsh-powers}
Define the tent-shaped function $\eta(t)=\eta_k(t)$ periodically in $t\in\Z/(2k)\Z$ via 
\begin{align}\label{eqn:eta}
\eta_k(t)&=\begin{cases}
t,&0\le t\le k\\
k-(t-k),&k+1\le t\le 2k-1
\end{cases}.
\end{align}
Then for $D\ge2$, $t\in\Z$, and any $0\le \ell\le k$,
\begin{enumerate}[(i)]
\item For $t\not\in (2k)\Z$, there are exactly $D^{\eta_k(t)}$ $(k,\ell)$-coherent state basis vectors $|\varepsilon'\cdot\varepsilon\rangle$ such that $\langle \varepsilon'\cdot\varepsilon|\hat B_k^t|\varepsilon'\cdot\varepsilon\rangle\ne0$. 
For $t\in (4k)\Z$, there are $D^k$ such solutions.
For $t\in 2k+(4k)\Z$, there is $D^{\eta_k(t)}=1$ such $|\varepsilon'\cdot\varepsilon\rangle$ if $D$ is odd, and $2^k$ such $|\varepsilon'\cdot\varepsilon\rangle$ if $D$ is even. 
\item If  $\langle \varepsilon'\cdot\varepsilon|\hat B_k^t|\varepsilon'\cdot\varepsilon\rangle\ne0$, then it has absolute value $|\langle \varepsilon'\cdot\varepsilon|\hat B_k^t|\varepsilon'\cdot\varepsilon\rangle|=D^{-\eta_k(t)/2}$.
\item There are $D^k\cdot D^{\eta_k(t)}$ non-zero entries $\langle \delta'\cdot\delta|\hat B_k^t|\varepsilon'\cdot\varepsilon\rangle$, and for these entries, $|\langle\delta'\cdot\delta|\hat B_k^t|\varepsilon'\cdot\varepsilon\rangle|=D^{-\eta_k(t)/2}$.
\end{enumerate}
\end{prop}

For this paper, we will need to know how the specific coherent states $|\cs{\varepsilon}\rangle,|\cs{\delta}\rangle$ where $\langle\cs{\varepsilon}|\hat B_k^t|\cs{\delta}\rangle\ne0$ are related to the classical baker's map evolution. For this we have the following more precise lemma:
\begin{lem}[time evolution]\label{lem:classical}
Let $D\ge3$, $k\in\N$, and $\ell\in\intbrr{0:k}$, and consider $(k,\ell)$-coherent states $|\cs{\varepsilon}\rangle,|\cs{\delta}\rangle$. 
For $t\in\Z$, define $[t]_k:=t\;\mathrm{mod}\,k\in\intbrr{0:k-1}$. For $B$ the classical $D$-baker's map and $R$ the map $(q,p)\mapsto(1-q,1-p)$, define the classical map $H(t)$, periodically in $\Z/(4k\Z)$, via
\begin{align}\label{eqn:H4}
H(t):=\begin{cases}
B^t,&0\le t\le k-1\\
B^{-(k-[t]_k)}R,&k\le t\le 2k-1\\
B^{[t]_k}R,&2k\le t\le 3k-1\\
B^{-(k-[t]_k)},&3k\le t\le 4k-1
\end{cases}.
\end{align}
Then
\begin{align}\label{eqn:H3}
\{(|\cs{\varepsilon}\rangle,|\cs{\delta}\rangle):\langle \cs{\delta}|\Ba_k^t|\cs{\varepsilon}\rangle\ne0\}
=\{([\cs{\varepsilon}],[\cs{\delta}]):(H(t)[\cs{\varepsilon}])\cap[\cs{\delta}]\ne\emptyset\}.
\end{align}

For $D=2$ and $t\in\intbrr{0:2k-1}$, the same holds with
\begin{align}\label{eqn:H2}
H(t):=\begin{cases}B^t,&0\le t\le k-1\\
B^{-(k-[t]_k)},&k\le t\le 2k-1
\end{cases}.
\end{align}
\end{lem}

\begin{proof}
Recall the action of $\Ba_k$ on tensor product states,
\begin{align*}
\Ba_k(v^{(1)}\otimes\cdots\otimes v^{(k)})&=v^{(2)}\otimes v^{(3)}\otimes\cdots\otimes v^{(k)}\otimes\hat{F}_D^\dagger v^{(1)},
\end{align*}
and that a $(k,\ell)$-coherent state $|\cs{\varepsilon}\rangle$ is defined as the tensor product state
\begin{align*}
|\cs{\varepsilon}\rangle &= |\varepsilon_\ell\rangle\otimes\cdots\otimes|\varepsilon_1\rangle\otimes\hat{F}_D^\dagger|\varepsilon_k\rangle\otimes\cdots\otimes\hat{F}_D^\dagger|\varepsilon_{\ell+1}\rangle,
\end{align*}
corresponding to the quantum rectangle $[\cs{\varepsilon}]$ of points $(q,p)\in\T^2$ with 2-sided $D$-ary representation $\leftarrow p\bullet q\rightarrow=\cdots**\,\varepsilon_k\cdots \varepsilon_{\ell+1}\bullet\varepsilon_\ell\cdots \varepsilon_1**\cdots$, with no trailing $(D-1)$s.

We first work with $D\ge3$, and consider different cases depending on the value of $t$.
\begin{itemize}[leftmargin=*]
\item For $0\le t<\ell$, we have $\Ba_k^t|\cs{\varepsilon}\rangle=|\varepsilon_{\ell-t}\rangle\otimes\cdots\otimes|\varepsilon_1\rangle\otimes\hat{F}_D^\dagger|\varepsilon_k\rangle\otimes\cdots\otimes\hat{F}_D^\dagger|\varepsilon_{\ell-t+1}\rangle$. To compare easily with $|\cs{\delta}\rangle$, we can write this in table form as 
\begin{align}\label{eqn:t1}
\def\arraystretch{1.3}
\begin{array}{c||ccc|ccc|ccc}
 & 1&\cdots & \ell-t &\ell-t+1 &\cdots & \ell&\ell+1 &\cdots &k\\\hline
|\cs{\delta}\rangle&|{\delta_\ell}\rangle & \cdots &|{\delta_{t+1}}\rangle &|\delta_{t}\rangle&\cdots& |{\delta_1}\rangle&\hat{F}_D^\dagger |{\delta_k}\rangle&\cdots&\hat{F}_D^\dagger |{\delta_{\ell+1}}\rangle\\\hline
\Ba_k^t|\cs{\varepsilon}\rangle&|{\varepsilon_{\ell-t}}\rangle&\cdots  &|{\varepsilon_1}\rangle & \hat{F}_D^\dagger |{\varepsilon_k}\rangle &\cdots&\hat{F}_D^\dagger |{\varepsilon_{k-t+1}}\rangle &\hat{F}_D^\dagger |{\varepsilon_{k-t}}\rangle&\cdots&\hat{F}_D^\dagger |{\varepsilon_{\ell-t+1}}\rangle
\end{array}
\end{align}
Since $|\langle\delta_i |\hat{F}_D^\dagger|\varepsilon_j\rangle|=D^{-1/2}\ne0$ for any coordinate basis vectors $|\delta_i\rangle,|\varepsilon_j\rangle$, we see the inner product $\langle\cs{\delta}|\Ba_k^t|\cs{\varepsilon}\rangle$ is nonzero exactly when
\begin{align}
\begin{aligned}\label{eqn:matelem1}
&\delta_\ell\cdots\delta_{t+1}=\varepsilon_{\ell-t}\cdots\varepsilon_1,\\
\text{and }\;&\delta_k\cdots\delta_{\ell+1}=\varepsilon_{k-t}\cdots\varepsilon_{\ell-t+1},
\end{aligned}
\end{align}
corresponding to the inner products of coordinate basis vectors in the first and third sections of \eqref{eqn:t1}.
The other coordinates of $|\cs{\varepsilon}\rangle$ and $|\cs{\delta}\rangle$ are free to vary. In terms of the two-sided $D$-ary expansions of points in  $[\cs{\varepsilon}]$ and $[\cs{\delta}]$, we line up the expansions as
\begin{equation}\label{eqn:classical1}
\tp{\begin{tikzpicture}[baseline=(current  bounding  box.center)]
\def\pw{1pt}
\node at (0,0) {\parbox{\pw}{$\bullet$\\$\bullet$}};
\node at (.25,0) {\parbox{\pw}{$\delta_\ell$\\$\varepsilon_{\ell-t}$}};
\node at (1,0) {\parbox{\pw}{$\cdots$\\$\cdots$}};
\node at (1.75,0) {\parbox{\pw}{$\delta_{t+1}$\\$\varepsilon_1$}};
\node at (2.5,0) {\parbox{\pw}{$\cdots\delta_1$\\}};
\node at (-1,0) {\parbox{\pw}{$\delta_{\ell+1}$\\$\varepsilon_{\ell-t+1}$}};
\node at (-1.5,0) {\parbox{\pw}{$\cdots$\\$\cdots$}};
\node at (-2.25,0) {\parbox{\pw}{$\delta_k$\\$\varepsilon_{k-t}$}};
\node at (-2.85,-.24) {$\varepsilon_k\cdots$};
\draw[xshift=-.1cm,dotted,line width=1pt] (-2.3,-.5)--(-2.3,.5)--(2.5,.5)--(2.5,-.5)--cycle;
\node at (4,-.1) {,};
\node at (-4,-.24) {$\cdots*\,*$};
\node at (-3,.2) {$\cdots*\,*$};
\node at (3.8,.2) {$*\,*\cdots$};
\node at (3,-.24) {$*\,*\cdots$};
\end{tikzpicture}}
\end{equation}
so that \eqref{eqn:matelem1} corresponds to each coordinate in the outlined rectangle being the same as the one above/below it. The bottom row is left shifted by $t$, which is the classical action $B^t$ on $[\cs{\varepsilon}]$. The set $B^t[\cs{\varepsilon}]$ intersects $[\cs{\delta}]$ exactly when there is a $(q,p)\in\T^2$ so that the two rows of \eqref{eqn:classical1} represent the same point, which occurs exactly when \eqref{eqn:matelem1} holds. (Recall we always take the $D$-ary expansion to have no trailing $(D-1)$s, so every number has a unique $D$-ary expansion.)

\item For $\ell\le t\le k-1$, the equivalent of \eqref{eqn:t1} is the table
\begin{align}\label{eqn:t2}
\def\arraystretch{1.3}
\begin{array}{c||ccc|ccc|ccc}
 & 1&\cdots & \ell &\ell+1 &\cdots & k-t+\ell&k-t+\ell+1 &\cdots &k\\\hline
|\cs{\delta}\rangle&|{\delta_\ell}\rangle & \cdots &|{\delta_{1}}\rangle &\hat{F}_D^\dagger|\delta_{k}\rangle&\cdots& \hat{F}_D^\dagger|{\delta_{t+1}}\rangle&\hat{F}_D^\dagger |{\delta_t}\rangle&\cdots&\hat{F}_D^\dagger |{\delta_{\ell+1}}\rangle\\\hline
\Ba_k^t|\cs{\varepsilon}\rangle&\hat{F}_D^\dagger|{\varepsilon_{k-t+\ell}}\rangle&\cdots  &\hat{F}_D^\dagger|{\varepsilon_{k-t+1}}\rangle & \hat{F}_D^\dagger |{\varepsilon_{k-t}}\rangle &\cdots&\hat{F}_D^\dagger |{\varepsilon_{1}}\rangle &\hat{R}_D |{\varepsilon_{k}}\rangle&\cdots&\hat{R}_D |{\varepsilon_{k-t+\ell+1}}\rangle
\end{array}
\end{align}
where $\hat{R}_D:=(\hat{F}_D^\dagger)^2$.
By the same reasoning as before, the inner product 
$\langle\cs{\delta}|\Ba_k^t|\cs{\varepsilon}\rangle$ is nonzero exactly when
\begin{align}
\begin{aligned}\label{eqn:matelem2}
\delta_k\delta_{k-1}\cdots\delta_{t+1}&=\varepsilon_{k-t}\varepsilon_{k-t-1}\cdots\varepsilon_1,
\end{aligned}
\end{align}
coming from the requirements from the middle section of \eqref{eqn:t2}.
This is the exact same condition as \eqref{eqn:matelem1}, which again corresponds to the classical map $B^t$ in the same way.

\item For $k\le t\le 2k-1$, we write $t=k+[t]_k$ and note that
\begin{align}\label{eqn:timek}
\Ba_k^k(v^{(1)}\otimes\cdots\otimes v^{(k)}) &= \hat{F}_D^\dagger v^{(1)}\otimes\cdots\otimes \hat{F}_D^\dagger v^{(k)}.
\end{align}
Thus we can use the same tables \eqref{eqn:t1} and \eqref{eqn:t2} with $[t]_k$ replacing $t$, and with an additional Fourier matrix $\hat{F}_D^\dagger$ applied to every term in the bottom row of \eqref{eqn:t1} and \eqref{eqn:t2}.
Note that by direct computation, one can check that $\hat{R}_D=(\hat{F}_D^\dagger)^2$ is the map $|x\rangle\mapsto|-x\;\mathrm{mod}\;D\rangle$. Then $\langle\delta_i|\hat{R}_D|\varepsilon_j\rangle\ne0$ iff $\delta_i=\bar\varepsilon_j$, where $\bar\varepsilon_j:=D-\varepsilon_j\;(\mathrm{mod}\;D)$.

Considering the case $[t]_k<\ell$ using \eqref{eqn:t1} and the case $[t]_k\ge\ell$ using \eqref{eqn:t2}, 
we get the same condition for both,
that $\langle \cs{\delta}|\Ba_k^t|\cs{\varepsilon}\rangle\ne0$
exactly when
\begin{align}
\delta_{[t]_k}\cdots\delta_1&=\bar\varepsilon_k\cdots\bar\varepsilon_{k-[t]_k+1}.
\end{align}
The analogue of \eqref{eqn:classical1} for $[t]_k<\ell$ is
\begin{equation}\label{eqn:classical3}
\tp{\begin{tikzpicture}[baseline=(current  bounding  box.center)]
\def\pw{1pt}
\node at (0,0) {\parbox{\pw}{$\bullet$\\$\bullet$}};
\node at (.25,0) {\parbox{\pw}{$\delta_\ell$\\$*$}};
\node at (.75,0) {\parbox{\pw}{$\cdots$\\$\cdots$}};
\node at (1.25,0) {\parbox{\pw}{$\delta_{[t]_k+1}$\\$*$}};
\node at (2.25,0) {\parbox{\pw}{$\delta_{[t]_k}$\\$\bar\varepsilon_k$}};
\node at (2.75,0) {\parbox{\pw}{$\cdots$\\$\cdots$}};
\node at (3.25,0) {\parbox{\pw}{$\delta_1$\\$\bar\varepsilon_{k-[t]_k+1}$}};
\node at (5,-.24) {$\cdots\bar\varepsilon_1$};
\node at (-.75,0) {\parbox{\pw}{$\delta_{\ell+1}$\\$*$}};
\node at (-1.25,0) {\parbox{\pw}{$\cdots$\\$\cdots$}};
\node at (-1.5,0) {\parbox{\pw}{$\delta_k$\\$*$}};
\draw[xshift=4.45cm,dotted,line width=1pt] (-2.3,-.5)--(-2.3,.5)--(0.15,.5)--(.15,-.5)--cycle;
\node at (6.75,-.1) {,};
\node at (-2.25,-.05) {\parbox{1cm}{$\cdots*\,*$\\$\cdots*\,*$}};
\node at (5.15,.2) {$*\,*\cdots$};
\node at (6,-.24) {$*\,*\cdots$};
\end{tikzpicture}}
\end{equation}
and a similar figure can be drawn for the case $[t]_k\ge\ell$. 
In either case, the bottom row has each digit $\varepsilon_j$ mapped to $\bar\varepsilon_j$, and then shifted to the right by $k-[t]_k$.
The right shift is generated by the inverse map $B^{-1}$, and so we see \eqref{eqn:classical3}, with entries in the dotted box equal to the one above/below, describes exactly when there is a
$(q,p)\in(B^{-(k-[t]_k)}R[\cs{\varepsilon}])\cap[\cs{\delta}]$, where
$R$ is the map $(q,p)\mapsto(1-q,1-p)$.

\item For $2k\le t\le 3k-1$, we note that
\begin{align}
\Ba_k^{2k}(v^{(1)}\otimes\cdots\otimes v^{(k)}) &= \hat{R}_D v^{(1)}\otimes\cdots\otimes \hat{R}_D v^{(k)},
\end{align}
so we may use tables \eqref{eqn:t1} and \eqref{eqn:t2} but with an additional factor of $\hat{R}_D$ applied to each entry in the bottom row, and with $t\mapsto[t]_k$. Similar reasoning as before gives the condition
\begin{align}
\delta_k\cdots\delta_{[t]_k+1}=\bar\varepsilon_{k-[t]_k}\cdots\bar\varepsilon_1,
\end{align}
which corresponds to the map $R$ followed by the left shift by $[t]_k$, or $B^{[t]_k}R$.

\item For $3k\le t\le 4k-1$, we use
\begin{align}
\Ba_k^{3k}(v^{(1)}\otimes\cdots\otimes v^{(k)}) &= \hat{F}_D v^{(1)}\otimes\cdots\otimes \hat{F}_D v^{(k)},
\end{align}
so tables \eqref{eqn:t1} and \eqref{eqn:t2} will acquire an additional factor of $\hat{F}_D$ on each entry in the bottom row.
Similar reasoning as before gives the condition
\begin{align}
\delta_{[t]_k}\cdots\delta_1&=\varepsilon_k\cdots\varepsilon_{k-[t]_k+1},
\end{align}
which corresponds to the right shift by $k-[t]_k$, or $B^{-(k-[t]_k)}$.
\end{itemize}

For $D=2$, the analysis is essentially the same but simpler, since 
$\hat{F}_2^\dagger=\hat{F}_2$ and so $\hat{R}_D=\operatorname{Id}$. (Recall the period of $\Ba_k$ is only $2k$.)
\end{proof}

Using the above, we can show
\begin{lem}\label{lem:intsum}
Recall the definition of the tent-shaped function $\eta(t)=\eta_k(t)$ in \eqref{eqn:eta},
which gives the power of $B$ or $B^{-1}$ in $H(t)$ from \eqref{eqn:H4} and \eqref{eqn:H2}.
Then for a Lipschitz observable $a:\T^2\to\R$ and any $t\in\Z$,
\begin{align}\label{eqn:walsh-int}
\frac{1}{D^{\eta(t)}}\sum_{\substack{[\cs{\varepsilon}],[\cs{\delta}]\in\mathcal R_{k,\ell}:\\\langle\cs{\varepsilon}|\Ba_k^t|\cs{\delta}\rangle\ne0}}\fint_{[\cs{\varepsilon}]}a\fint_{[\cs{\delta}]}a
&= N\left[\int_{\T^2}a(\x)a(H(t)\x)\,d\x+O(\|a\|_\infty\|a\|_\Lip D^{-\min(\ell,k-\ell)})\right],
\end{align}
where $H(t)$ is defined in \eqref{eqn:H4} for $D\ge3$, or in \eqref{eqn:H2} for $D=2$.
As a consequence, if $\ell,k-\ell\ge 2\log_D k\to\infty$, where $\Gamma$ is the constant in \eqref{eqn:expmix}, then for $a\in \operatorname{Lip}(\T^2)$, 
\begin{multline}\label{eqn:inf-int}
\sum_{t=-q(k)/2}^{q(k)/2-1}\Tr(\Opkl(a)\hat B_k^t\Opkl(a)\hat B_k^{-t})= \sum_{t=-q(k)/2}^{q(k)/2-1}\sum_{\substack{[\cs{\varepsilon}],[\cs{\delta}]\in\mathcal R_{k,\ell}}}|\langle\cs{\delta}|\Ba_k^t|\cs{\varepsilon}\rangle|^2\fint_{[\cs{\varepsilon}]}a\fint_{[\cs{\delta}]}a\\
=N\Bigg[\sum_{t=-\infty}^\infty \left(\int_{\T^2}a(\x)a(B^t\x)\,d\x+\mathbf{1}_{D\ge3}\int_{\T^2}a(\x)a(B^tR\x)\,d\x\right)\\
+O\left(\|a\|_{\mathrm{Lip}}^2 k\left(D^{-\min(\ell,k-\ell)}+e^{-\Gamma k}\right)\right)\Bigg],
\end{multline}
with the $O(\cdots)$ remainder term is $o(1)$ as $k\to\infty$, and where $R$ is the map $(q,p)\mapsto(1-q,1-p)$.
\end{lem}
\begin{proof}
Applying Lemma~\ref{lem:classical} gives
\begin{align}
\sum_{\substack{|\cs{\varepsilon}\rangle,|\cs{\delta}\rangle:\\\langle\cs{\varepsilon}|\Ba_k^t|\cs{\delta}\rangle\ne0}}\fint_{[\cs{\varepsilon}]}a\fint_{[\cs{\delta}]}a &=
\sum_{\substack{[\cs{\varepsilon}],[\cs{\delta}]:\\(H(t)[\cs{\varepsilon}])\cap[\cs{\delta}]\ne\emptyset}}\fint_{[\cs{\varepsilon}]}a\fint_{[\cs{\delta}]}a.
\end{align}
For the latter expression, we first fix $[\cs{\varepsilon}]$ and sum over all $(k,\ell)$-rectangles $[\cs{\delta}]$ that intersect $H(t)[\cs{\varepsilon}]$. 
By considering the symbolic classical dynamics, we can see the set $H(t)[\cs{\varepsilon}]$ intersects $D^{\eta(t)}$ distinct $(k,\ell)$-rectangles $[\cs{\delta}]$, overlapping a fractional $D^{-\eta(t)}$ part of each such $[\cs{\delta}]$.
This is because the exponent on the power of the classical baker's map $B$ in $H(t)$ is $\pm\eta(t)$, so $H(t)$ is either a left or right shift by $\eta(t)$ positions, with a possible digit swap from $R$. As illustrated below for a left shift, there are $D^{\eta(t)}$ different choices of $[\cs{\delta}]$, and $H(t)[\cs{\varepsilon}]$ intersects a fractional $D^{-\eta(t)}$ of the area of each $[\cs{\delta}]$:
\begin{equation}\label{eqn:intersection}
\tp{\begin{tikzpicture}[baseline=(current  bounding  box.center)]
\def\pw{1pt}
\node at (0,0) {\parbox{\pw}{$\bullet$\\$\bullet$}};
\node at (.25,0) {\parbox{\pw}{$\delta_\ell$\\$\varepsilon_{\ell-\eta(t)}$}};
\node at (1.25,0) {\parbox{\pw}{$\cdots$\\$\cdots$}};
\node at (1.75,0) {\parbox{\pw}{$\delta_{\eta(t)+1}$\\$\varepsilon_1$}};
\node at (2.75,0) {\parbox{\pw}{$\delta_{\eta(t)}\cdots\delta_1$\\}};
\node at (-1.4,0) {\parbox{\pw}{$\delta_{\ell+1}$\\$\varepsilon_{\ell-\eta(t)+1}$}};
\node at (-2,0) {\parbox{\pw}{$\cdots$\\$\cdots$}};
\node at (-3,0) {\parbox{\pw}{$\delta_k$\\$\varepsilon_{k-\eta(t)}$}};
\node at (-3.75,-.24) {$\varepsilon_k\cdots$};
\draw[xshift=-.1cm,dotted,line width=1pt] (-3,-.5)--(-3,.5)--(2.85,.5)--(2.85,-.5)--cycle;
\node at (-4.75,-.24) {$\cdots*\,*$};
\node at (-3.75,.2) {$\cdots*\,*$};
\node at (4.75,.2) {$*\,*\cdots$};
\node[right] at (2.75,-.24) {$*\quad*\cdots$};
\draw[decoration={brace,raise=3pt,aspect=.5,amplitude=4pt},decorate] (2.8,.5)--(4.1,.5);
\node[above right] at (3.1,.75) {$D^{\eta(t)}$ choices};
\draw[decoration={brace,raise=3pt,aspect=.5,amplitude=4pt},decorate] (-3.2,-.4)--(-4.2,-.4);
\node[below right] at (-4.25,-.6) {$D^{-\eta(t)}$ fractional overlap};
\end{tikzpicture}}
\end{equation}

On a single $[\cs{\delta}]$, the observable $a$ varies by at most $O(\|a\|_\Lip D^{-\min(\ell,k-\ell)})$; thus for any point $\x_0\in[\cs{\delta}]$ and any region $R\subseteq[\cs{\delta}]$, we have
\begin{align*}
\fint_{[\cs{\delta}]}a(\x)\,d\x&=a(\x_0)+O(\|a\|_\Lip D^{-\min(\ell,k-\ell)})\\
&=\fint_{R}a(\x)\,d\x+O(\|a\|_\Lip D^{-\min(\ell,k-\ell)}).\numberthis
\end{align*}
Applying this 
with the $D^{\eta(t)}$ non-empty regions $R_{[\cs{\delta}]}:=(H(t)[\cs{\varepsilon}])\cap[\cs{\delta}]\subseteq[\cs{\delta}]$, which each have area $|R_{[\cs{\delta}]}|=D^{-\eta(t)}N^{-1}$, and using that $B$, and hence $H(t)$, is measure-preserving and invertible, we see
\begin{align*}
\sum_{\substack{[\cs{\varepsilon}],[\cs{\delta}]:\\(H(t)[\cs{\varepsilon}])\cap[\cs{\delta}]\ne\emptyset}}&\fint_{[\cs{\varepsilon}]}a\fint_{[\cs{\delta}]}a \\
&=\sum_{[\cs{\varepsilon}]}\fint_{[\cs{\varepsilon}]}a(\x)\,d\x\, \sum_{[\cs{\delta}]}\left(\fint_{(H(t)[\cs{\varepsilon}])\cap [\cs{\delta}]}a(\x)\,d\x+O(\|a\|_\Lip D^{-\min(\ell,k-\ell)})\right)\\
&=\sum_{[\cs{\varepsilon}]}\fint_{[\cs{\varepsilon}]}a(\x)\,d\x
\left(D^{\eta(t)}\fint_{H(t)[\cs{\varepsilon}]}a(\x)\,d\x+O(\|a\|_{\Lip}D^{\eta(t)-\min(\ell,k-\ell)})\right)\\
&=D^{\eta(t)}\left(\sum_{[\cs{\varepsilon}]}\fint_{[\cs{\varepsilon}]}a(\x)\,d\x
\fint_{[\cs{\varepsilon}]}a(H(t)\x)\,d\x\right)+O(N\|a\|_\infty\|a\|_{\Lip}D^{\eta(t)-\min(\ell,k-\ell)}).\numberthis\label{eqn:mix1}
\end{align*}
For functions $a,b$ continuous on $[\cs{\varepsilon}]$,
\begin{align*}
\left|\fint_{[\cs{\varepsilon}]}a(\x)b(\x)\,d\x-\fint_{[\cs{\varepsilon}]}a(\x)\,d\x\fint_{[\cs{\varepsilon}]}b(\x)\,d\x\right|
&=\left|\fint_{[\cs{\varepsilon}]}\left(a(\x)-\fint_{[\cs{\varepsilon}]} a\right)b(\x)\,d\x\right|
\le 2\|a\|_{\Lip}D^{-\min(\ell,k-\ell)}\|b\|_\infty.
\end{align*}
Applying this with $b=a\circ H(t)$ to \eqref{eqn:mix1} yields 
\begin{align*}
\frac{1}{D^{\eta(t)}}\sum_{\substack{|\cs{\varepsilon}\rangle,|\cs{\delta}\rangle:\\\langle\cs{\varepsilon}|\Ba_k^t|\cs{\delta}\rangle\ne0}}\fint_{[\cs{\varepsilon}]}a\fint_{[\cs{\delta}]}a
&= N\int_{\T^2}a(\x)a(H(t)\x)\,d\x + O(N\|a\|_\infty\|a\|_\Lip D^{-\min(\ell,k-\ell)}),
\end{align*}
as desired.

For \eqref{eqn:inf-int}, note that $|\langle\cs{\delta}|\Ba_k^t|\cs{\varepsilon}\rangle|^2=D^{-\eta(t)}\mathbf1_{\langle\cs{\varepsilon}|\Ba_k^t|\cs{\delta}\rangle\ne0}$, and recall that $q(k)=4k$ if $D\ge3$, and $q(k)=2k$ if $D=2$.
By the exponential mixing property \eqref{eqn:expmix} with functions $a_0$ and $a_0\circ R$ (note $R$ is a classical symmetry which commutes with the classical baker's map $B$), we have
\begin{align}\label{eqn:remaining}
\begin{aligned}
&\sum_{|t|\ge k}\int_{\T^2}a_0(\x)a_0(B^t\x)\,d\x\le C_\Gamma e^{-\Gamma k}\|a_0\|_\mathrm{Lip}^2,\\
&\sum_{|t|\ge k}\int_{\T^2}a_0(\x)a_0(B^tR\x)\,d\x\le C_\Gamma e^{-\Gamma k}\|a_0\|_\mathrm{Lip}^2.
\end{aligned}
\end{align}
Therefore to obtain \eqref{eqn:inf-int}, it suffices to sum \eqref{eqn:walsh-int} over $t=-q(k)/2,\ldots,q(k)/2-1$,
since we can add in the remaining terms in \eqref{eqn:remaining} (skipping the bottom remainder if $D=2$) as part of the error term in \eqref{eqn:inf-int}.  
\end{proof}

\subsection{Averaged expected variance}\label{sec:tr-exp}

In this section, we apply the lemmas in the previous section to prove the remaining part of Proposition~\ref{prop:averages} on the expectation value of $\frac{1}{N}\sum_{j=1}^NF_j^2$.
First applying \eqref{eqn:variance} from Section~\ref{subsec:weingarten}, we have
\begin{align}\label{eqn:warm-start}
\E_\omega\left[\frac{1}{N}\sum_{j=1}^N|F_j|^2\right] &= 
\frac{q(k)}{N}\sum_{\jj=0}^{q(k)-1}\left[(\Tr \Opkl(a_0)P_\jj)^2+\Tr(\Opkl(a_0)P_\jj \Opkl(a_0)P_\jj)\right](1+o(1)),
\end{align}
where $P_\jj$ is the projection onto the eigenspace of $e^{2\pi i\jj/q(k)}$.
We will utilize the sum over $\jj$ to get cancellation that simplifies the expression compared to the individual $\E_\omega F_j$ or $\E_\omega|F_j|^2$ cases.
The main term will turn out to be the second trace, which is
\begin{align}
\sum_{\jj=0}^{q(k)-1}\Tr(\Opkl(a_0)P_\jj \Opkl(a_0)P_\jj)
&=\frac{1}{q(k)^2}\sum_{\jj=0}^{q(k)-1}\sum_{t,s=-q(k)/2}^{q(k)/2-1}\Tr(\Opkl(a_0)\Ba_k^t\Opkl(a_0)\Ba_k^s)e^{2\pi i\jj(t+s)/q(k)}.
\end{align}
Evaluating the sum over $\jj$, which only involves the term $e^{2\pi i\jj(t+s)/q(k)}$, gives complete cancellation unless $t+s=0\;\mathrm{mod}\;q(k)$, which corresponds to $t=-s$ for $|t|=|s|\le q(k)/2-1$, or $t=s=-q(k)/2$, in which case $B^s=B^{-q(k)/2}=B^{q(k)/2}=B^{-t}$, so we can still consider $t=-s$. We thus get
\begin{align}\label{eqn:tr-expand}
\sum_{\jj=0}^{q(k)-1}\Tr(\Opkl(a_0)P_\jj \Opkl(a_0)P_\jj)&=\frac{1}{q(k)}\sum_{t=-q(k)/2}^{q(k)/2-1}\Tr(\Opkl(a_0)\Ba_k^t\Opkl(a_0)\Ba_k^{-t}).
\end{align}
Equation~\eqref{eqn:inf-int} of Lemma~\ref{lem:intsum} then implies that
\begin{multline}\label{eqn:mean-var}
\frac{1}{N}\sum_{t=-q(k)/2}^{q(k)/2-1}\Tr(\Opkl(a_0)\Ba_k^t\Opkl(a_0)\Ba_k^{-t})\\
=\sum_{t=-\infty}^\infty \left(\int_{\T^2}a_0(\x)a_0(B^t\x)\,d\x+\mathbf{1}_{D\ge3}\int_{\T^2}a_0(\x)a_0(B^tR\x)\,d\x \right)+o(1).
\end{multline}

For the other trace term in \eqref{eqn:warm-start}, similar calculations and using that the $t=s=0$ term below does not contribute in the end since $\int a_0=0$, gives
\begin{align*}
\sum_{\jj=0}^{q(k)-1}(\Tr \Opkl(a_0)P_\jj)^2&=
\frac{1}{q(k)^2}\sum_{t,s=-q(k)/2}^{q(k)/2-1}\Tr(\Opkl(a_0)\Ba_k^t)\Tr(\Opkl(a_0)\Ba_k^s)\sum_{\jj=0}^{q(k)-1}e^{2\pi i\jj(t+s)/q(k)}\\
&=\frac{1}{q(k)}\sum_{\substack{t=-q(k)/2\\t\ne0}}^{q(k)/2-1}|\Tr(\Opkl(a_0)\Ba_k^t)|^2.
\numberthis\label{eqn:var-mean0}
\end{align*}
This finishes the proof of \eqref{eqn:avg-var}. \qed

Note that attempting to estimate the term $\frac{q(k)}{N}\sum_{t=-q(k)/2}^{q(k)/2-1}|\Tr(\Opkl(a_0)\Ba_k^t)|^2$ in \eqref{eqn:avg-var} using the triangle inequality combined with Proposition~\ref{prop:walsh-powers} gives a constant rather than the desired $o(1)$ bound for $D\ne4$:
\begin{align*}
\frac{1}{N}\sum_{\substack{t=-q(k)/2\\t\ne0}}^{q(k)/2-1}|\Tr(\Opkl(a_0)\Ba_k^t)|^2&=\frac{1}{N}\sum_{\substack{t=-q(k)/2\\t\ne0}}^{q(k)/2-1}\left|\sum_{[\cs{\varepsilon}]\in\mathcal R_{k,\ell}}\langle \cs{\varepsilon}|\Ba_k^t|\cs{\varepsilon}\rangle \fint_{[\cs{\varepsilon}]}a_0\right|^2\\
&\le \frac{1}{N}\sum_{\substack{t=-q(k)/2\\t\ne0}}^{q(k)/2-1}|D^{\eta(t)}D^{-\eta(t)/2}|^2\|a_0\|_\infty^2+\frac{1}{N}\underbrace{4^k\oneb_{\{D\ge4\text{ even}\}}\|a_0\|_\infty^2}_{\text{ from }t=-2k\text{ term}}\\
&\le\frac{C}{N}\sum_{t=1}^{k}D^{t}\|a_0\|_\infty^2\le C'\|a_0\|_\infty^2,
\end{align*}
where we used that $N=D^k$ and $q(k)/2\ge k$. 
In the next section, we will use the phases of $\Ba_k^t$ to obtain cancellations which will in particular imply such terms are actually $o(1)$ for $D\ne4$.

\section{Quantum time evolution and phase cancellation}\label{sec:phases}

In order to show \eqref{eqn:var} in Theorem~\ref{thm:mat}, we need to show convergence in probability of the quantum variance to its mean. 
To this end, we want to prove that for an {individual} $F_j$, that $\E_\omega F_j^2-(\E_\omega F_j)^2\to V(a)$. 
This will require a better understanding of the quantum phases of the time evolution $\Ba_k^t$ for all times $t$, in particular, up to and after the Ehrenfest time $t=k$, at least well enough to demonstrate some cancellation of terms.
Conveniently, we only need a fairly small, though important, amount of cancellation.
The main result of this section is Proposition~\ref{thm:trace-cancel} on bounding $\Tr(\Opkl(a)\Ba_k^{t_1})$ and $\Tr(\Opkl(a)\Ba_k^{t_1}\Opkl(a)\Ba_k^{t_2})$. We will prove this using several lemmas regarding cancellation due to phases, and then apply it in Section~\ref{sec:tr-ind} to prove convergence of the individual means and variances of the $F_j$.

To explain and motivate the cancellation due to phases, we first consider a naive bound for $\Tr(\Ba_k^t)$, estimated by applying Proposition~\ref{prop:walsh-powers}(i,ii), in particular that there are $D^{\eta(t)}$ nonzero diagonal entries $\langle\cs{\varepsilon}|\Ba_k^t|\cs{\varepsilon}\rangle$, each of magnitude $|\langle\cs{\varepsilon}|\Ba_k^t|\cs{\varepsilon}\rangle|=D^{-\eta(t)/2}$: 
\begin{align}\label{eqn:trace0}
|\Tr(\Ba_k^t)|&=\bigg|\sum_{[\cs{\varepsilon}]}\langle\cs{\varepsilon}|\Ba_k^t|\cs{\varepsilon}\rangle\bigg|
\le\sum_{[\cs{\varepsilon}]}|\langle\cs{\varepsilon}|\Ba_k^t|\cs{\varepsilon}\rangle|
\le D^{\eta(t)}D^{-\eta(t)/2}=D^{\eta(t)/2}.
\end{align}
This type of estimate is useful when $\eta(t)$ is small, but we will need a better estimate for $t$ near $\pm k$. 
We should expect that for $t$ near the Ehrenfest time $k$, that the phases of $\langle\cs{\varepsilon}|\Ba_k^t|\cs{\varepsilon}\rangle$ oscillate over quantum rectangles $[\cs{\varepsilon}]$, and so there should be lots of cancellation in sums like in $\Tr(\Ba_k^t)$. 
As a numerical picture, Fig.~\ref{fig:phases} plots the phases of each term $\langle\cs{\varepsilon}|\Ba_k^t|\cs{\varepsilon}\rangle$ in the case $D=3$ and $N=3^8=6561$ for $t=k$ and $t=k-1$. As pictured, the phases vary over nearby rectangles $[\cs{\varepsilon}]$, which will cause cancellation.

\begin{figure}[htb]
\includegraphics[height=2.5in]{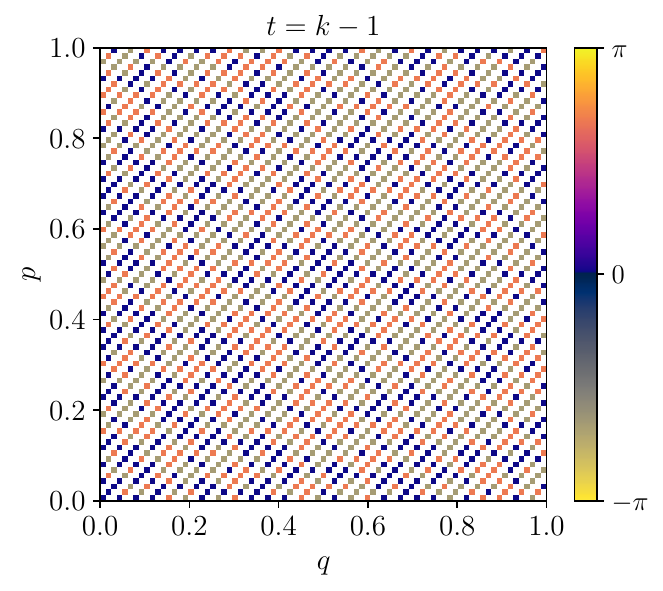}
\includegraphics[height=2.5in]{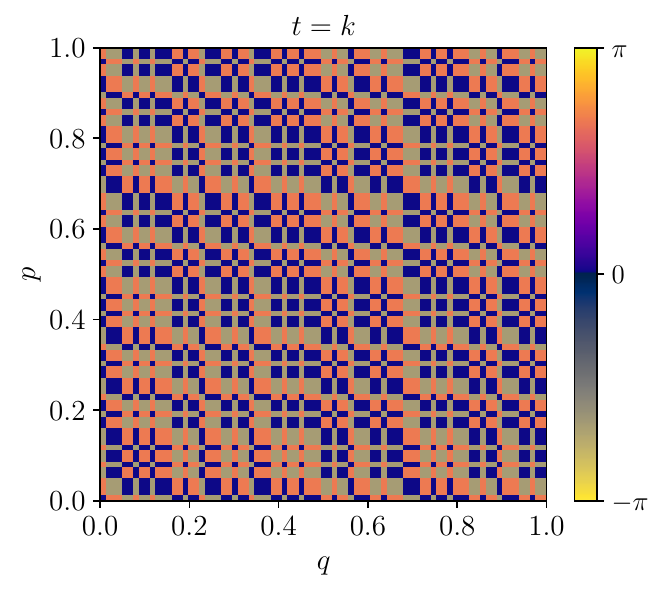}
\caption{Phases of $\langle\cs{\varepsilon}|\Ba_k^t|\cs{\varepsilon}\rangle$ for $D=3$, $k=8$, $\ell=4$, and $t=7,8$, plotted as a function of the quantum rectangle $[\cs{\varepsilon}]\subset\T^2$ in coordinates $q=\varepsilon_\ell\cdots\varepsilon_1$ and $p=\varepsilon_k\cdots\varepsilon_{\ell+1}$. Entries where the value is zero are shown in white. 
There are no large groupings of a single phase, suggesting that the phases vary enough over nearby rectangles to cause cancellations.} \label{fig:phases}
\end{figure}

Similarly, in order to estimate a more complicated quantity like $\Tr(\Opkl(a)\Ba_k^{t})=\sum_{[\cs{\varepsilon}]}\langle\cs{\varepsilon}|\Ba_k^{t}|\cs{\varepsilon}\rangle\fint_{[\cs{\varepsilon}]}a$, we will rely on cancellations of $\langle\cs{\varepsilon}|\Ba_k^t|\cs{\varepsilon}\rangle$ over small regions in the torus. We will partition $\T^2$ into many small squares $\{S\}$, and show that $\sum_{[\cs{\varepsilon}]\in S}\langle\cs{\varepsilon}|\Ba_k^t|\cs{\varepsilon}\rangle$, the sum over a single square $S$, is ``small'' due to phase cancellations (Lemmas~\ref{lem:trace-part}, \ref{lem:trace-part2}), as visualized in Fig.~\ref{fig:phases}. Then on a single square $S$, $a$ will be nearly constant, so 
\[\sum_{[\cs{\varepsilon}]\in S}\langle\cs{\varepsilon}|\Ba_k^t|\cs{\varepsilon}\rangle\fint_{[\cs{\varepsilon}]}a\approx \Bigg(\sum_{[\cs{\varepsilon}]\in S}\langle\cs{\varepsilon}|\Ba_k^t|\cs{\varepsilon}\rangle\Bigg)\fint_S a
\]
is also small, and summing over all squares $S$ will show $\Tr(\Opkl(a)\Ba_k^t)$ is small (Lemma~\ref{lem:a-cancellation}).

We start with the following lemma on just $\Tr(\Ba_k^t)$, and then will use its proof ideas to prove the main result of this section, Proposition~\ref{thm:trace-cancel}, on traces with $\Opkl(a)$.
\begin{lem}\label{lem:trace}
For any $t\in\Z$, 
\begin{align}\label{eqn:trace-gcd}
|\Tr(\Ba_k^t)|&\le D^{\gcd([t]_k,k)}.
\end{align}
For $t=\pm k$ (or more generally, $t=k$ or $3k\;\mathrm{mod}\;q(k)$), we have the precise result
\begin{align}\label{eqn:trace-k}
|\Tr(\Ba_k^{\pm k})|&=\begin{cases} 1,&D=1\;\mathrm{or }\;3\;\mathrm{mod}\;4\\
0,&D=2\;\mathrm{mod}\;4\\
2^{k/2},&D=0\;\mathrm{mod}\;4
\end{cases}.
\end{align}
For $t=2k$, we have
\begin{align}\label{eqn:trace-2k}
\Tr(\Ba_k^{2k})&=\begin{cases}1,&D\text{ odd}\\
2^k,&D\text{ even}
\end{cases}.
\end{align}
\end{lem}
\begin{proof}
We take the trace in the position basis for notational convenience, though one could just as easily use a $(k,\ell)$-coherent state basis, which we will do later in this section. The position basis elements will be denoted by $|\varepsilon\rangle\equiv|\varepsilon_k\cdots\varepsilon_1\rangle\equiv|\varepsilon_k\rangle\otimes\cdots\otimes|\varepsilon_1\rangle$ for $\varepsilon_j\in\intbrr{0:D-1}$, so we have
\begin{align}\label{eqn:trace-1}
\Tr(\Ba_k^t)&=\sum_{\varepsilon_k=0}^{D-1}\cdots\sum_{\varepsilon_1=0}^{D-1}\langle \varepsilon_k\cdots\varepsilon_1|\Ba_k^t|\varepsilon_k\cdots\varepsilon_1\rangle.
\end{align}
From the definition \eqref{eqn:walsh-action} of $\Ba_k$, we see the effect of $\Ba_k^t$ on tensor product states $|\varepsilon\rangle$ is always a leftwards cycle by $[t]_k=t\;\mathrm{mod}\;k$, along with some power of discrete Fourier transform matrices $\hat{F}_D^\dagger$ applied to some of the vectors. (Note for $t<0$, it is a rightwards cycle by $[|t|]_k=[-t]_k=[k-t]_k$, which is equivalent to a left shift by $[t]_k$.)
Due to the tensor product structure we can represent the terms in the inner product $\langle \varepsilon_k\cdots\varepsilon_1|\Ba_k^t|\varepsilon_k\cdots\varepsilon_1\rangle$ graphically, similarly as in the proof of Lemma~\ref{lem:classical}, as
\begin{equation}\label{eqn:tikz1}
\tp{\begin{tikzpicture}[baseline=(current  bounding  box.center),ampersand replacement=\&]
\node[matrix,right] at (0,0)
{
	\node {$|\varepsilon\rangle$:}; \&\node {$\varepsilon_k$}; \&\node {$\varepsilon_{k-1}$}; \&\node {$\cdots$}; \&\node {$\varepsilon_2$}; \& \node {$\varepsilon_1$}; \\
	\node{};\&\\
	\node {$\Ba_k^t|\varepsilon\rangle$:}; \&\node {$\fmat\varepsilon_{k-[t]_k}$}; \&\node {$\fmat\varepsilon_{k-[t]_k-1}$}; \&\node {$\cdots$}; \&\node {$\fmat\varepsilon_{2-[t]_k+k}$}; \& \node {$\fmat\varepsilon_{1-[t]_k+k}$}; \\
};
\node at (9.8,0) {,};
\foreach \pos in {1.5, 3.2, 6, 8}{
\draw[line width=2] (\pos,-.15)--++(0,.4);
}
\end{tikzpicture}}
\end{equation}
where $\fmat$ represents an operator in $\{\operatorname{Id},\hat{F}_D^\dagger,\hat{R}_D,\hat{F}_D\}$, which may differ for different $\varepsilon_j$, and where the indices on the $\varepsilon_j$ are taken modulo $k$. The thick vertical lines between the rows indicate an inner product.

Each term $\varepsilon_j$, $j=1,\ldots,k$, appears exactly twice in \eqref{eqn:trace-1} or \eqref{eqn:tikz1}, once on the left side of an inner product (top row of \eqref{eqn:tikz1}), and once on the right side (bottom row of \eqref{eqn:tikz1}).
Using that $\sum_{\varepsilon_j=0}^{D-1}|\varepsilon_j\rangle\langle\varepsilon_j|=\operatorname{Id}$, we sum over the variables in \eqref{eqn:tikz1} by removing or ``contracting'' them whenever they appear as $|\varepsilon_j\rangle\langle\varepsilon_j|$.
Graphically, we express this by drawing connections between $\varepsilon_j$ in the top row of \eqref{eqn:tikz1} and the same $\varepsilon_j$ in the bottom row of \eqref{eqn:tikz1}.
This generates a graph $G_{k,t}$ with a number of cycles.
A specific example where $k=6$ and $t=3$, with different cycles shown in different colors, is
\begin{equation}\label{eqn:tikz2}
\tp{\begin{tikzpicture}[baseline=(current  bounding  box.center),ampersand replacement=\&]
\node[matrix,right] at (0,0)
{
	\node {$|\varepsilon\rangle$:}; \&\node {$\varepsilon_6$}; \&\node {$\varepsilon_{5}$}; \&\node {$\varepsilon_4$}; \&\node {\;\;\;$\varepsilon_3$};\&\node {\;\;$\varepsilon_2$}; \& \node {\;\;$\varepsilon_1$}; \\
	\node{};\&\\
	\node{};\&\\
	\node {$\Ba_k^t|\varepsilon\rangle$:}; \&\node {$\varepsilon_{3}$}; \&\node {$\varepsilon_{2}$}; \&\node {$\varepsilon_{1}$}; \&\node{$\hat{F}_D^\dagger\varepsilon_6$}; \&\node {$\hat{F}_D^\dagger\varepsilon_{5}$}; \& \node {$\hat{F}_D^\dagger\varepsilon_{4}$}; \\
};
\def\topc{.4};
\def\botc{-.25};
\def\lw{.8};
\draw[color=blue,line width=\lw] (5.5,\topc)--(2.6,\botc)--(2.6,\topc)--(5.5,\botc)--cycle;
\draw[color=orange,line width=\lw]  (4.5,\topc)--(2.1,\botc)--(2.1,\topc)--(4.5,\botc)--cycle;
\draw[color=green,line width=\lw] (3.5,\topc)--(1.6,\botc)--(1.6,\topc)--(3.5,\botc)--cycle;
\def\tlw{3};
\draw[color=blue, line width=\tlw] (5.5,\topc)--(5.5,\botc);
\draw[color=blue, line width=\tlw] (2.6,\topc)--(2.6,\botc);
\draw[color=orange, line width=\tlw] (4.5,\topc)--(4.5,\botc);
\draw[color=orange, line width=\tlw] (2.1,\topc)--(2.1,\botc);
\draw[color=green, line width=\tlw] (3.5,\topc)--(3.5,\botc);
\draw[color=green, line width=\tlw] (1.6,\topc)--(1.6,\botc);
\end{tikzpicture}}
\end{equation}
For each distinct cycle, summing over all except one of the variables leaves a term of the form
\begin{align}\label{eqn:allbutone}
\sum_{\varepsilon_j=0}^{D-1}\langle\varepsilon_j|\hat{F}_D^\gamma|\varepsilon_j\rangle,
\end{align}
where $\gamma\in\{0,1,2,3\}$, and $\hat{F}_D^\gamma$ represents all the terms $\fmat$ from \eqref{eqn:tikz1} that are picked up along the cycle. The sum \eqref{eqn:allbutone} has magnitude at most $D$, corresponding to the ``worst case'' $\gamma=0$.
Using this bound for each cycle, so that we do not need to determine the precise values of $\gamma$, yields
\begin{align}\label{eqn:trace-cycles}
|\Tr(\Ba_k^t)|&\le D^{\#\text{cycles in $G_{k,t}$}}.
\end{align}
Because the action of $\Ba_k^t$ is a left cycle by $[t]_k$, the length of each cycle is twice (due to the 2 rows of $G_{k,t}$) the order of $[t]_k$ in the additive group $\Z/k\Z$, i.e. twice the smallest $s>0$ so that $s[t]_k\in k\Z$. Since the order of $[t]_k$ is $s=k/\operatorname{gcd}(k,[t]_k)$, the total number of cycles is
\begin{align}
\frac{2k}{2s}&=\gcd([t]_k,k),
\end{align}
giving \eqref{eqn:trace-gcd}.

For the case $t=\pm k$ in \eqref{eqn:trace-k}, we note that
\begin{align}
\langle\varepsilon|\Ba_k^{\pm k}|\varepsilon\rangle&=\prod_{j=1}^k\langle\varepsilon_j|\hat{F}_D^\pm|\varepsilon_j\rangle,
\end{align}
where $\hat{F}_D^\pm$ denotes $\hat{F}_D^\dagger$ if $+$ and $\hat{F}_D$ if $-$.
In this case, we evaluate directly,
\begin{align*}
\left|\sum_{\varepsilon_k=0}^{D-1}\cdots\sum_{\varepsilon_1=0}^{D-1}\langle{\varepsilon}|\Ba_k^{\pm k}|{\varepsilon}\rangle\right|&=\left|\frac{1}{D^{k/2}}\sum_{\varepsilon_k=0}^{D-1}\cdots\sum_{\varepsilon_1=0}^{D-1}\prod_{j=1}^De^{\mp2\pi i\varepsilon_j^2/D}\right|,
\end{align*}
and use the Gauss quadratic sum formula
\begin{align}\label{eqn:gauss}
\sum_{m=0}^{D-1}e^{2\pi im^2/D}&=\begin{cases}\sqrt{D},&D=1\;\mathrm{mod}\;4\\
0,&D=2\;\mathrm{mod}\;4\\
i\sqrt{D},&D=3\;\mathrm{mod}\;4\\
(1+i)\sqrt{D},&D=0\;\mathrm{mod}\;4
\end{cases},
\end{align}
to obtain \eqref{eqn:trace-k}.

For the case $t=2k$ in \eqref{eqn:trace-2k}, we have
\begin{align}
\langle\varepsilon|\Ba_k^{2k}|\varepsilon\rangle&=\prod_{j=1}^k\langle\varepsilon_j|\hat R_D|\varepsilon_j\rangle,
\end{align}
where $\hat R_D=(\hat F_D^\dagger)^2:|x\rangle\mapsto|-x\;\mathrm{mod}\;D\rangle$. If $D$ is even, there are exactly two values for $\varepsilon_j\in\{0,\ldots,D-1\}$ such that $\langle\varepsilon_j|\hat R_D|\varepsilon_j\rangle\ne0$; these are $\varepsilon_j=0$ and $\varepsilon_j=D/2$. (This holds trivially for $D=2$ for which $\hat R_D=I_2$). This gives $\Tr(\Ba_k^{2k})=2^k$. If $D$ is odd, there is only one such $\varepsilon_j$, which is $\varepsilon_j=0$, giving $\Tr(\Ba_k^{2k})=1$.
\end{proof}

Similar ideas will be used to prove the main result of this section:
\begin{prop}[trace bounds]\label{thm:trace-cancel}
Let $a:\T^2\to\C$ be a Lipschitz observable, and let $\Ba_k$ be the quantization of the $D$-baker's map.
Then for $t_1,t_2\in\Z$ and $r\in\intbrr{0:\min(\ell,k-\ell)}$,
\begin{align}\label{eqn:tr-bound1}
|\Tr(\Opkl(a)\Ba_k^{t_1})|\le \begin{cases}
\|a\|_\infty D^{2r}G_{D,k}(t_1)+\|a\|_\mathrm{Lip}\sqrt{2}D^{k/2-r},&t_1\not\in q(k)\Z\\
N\left|\int_{\T^2}a(\x)\,d\x\right|,&t_1\in q(k)\Z
\end{cases},
\\
|\Tr(\Opkl(a)\Ba_k^{t_1}\Opkl(a)\Ba_k^{t_2})|\le \|a\|_\infty^2 D^{4r}G_{D,k}(t_1+t_2)+2\sqrt{2}\|a\|_\mathrm{Lip}\|a\|_\infty D^{k-r},\label{eqn:tr-bound2}
\end{align}
where 
\begin{align}\label{eqn:tr-boundp1}
G_{D,k}(t):=D^{\operatorname{gcd}([t]_k,k)},
\end{align}
for $[t]_k=t\;\mathrm{mod}\;k$.
We also have better estimates for certain $t_1,t_2$. For \eqref{eqn:tr-bound1}, we can take
\begin{align}
G_{D,k}(\pm k)=D^{k/4}. 
\end{align}
For \eqref{eqn:tr-bound2} and $[t_1+t_2]_k=0$ with $t_1+t_2\ne0\;\mathrm{mod}\;q(k)$, we can take
\begin{align}\label{eqn:tr-boundp2}
G_{D,k}(t_1+t_2)&=\begin{cases}1,&D=2\text{ or $D$ odd}\\2^k,&D\ge4\text{ even}\end{cases}.
\end{align}
\end{prop}

\subsection{Proof of Proposition~\ref{thm:trace-cancel}}
The rest of this section is devoted to the proof of Proposition~\ref{thm:trace-cancel}.
To prove the theorem, we will use similar ideas as in Lemma~\ref{lem:trace}, but we will need estimates for when we only sum over coherent states corresponding to certain regions of $\T^2$. It will be enough for our purposes to show that the bound in Lemma~\ref{lem:trace} for the trace also holds for these ``partial'' traces in coherent state bases.
We state Lemma~\ref{lem:trace-part} below in more generality than needed; we only need the cases $p=1,2$ for Proposition~\ref{thm:trace-cancel}, but we give a proof for $p\in\N$ since $p>2$ is the same method as for $p=2$.

\begin{lem}[partial traces over squares] 
\label{lem:trace-part}
For any $p\in\N$ and $r\in\intbrr{0:\min(\ell,k-\ell)}$, let $S_1,\ldots,S_p\subseteq\T^2$ be $D^{-r}\times D^{-r}$ squares with edges in the grid $D^{-r}\Z^2$. Let $(k,\ell)$-coherent states be denoted by the notation $|\csn{\varepsilon^{(j)}}\rangle$ (with the prime dropped for notational convenience, i.e. $|\csn{\varepsilon^{(j)}}\rangle$ would be $|\varepsilon^{(j)}{}'\cdot\varepsilon^{(j)}\rangle$ in the usual notation). Then for any $t_1,\ldots,t_p\in\Z$,
\begin{align}\label{eqn:trace-partial-p}
\Bigg|\sum_{[\csn{\varepsilon^{(1)}}]\in S_1}\cdots\sum_{[\csn{\varepsilon^{(p)}}]\in S_p}\langle\csn{\varepsilon^{(1)}}|\Ba_k^{t_1}|\csn{\varepsilon^{(2)}}\rangle
\langle\csn{\varepsilon^{(2)}}|\Ba_k^{t_2}|\csn{\varepsilon^{(3)}}\rangle
\cdots
\langle\csn{\varepsilon^{(p)}}|\Ba_k^{t_p}|\csn{\varepsilon^{(1)}}\rangle
\Bigg|
\le D^{\gcd([t_1+\cdots+t_p]_k,k)}.
\end{align}
\end{lem}
\begin{proof}
A $D^{-r}\times D^{-r}$ square $S=\{(q,p)\in\T^2: m D^{-r}\le q<(m+1)D^{-r}, n D^{-r}\le p<(n+1)D^{-r}\}$, for $m,n\in\intbrr{0:D^r-1}$, can be defined in terms of $D$-ary expansions of $(q,p)\in\T^2$ as follows. Writing $m=\sum_{i=1}^{r}x_iD^{r-i}$ and $n=\sum_{i=1}^{r}y_iD^{r-i}$, then
\begin{align}\label{eqn:Sdef}
S&=\big\{(q,p)\in\T^2:\underbrace{\cdots**y_{r}\cdots y_1}_{\longleftarrow p}\bullet \underbrace{x_1\cdots x_{r}**\cdots}_{q\longrightarrow}\big\},
\end{align}
i.e. $S$ is the set of all points $(q,p)\in\T^2$ such that the $D$-ary expansion of $q$ starts with the $r$ digits $0.x_1\cdots x_{r}$ and the $D$-ary expansion of $p$ starts with the $r$ digits $0.y_1\cdots y_r$.
Recall the $D$-ary expansion definition of a $(k,\ell)$-rectangle $[\cs{\varepsilon}]$ from \eqref{eqn:rectangle}. Since $r\le \min(\ell,k-\ell)$, restricting a coherent state rectangle $[\cs{\varepsilon}]$ to lie in a $D^{-r}\times D^{-r}$ square $S$ amounts to fixing the first $r$ coordinates in the $D$-ary expansion of $\varepsilon$ and $\varepsilon'$ to match $S$; i.e. fixing the $2r$ coordinates
\begin{align}
\numberthis\label{eqn:frozen}
\varepsilon_{\ell+r},\ldots,\varepsilon_{\ell+1}\bullet\varepsilon_\ell,\ldots\varepsilon_{\ell-r+1}
\end{align}
to match the $2r$ corresponding values $y_r,\ldots, y_1\bullet x_1,\ldots,x_r$ which identify the square in \eqref{eqn:Sdef}.

To prove the lemma, we first prove the $p=1$ case, and then generalize to arbitrary $p\in\N$.
For $p=1$, we then want to bound the sum
\begin{align}\label{eqn:p1}
\Bigg|\sum_{[\cs{\varepsilon}]\in S_1}\langle\cs{\varepsilon}|\Ba_k^{t_1}|\cs{\varepsilon}\rangle\Bigg|&=\Bigg|\sum_{\varepsilon_k=0}^{D-1}\cdots\sum_{\varepsilon_{\ell+r+1}=0}^{D-1}\sum_{\varepsilon_{\ell-r}=0}^{D-1}\cdots\sum_{\varepsilon_1=0}^{D-1}\langle\cs{\varepsilon}|\Ba_k^t|\cs{\varepsilon}\rangle \Bigg|,
\end{align}
where the sums are over the variables $\varepsilon_j$ which do not appear in the fixed coordinates in \eqref{eqn:frozen}.
As for the graphs $G_{k,t}$ in \eqref{eqn:tikz1} and \eqref{eqn:tikz2}, we can draw the edges representing $\langle\cs{\varepsilon}|\Ba_k^t|\cs{\varepsilon}\rangle$ and determine the resulting cycles. The only difference is that since we have a $(k,\ell)$-coherent state, the top row for $|\cs{\varepsilon}\rangle$ is cyclically shifted and has some terms with $\hat{F}_D^\dagger$ applied according to \eqref{eqn:cs}, and the bottom row has the same cyclic shift and additional $\hat{F}_D^\dagger$ terms. This does not affect the cycles, which are the same as before. However, some of these cycles contain the frozen coordinates $\varepsilon_{\ell+r},\ldots,\varepsilon_{\ell-r+1}$ in \eqref{eqn:frozen}.
When summing over all the variables $\varepsilon_j$ in a cycle that contains any frozen coordinates, we do not end up with a term of the form \eqref{eqn:allbutone}. Instead, after summing over all non-frozen variables in the cycle, we are left with a single term of inner products only involving frozen variables. Because every inner product $\langle \varepsilon_j|\hat{F}_D^\gamma|\varepsilon_{j'}\rangle$ has norm at most $1$, the contribution to \eqref{eqn:p1} from a cycle with frozen coordinates is at most $1$, which is less than the bound $D$ we used for \eqref{eqn:allbutone}. Without keeping track of which cycles have frozen coordinates or not, and just using the bound $D$ for each cycle, the same reasoning following \eqref{eqn:allbutone} then gives the same bound,
\begin{align}
\Bigg|\sum_{[\cs{\varepsilon}]\in S_1}\langle\cs{\varepsilon}|\Ba_k^{t_1}|\cs{\varepsilon}\rangle\Bigg|&\le D^{\#\text{cycles in $G_{k,{t_1}}$}}=D^{\gcd([t_1]_k,k)}.
\end{align}

For $p>1$, we apply a similar argument, but the analogue of the graph in \eqref{eqn:tikz1} is more complicated.
Recall we use $|\csn{\varepsilon^{(j)}}\rangle$ as shorthand for $|(\varepsilon^{(j)})'\cdot\varepsilon^{(j)}\rangle$.
Each inner product $\langle\csn{\varepsilon^{(m)}}|\Ba_k^{t_m}|\csn{\varepsilon^{(m+1)}}\rangle$ in \eqref{eqn:trace-partial-p} creates two rows as in \eqref{eqn:tikz1}, and the corresponding graph $G_{k,\mathbf{t},p}$ now includes all $2p$ rows.
This graph $G_{k,\mathbf{t},p}$ is shown in Fig.~\ref{fig:tikz-p}.
We omit drawing the terms $\fmat$ on each $\varepsilon^{(j)}_m$, but note that a power of $\hat{F}_D^\dagger$ is permitted on many of the entries.

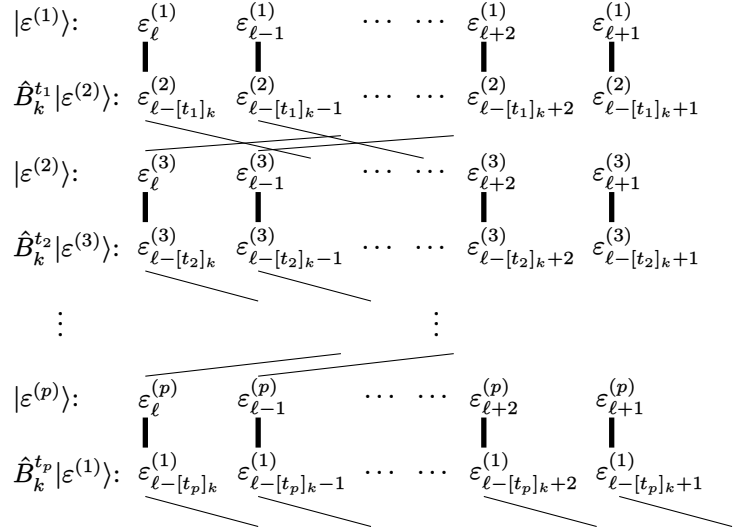
\begin{figure}[!ht]
\tp{\begin{tikzpicture}[baseline=(current  bounding  box.center),ampersand replacement=\&]
\node[matrix,right] at (0,0)
{
	\node {$|\csn{\varepsilon^{(1)}}\rangle$:}; \&\node {$\varepsilon^{(1)}_\ell$};  \&\node {$\varepsilon^{(1)}_{\ell-1}$};\&\node {$\cdots$};\&\node {$\cdots$}; \&\node {$\varepsilon^{(1)}_{\ell+2}$}; \& \node {$\varepsilon^{(1)}_{\ell+1}$}; \\
	\node{};\&\\
	\node {$\Ba_k^{t_1}|\csn{\varepsilon^{(2)}}\rangle$:}; \&\node {$\varepsilon^{(2)}_{\ell-[t_1]_k}$};  \&\node {$\varepsilon^{(2)}_{\ell-[t_1]_k-1}$};\&\node {$\cdots$};\&\node {$\cdots$}; \&\node {$\varepsilon^{(2)}_{\ell-[t_1]_k+2}$}; \& \node {$\varepsilon^{(2)}_{\ell-[t_1]_k+1}$}; \\
};
\foreach \vert in {0,-2,-5}{
	\foreach \pos in {2,3.5,6.5,8.2}{
		\draw[yshift=\vert cm, line width=2] (\pos,-.15)--++(0,.4);
	}
}
\foreach \lpos in {2,3.5}{
	\draw (\lpos,-.8)--(\lpos+2.2,-1.3); 
	\draw (\lpos,-1.2)--(\lpos+2.6,-1); 
}
\foreach \lpos in {2,3.5}{
	\draw[yshift= -2cm] (\lpos,-.8)--(\lpos+1.5,-1.2); 
	\draw[yshift=-3cm] (\lpos,-1.2)--(\lpos+2.6,-.9); 
}
\foreach \lpos in {2,3.5,6.5,8.3}{
	\draw[yshift= -5cm] (\lpos,-.8)--(\lpos+1.5,-1.2); 
}

\node[matrix,right] at (0,-2)
{
	\node {$|\csn{\varepsilon^{(2)}}\rangle$:}; \&\node {$\varepsilon^{(3)}_\ell$};  \&\node {$\varepsilon^{(3)}_{\ell-1}$};\&\node {$\cdots$};\&\node {$\cdots$}; \&\node {$\varepsilon^{(3)}_{\ell+2}$}; \& \node {$\varepsilon^{(3)}_{\ell+1}$}; \\
	\node{};\&\\
	\node {$\Ba_k^{t_2}|\csn{\varepsilon^{(3)}}\rangle$:}; \&\node {$\varepsilon^{(3)}_{\ell-[t_2]_k}$};  \&\node {$\varepsilon^{(3)}_{\ell-[t_2]_k-1}$};\&\node {$\cdots$};\&\node {$\cdots$}; \&\node {$\varepsilon^{(3)}_{\ell-[t_2]_k+2}$}; \& \node {$\varepsilon^{(3)}_{\ell-[t_2]_k+1}$}; \\
};

\foreach \x in {0,5}{
\node[xshift=\x cm,right] at (.7,-3.4) {$\vdots$};
}

\node[matrix,right] at (0,-5)
{
	\node {$|\csn{\varepsilon^{(p)}}\rangle$:}; \&\node {$\varepsilon^{(p)}_\ell$};  \&\node {$\varepsilon^{(p)}_{\ell-1}$};\&\node {$\cdots$};\&\node {$\cdots$}; \&\node {$\varepsilon^{(p)}_{\ell+2}$}; \& \node {$\varepsilon^{(p)}_{\ell+1}$}; \\
	\node{};\&\\
	\node {$\Ba_k^{t_p}|\csn{\varepsilon^{(1)}}\rangle$:}; \&\node {$\varepsilon^{(1)}_{\ell-[t_p]_k}$};  \&\node {$\varepsilon^{(1)}_{\ell-[t_p]_k-1}$};\&\node {$\cdots$};\&\node {$\cdots$}; \&\node {$\varepsilon^{(1)}_{\ell-[t_p]_k+2}$}; \& \node {$\varepsilon^{(1)}_{\ell-[t_p]_k+1}$}; \\
};
\end{tikzpicture}}
\caption{The graph $G_{k,\mathbf{t},p}$ describing \eqref{eqn:trace-partial-p} for $p>1$. The hanging edges in the bottom row connect to the corresponding variable in the top row.}\label{fig:tikz-p}
\end{figure}

Since we will be summing over all the non-frozen entries $\varepsilon^{(m)}_j\in\intbrr{0:D-1}$, we still consider cycles and removing/contracting those variables in the graph.
As before, every cycle will still either contribute a term of the form \eqref{eqn:allbutone}, which has norm at most $D$, or, if it involves any frozen variables, a term of norm at most $1$. We thus have the upper bound $D^{\#\mathrm{cycles}}$ again, and just need to count the number of cycles in this more complicated graph.

We can do this by focusing only on the top row of the graph corresponding to $\csn{\varepsilon^{(1)}}$. A cycle in the whole graph is always of the form 
\begin{align*}
\varepsilon^{(1)}_j\to\varepsilon^{(2)}_{j-[t_1]_k}\to\cdots\to \varepsilon^{(p)}_{j-[t_1]_k-[t_2]_k-\cdots-[t_{p-1}]_k}\to\varepsilon^{(1)}_{j-[t_1]_k-[t_2]_k-\cdots-[t_p]_k}\to\cdots\to\varepsilon^{(1)}_j,
\end{align*}
where the subscripts are taken modulo $k$. 
In other words, the path always goes from row 1, to row 2, $\ldots$, to row $p$, and then back to row 1, and so on until completing a cycle.
If we only consider the coordinates for $\varepsilon^{(1)}$ in the cycle, we see they are
\begin{align*}
\varepsilon^{(1)}_j\to\varepsilon^{(1)}_{j-[t_1]_k-[t_2]_k-\cdots-[t_p]_k}\to\varepsilon^{(1)}_{j-2([t_1]_k-[t_2]_k-\cdots-[t_p]_k)}\to\cdots\to\varepsilon^{(1)}_j.
\end{align*}
The action each time is a cyclic left shift by $[t_1]_k+\cdots+[t_p]_k$ modulo $k$. 
The argument following \eqref{eqn:trace-cycles} then yields the upper bound
\begin{align}
D^{\#\mathrm{cycles\;in\;}G_{k,\mathbf{t},p}}\le D^{\gcd([t_1+\cdots+t_p]_k,k)},
\end{align}
which gives \eqref{eqn:trace-partial-p}.
\end{proof}

Lemma~\ref{lem:trace-part} does not offer a good bound when $[t_1+\cdots+t_p]_k=k$. 
In this case we will need a better estimate for $p=1$ and $p=2$. 
\begin{lem}[partial traces, special $t$]\label{lem:trace-part2}
Let $S_1,S_2$ be $D^{-r}\times D^{-r}$ squares as in the previous lemma.
\begin{enumerate}[(i)]
\item For $t_1=\pm k$,
\begin{align}
\Bigg|\sum_{[\cs{\varepsilon}]\in S_1}\langle\cs{\varepsilon}|\Ba_k^{t_1}|\cs{\varepsilon}\rangle\Bigg|
&\le \begin{cases}1,&D\ne0\;\mathrm{mod}\;4\\
D^{k/4},&D=0\;\mathrm{mod}\;4
\end{cases}.
\end{align}

\item For $[t_1+t_2]_k=0$ with $t_1+t_2\ne 0\;\mathrm{mod}\;q(k)$,
\begin{align}\label{eqn:trace-partial-2}
\Bigg|\sum_{[\cs{\varepsilon}]\in S_1}\sum_{[\cs{\delta}]\in S_2}\langle\cs{\varepsilon}|\Ba_k^{t_1}|\cs{\delta}\rangle\langle\cs{\delta}|\Ba_k^{t_2}|\cs{\varepsilon}\rangle\Bigg|
&\le \begin{cases}
1,& D=2\text{ or $D$ odd}\\
2^{k},&D\ge4\text{ even}
\end{cases}.
\end{align}
\end{enumerate}
\end{lem}
\begin{proof}
(i) When $t_1=\pm k$, $\Ba_k^{t_1}=(F_D^\gamma)^{\otimes k}$ for $\gamma=\mp 1$. Although $|\cs{\varepsilon}\rangle$ has $k-\ell$ coordinate basis vectors Fourier transformed according to \eqref{eqn:cs}, these extra Fourier transforms will play no role in the inner product $\langle\cs{\varepsilon}|\Ba_k^{t_1}|\cs{\varepsilon}\rangle$ since they will cancel with the matching Fourier transform on the other side of the inner product. Letting $\frozen_r:=\{\ell+r,\ldots,\ell-r+1\}$ denote the $2r$ ``frozen'' indices corresponding to $S_1$ as in \eqref{eqn:frozen}, then
\begin{align*}
\Bigg|\sum_{[\cs{\varepsilon}]\in S_1}\langle\cs{\varepsilon}|\Ba_k^{t_1}|\cs{\varepsilon}\rangle\Bigg|=\Bigg|\sum_{[\cs{\varepsilon}]\in S_1}\prod_{i=1}^{k}\langle\varepsilon_i|\hat F_D^\gamma|\varepsilon_i\rangle\Bigg|
&=\Bigg|\prod_{j\not\in\frozen_r}\sum_{\varepsilon_j=0}^{D-1}\langle\varepsilon_j|\hat F_D^\gamma|\varepsilon_j\rangle\Bigg|\left|\prod_{i\in\frozen_r}\langle\varepsilon_i|\hat F_D^\gamma|\varepsilon_i\rangle\right|\\
&\le\begin{cases}
D^{-r},&D=1,2,3\;\mathrm{mod}\;4\\
2^{k/2-r}D^{-r},&D=0\;\mathrm{mod}\;4
\end{cases},\numberthis\label{eqn:special}
\end{align*}
using the Gauss sum \eqref{eqn:gauss} to evaluate 
$\sum_{\varepsilon_j=0}^{D-1}\langle\varepsilon_i|\hat F_D^\gamma|\varepsilon_i\rangle=\frac{1}{\sqrt{D}}\sum_{\varepsilon_j=0}^{D-1}e^{\pm 2\pi i\varepsilon_j^2/D}$ for $\gamma=\pm1$,
and also using that $\big|\prod_{i\in\frozen_r}\langle\varepsilon_i|\hat F_D^\gamma|\varepsilon_i\rangle\big|=D^{-r}$.
Part (i) of the lemma follows for any $r\ge0$ since for $D=0\;\mathrm{mod}\;4$, \eqref{eqn:special} is bounded by $2^{k/2}\le D^{k/4}$.

(ii) We first consider the case $t_1,t_2\in k\Z$ with $t_1+t_2\ne0\;\mathrm{mod}\;q(k)$. Then for $\gamma_1=t_1/k$ and $\gamma_2=t_2/k$,
\begin{align}\label{eqn:t12k}
\Bigg|\sum_{[\cs{\varepsilon}]\in S_1}\sum_{[\cs{\delta}]\in S_2}\langle\cs{\varepsilon}|\Ba_k^{t_1}|\cs{\delta}\rangle\langle\cs{\delta}|\Ba_k^{t_2}|\cs{\varepsilon}\rangle\Bigg|
&= \Bigg|\sum_{[\cs{\varepsilon}]\in S_1}\sum_{[\cs{\delta}]\in S_2}\prod_{i=1}^k\langle\varepsilon_i|(\hat F_D^\dagger)^{\gamma_1}|\delta_i\rangle\langle\delta_i|(F_D^\dagger)^{\gamma_2}|\varepsilon_i\rangle\Bigg|.
\end{align}
The sums over $[\cs{\varepsilon}]\in S_1$ and $[\cs{\delta}]\in S_2$ are sums over the $k-2r$ free (non-frozen) coordinates $\varepsilon_i$ from $0$ to $D-1$ and over the $k-2r$ free $\delta_i$ from $0$ to $D-1$. If $\delta_i$ and $\varepsilon_i$ are both free, then the term $\langle\varepsilon_i|(\hat F_D^\dagger)^{\gamma_1}|\delta_i\rangle\langle\delta_i|(F_D^\dagger)^{\gamma_2}|\varepsilon_i\rangle$ contributes
\begin{align*}
\left|\sum_{\varepsilon_i=0}^{D-1}\sum_{\delta_i=0}^{D-1}\langle\varepsilon_i|(\hat F_D^\dagger)^{\gamma_1}|\delta_i\rangle\langle\delta_i|(F_D^\dagger)^{\gamma_2}|\varepsilon_i\rangle\right|&=\left|\sum_{\varepsilon_i=0}^{D-1}\langle\varepsilon_i|(\hat F_D^\dagger)^{\gamma_1+\gamma_2}|\varepsilon_i\rangle\right|
\le \begin{cases}
1,&D\text{ odd or }D=2\\
2,&D\ge4\text{ even}
\end{cases},
\end{align*}
using that $\gamma_1+\gamma_2\ne 0\;\mathrm{mod}\;4$ for $D\ge3$, or $\gamma_1+\gamma_2\ne0\;\mathrm{mod}\;2$ for $D=2$, to apply either the Gauss sum \eqref{eqn:gauss} or trace bound on $\hat R_D=(\hat F_D^\dagger)^2$. If either of $\delta_i$ or $\varepsilon_i$ is frozen due to coordinates defining $S_1$ or $S_2$, then the term $\langle\varepsilon_i|(\hat F_D^\dagger)^{\gamma_1}|\delta_i\rangle\langle\delta_i|(F_D^\dagger)^{\gamma_2}|\varepsilon_i\rangle$ contributes at most $1$, by considering a sum over only one of $\varepsilon_i,\delta_i$, or no sum at all (if both are fixed).
There are $k$ values of $i=1,\ldots,k$, so \eqref{eqn:t12k} is bounded as
\begin{align*}
\Bigg|\sum_{[\cs{\varepsilon}]\in S_1}\sum_{[\cs{\delta}]\in S_2}\langle\cs{\varepsilon}|\Ba_k^{t_1}|\cs{\delta}\rangle\langle\cs{\delta}|\Ba_k^{t_2}|\cs{\varepsilon}\rangle\Bigg|
&\le\begin{cases}
1,&D\text{ odd or }D=2\\
2^k,&D\ge4\text{ even}
\end{cases},
\end{align*}
which is \eqref{eqn:trace-partial-2}.

Now we consider $t_1,t_2\not\in k\Z$, with $[t_1+t_2]_k=0$ and $t_1+t_2\ne0\;\mathrm{mod}\;q(k)$.
This part will use the graphs and cycles from the previous lemmas. In order to obtain a more precise bound than \eqref{eqn:trace-partial-p}, we will need to keep track of how many powers of $\hat{F}_D^\dagger$ we accumulate along a cycle leading to an expression like \eqref{eqn:allbutone}.
We note from their definition that $|\cs{\varepsilon}\rangle$ and $|\cs{\delta}\rangle$ each have their $k-\ell$ momentum coordinate basis vectors Fourier transformed. However, because these Fourier-transformed basis vectors show up exactly twice in the inner products in the left side of \eqref{eqn:trace-partial-2}, 
always in the form $|\hat F_D^\dagger\varepsilon_j\rangle\langle\hat F_D^\dagger\varepsilon_j|=\hat F_D^\dagger|\varepsilon_j\rangle\langle\varepsilon_j|\hat F_D$, these Fourier transforms cancel when contracting over such a momentum coordinate basis vector $|\varepsilon_j\rangle$. 
And if $|\varepsilon_j\rangle$ is a frozen coordinate, then extra Fourier transforms do not matter when bounding the contribution by 1.
Therefore when we consider the following graph and cycles as in the previous lemma (with indices on $\varepsilon$ or $\delta$ taken modulo $k$), we can neglect the ``inherent'' Fourier transforms on $|\cs{\varepsilon}\rangle$ or $|\cs{\delta}\rangle$, and only count the ones coming from $\Ba_k^{t_1}$ and $\Ba_k^{t_2}$. 
We first draw the graph and cycles in the case $t_1,t_2\in\intbrr{1:k-1}$ with $t_1+t_2=k$:
\begin{equation}\label{eqn:tikz-p2}
\tp{\begin{tikzpicture}[baseline=(current  bounding  box.center),ampersand replacement=\&]
\node[matrix,right] at (0,0)
{
	\node {$|\cs{\varepsilon}\rangle$:}; \&\node {$\varepsilon_\ell$}; \&\node {$\cdots$}; \&\node {$\varepsilon_{\ell+[t_1]_k+1}$};\&\node {$\varepsilon_{\ell+[t_1]_k}$}; \&\node {$\cdots$}; \& \node {$\varepsilon_{\ell+1}$}; \\
	\node{};\&\\
	\node {$\Ba_k^{t_1}|\cs{\delta}\rangle$:}; \&\node {$\delta_{\ell-[t_1]_k}$}; \&\node{$\cdots$}; \&\node {$\delta_{\ell+1}$};  \&\node {$\hat{F}_D^\dagger\delta_{\ell}$}; \& \node {$\cdots$};\&\node {$\hat{F}_D^\dagger\delta_{\ell-[t_1]_k+1}$}; \\
};
\foreach \pos in {2,4.2,6.2,8.2}{
	\foreach \vert in {0,-2}{
		\draw[yshift=\vert cm, line width=2] (\pos,-.15)--++(0,.4);
	}
	\draw[yshift=-1 cm] (\pos,-.15)--++(0,.4);
}
\foreach \hshift in {2,4.2,6.2,8.2}{
	\draw[xshift=\hshift cm] plot [smooth, tension=1] coordinates {(0,-2.2) (-.3,-1) (0,.2)};
}

\node[matrix,right] at (0,-2)
{
	\node {$|\cs{\delta}\rangle$:}; \&\node {$\delta_{\ell-[t_1]_k}$};  \&\node {$\cdots$}; \&\node{$\delta_{\ell+1}$};\&\node {$\delta_\ell$}; \& \node {$\cdots$}; \& \node {$\delta_{\ell-[t_1]_k+1}$}; \\
	\node{};\&\\
	\node {$\Ba_k^{t_2}|\cs{\varepsilon}\rangle$:}; \&\node {$\hat{F}_D^\dagger\varepsilon_{\ell}$}; \&\node {$\cdots$};  \&\node {$\hat{F}_D^\dagger\varepsilon_{\ell-[t_2]_k+1}$};\&\node {$\varepsilon_{\ell-[t_2]_k}$}; \&\node{$\cdots$};\& \node {$\varepsilon_{\ell+1}$}; \\
};
\node at (9.9,-1) {,};
\end{tikzpicture}}
\end{equation}
In the bottom two rows, we have cyclically shifted both rows by $t_1=[t_1]_k$ to the left to line up with the top two rows, and used that $t_1+t_2=k$. 
We see in each cycle, there is exactly one term $\hat{F}_D^\dagger$, not including the inherent ones on momentum coordinates which we can ignore since they come in canceling pairs. 
Thus each cycle looks like $\langle\varepsilon_j|\delta_{j-[t_1]_k}\rangle\langle\delta_{j-[t_1]_k}|\hat F_D^\dagger|\varepsilon_j\rangle$ or $\langle \varepsilon_j|\hat F_D^\dagger|\delta_{j-[t_1]_k}\rangle\langle\delta_{j-[t_1]_k}|\varepsilon_j\rangle$.
The contribution from each cycle to the sum in \eqref{eqn:trace-partial-2} is then at most
\begin{align}\label{eqn:cycle-contrib}
\max\left(\Bigg|\sum_{\varepsilon_j=0}^{D-1}\langle\varepsilon_j|\hat{F}_D^\dagger|\varepsilon_j\rangle\Bigg|,1\right)=\begin{cases} 1,&D\ne 0\;\mathrm{mod}\;4\\
2^{1/2},&D=0\;\mathrm{mod}\;4
\end{cases},
\end{align} 
using the quadratic Gauss sum \eqref{eqn:gauss} along with the fact that any cycle with frozen variables contributes at most 1.
Since there are $k$ cycles total, raising this to the $k$th power shows \eqref{eqn:trace-partial-2} for the case $t_1,t_2\in\intbrr{1:k-1}$.

Now we consider general $t_1+t_2$. We first restrict to $D\ge3$. It is enough to consider $-2k\le t_1,t_2\le 2k-1$ since this covers a full period $q(k)=4k$. 
Since we consider $t_1,t_2\not\in k\Z$, there are four regimes $\intbrr{-2k+1:-k-1},\intbrr{-k+1:-1},\intbrr{1:k-1}, \intbrr{k+1:2k-1}$ for each of $t_1$ and $t_2$. Each regime corresponds to multiplying the entire second row (for $t_1$) or fourth row (for $t_2$) of \eqref{eqn:tikz-p2} by a different power of $\hat{F}_D^\dagger$ and then considering the shifts $[t_1]_k,[t_2]_k$. For example, if $t_1\in\intbrr{k+1:2k-1}$, then $t_1=k+[t_1]_k$, and $\Ba_k^{t_1}=(\hat{F}_D^\dagger)^{\otimes k}\Ba_k^{[t_1]_k}$, corresponding to adding a factor $\hat{F}_D^\dagger$ to every term in the second row of \eqref{eqn:tikz-p2}. 
We record the total number of $\hat{F}_D^\dagger$ factors (modulo 4) along each cycle in \eqref{eqn:tikz-p2} for each case of $t_1,t_2$ in Table~\ref{tab:tikz-p22} below.
The corresponding $\hat{F}_D^\dagger$ factors for each regime are written above the columns and to the left of the rows of Table~\ref{tab:tikz-p22}. The total power of $\hat{F}_D^\dagger$ is the product of the factors for each regime, plus one factor from the original \eqref{eqn:tikz-p2}. 
\begin{table}[!ht]
\[
\def\arraystretch{1.3}
\begin{array}{rr|cccc}
\hline\hline
&& (\hat{F}_D^\dagger)^2& (\hat{F}_D^\dagger)^3 & \operatorname{Id}& \hat{F}_D^\dagger\\
&\hbox{\diagbox[height=6mm]{$t_2$}{$t_1$}}
&\intbrr{-2k+1:-k-1}&\intbrr{-k+1:-1}&\intbrr{1:k-1} &\intbrr{k+1:2k-1}\\\hline
(\hat{F}_D^\dagger)^2&\intbrr{-2k+1:-k-1}& \hat{F}_D^\dagger & (\hat{F}_D^\dagger)^2 & (\hat{F}_D^\dagger)^3& \operatorname{Id}\\
(\hat{F}_D^\dagger)^3&\intbrr{-k+1:-1} &(\hat{F}_D^\dagger)^2& (\hat{F}_D^\dagger)^3 & \operatorname{Id} & \hat{F}_D^\dagger\\
\operatorname{Id}&\intbrr{1:k-1}& (\hat{F}_D^\dagger)^3& \operatorname{Id} & \hat{F}_D^\dagger& (\hat{F}_D^\dagger)^2 \\
\hat{F}_D^\dagger&\intbrr{k+1:2k-1}& \operatorname{Id} & \hat{F}_D^\dagger & (\hat{F}_D^\dagger)^2& (\hat{F}_D^\dagger)^3\\\hline\hline
\end{array}
\]
\caption{Tracking the extra powers of $\hat{F}_D^\dagger$ added to the rows of \eqref{eqn:tikz-p2} depending on the values of $t_1,t_2$. The total power of $\hat F_D^\dagger$ is given in the center of the table, and is the product of the factors written to the side of each corresponding regime, plus one factor from the original.}\label{tab:tikz-p22}
\end{table}

For cycles producing in total $(\hat F_D^\dagger)^\gamma$, for $\gamma\in\{1,2,3\}$, the contribution from that cycle to the sum in \eqref{eqn:trace-partial-2} is at most
\begin{align*}
\max\left(\Bigg|\sum_{\varepsilon_j=0}^{D-1}\langle\varepsilon_j|(\hat{F}_D^\dagger)^\gamma|\varepsilon_j\rangle\Bigg|,1\right)\le 
\begin{cases} 1,&D\ge3\text{ odd}\\ 
2,&D\ge4\text{ even}
\end{cases},
\end{align*}
considering both the Gauss sum and the trace of $\hat R_D$. Raising to the $k$th power gives \eqref{eqn:trace-partial-2}.

It remains to consider the four cases where we have $\operatorname{Id}$ in Table~\ref{tab:tikz-p22}.
However, we see from the table that in those cases, we always have $-(k-2)\le t_1+t_2\le k-2$. So if $[t_1+t_2]_k=0$, then we must have $t_1+t_2=0$, which is assumed not to happen in the hypotheses, and so \eqref{eqn:trace-partial-2} holds for $D\ge3$.

In the case of general $t_1+t_2$ when $D=2$, the analysis is simpler since we only need to consider $-k\le t_1,t_2\le k-1$ to cover a whole period. We only need to use
\begin{align}
\max\left(\left|\sum_{\varepsilon_j=0}^{1}\langle\varepsilon_j|\hat{F}_2|\varepsilon_j\rangle\right|,1\right)\le 1,
\end{align}
which gives \eqref{eqn:trace-partial-2} for $D=2$ by the same reasoning as above.
\end{proof}

To apply the previous two lemmas for Proposition~\ref{thm:trace-cancel}, we will use the following cancellation lemma.
\begin{lem}[cancellation]\label{lem:a-cancellation}
Let $[\cs{\varepsilon}],[\cs{\delta}]$ denote $(k,\ell)$-rectangles, and fix $t_1,t_2\in\Z$.
Let $r\in\intbrr{0:\min(\ell,k-\ell)}$, and consider a partition of $\T^2$ into the $D^{2r}$ size $D^{-r}\times D^{-r}$ squares
\begin{align*}
\mathscr S_r&=\left\{\Big[\frac{j}{D^r},\frac{j+1}{D^r}\Big)\times\left[\frac{k}{D^r},\frac{k+1}{D^r}\right)\right\}_{j,k=0}^{D^r-1}.
\end{align*}
Note since $r\le \min(\ell,k-\ell)$, each $(k,\ell)$-rectangle $[\cs{\varepsilon}]$ or $[\cs{\delta}]$ fits within a single square in $\mathscr S_r$.
\begin{enumerate}[(i)]
\item If
\begin{align}\label{eqn:ptrace-bound1}
\max_{S_1\in\mathscr S_r}\Bigg|\sum_{[\cs{\varepsilon}]\in S_1}\langle\cs{\varepsilon}|\Ba_k^{t_1}|\cs{\varepsilon}\rangle
\Bigg|&\le \mathfrak{b},
\end{align}
then
\begin{align}\label{eqn:cancellationp1}
\Bigg|\sum_{[\cs{\varepsilon}]\in\mathcal R_{k,\ell}}\fint_{[\cs{\varepsilon}]}a\cdot\langle\cs{\varepsilon}|\Ba_k^{t_1}|\cs{\varepsilon}\rangle\Bigg|&\le \begin{cases}\|a\|_\infty D^{2r}\mathfrak{b}+\sqrt{2}\|a\|_\Lip D^{k/2-r},&t_1\not\in q(k)\Z\\
N\left|\int_{\T^2}a(\x)\,d\x\right|,&t_1\in q(k)\Z
\end{cases}.
\end{align}

\item If
\begin{align}\label{eqn:ptrace-bound2}
\max_{S_1,S_2\in\mathscr S_r}\Bigg|\sum_{[\cs{\varepsilon}]\in S_1}\sum_{[\cs{\delta}]\in S_2}\langle\cs{\varepsilon}|\Ba_k^{t_1}|\cs{\delta}\rangle
\langle\cs{\delta}|\Ba_k^{t_2}|\cs{\varepsilon}\rangle\Bigg|&\le \mathfrak{b},
\end{align}
then
\begin{align}\label{eqn:cancellation}
\nonumber\Bigg|\sum_{[\cs{\varepsilon}],[\cs{\delta}]\in\mathcal R_{k,\ell}}\fint_{[\cs{\varepsilon}]}a\,\fint_{[\cs{\delta}]}a\;\cdot&\,\langle \cs{\varepsilon}|\Ba_k^{t_1}|\cs{\delta}\rangle\langle\cs{\delta}|\Ba_k^{t_2}|\cs{\varepsilon}\rangle\Bigg|\\
&\le \|a\|_\infty^2 D^{4r}\mathfrak{b}+2\sqrt{2}\|a\|_\Lip\|a\|_\infty D^{k-r}.
\end{align}
\end{enumerate}

\end{lem}
As briefly described near the beginning of Section~\ref{sec:phases}, part (i) of the lemma and its application can be explained intuitively for $t\not\in q(k)\Z$ in terms of Fig.~\ref{fig:phases}. If we partition the $(q,p)$ phase space in Fig.~\ref{fig:phases} into square regions $S\in \mathscr S_r$, then the statement \eqref{eqn:ptrace-bound1}, if $\mathfrak{b}$ is small enough, says that either there are not many nonzero values of $\langle\cs{\varepsilon}|\Ba_k^{t_1}|\cs{\varepsilon}\rangle$, or there are but the phases of $\langle\cs{\varepsilon}|\Ba_k^{t_1}|\cs{\varepsilon}\rangle$ fluctuate on a square $S$ enough to cause some cancellations. The latter situation is visualized in Fig.~\ref{fig:phases} through a lack of large regions of a single color.
For a fixed observable $a:\T^2\to\C$, the value of $a$ on a single $D^{-r}\times D^{-r}$ square $S$ is close to a constant as $r\to\infty$. Integrating $a$ on each square $S$ against the fluctuating or mostly zero $\langle\cs{\varepsilon}|\Ba_k^{t_1}|\cs{\varepsilon}\rangle$, for $[\varepsilon^{(1)}]\in S$, will then be small either due to the cancellations from phases or many zero values. 

\begin{proof}[Proof of Lemma~\ref{lem:a-cancellation}]

(i) First, if $t_1\in q(k)\Z$ then $\Ba_k^{t_1}=\operatorname{Id}$, and the left hand side of \eqref{eqn:cancellationp1} is just
\begin{align*}
\Bigg|\sum_{[\cs{\varepsilon}]\in\mathcal R_{k,\ell}}\fint_{[\cs{\varepsilon}]}a\cdot\langle\cs{\varepsilon}|\cs{\varepsilon}\rangle\Bigg|&=N\Bigg|\int_{\T^2}a(\x)\,d\x\Bigg|.
\end{align*}
For $t_1\not\in q(k)\Z$, we start by estimating
$\left|\fint_{[\cs{\varepsilon}]}a-\fint_{S_1}a\right|\le \sqrt{2}\|a\|_\Lip D^{-r}$ for $[\cs{\varepsilon}]\subseteq S_1$. 
Then we have
\begin{align*}
\begin{aligned}
\Bigg|\sum_{[\cs{\varepsilon}]\in\mathcal R_{k,\ell}}\fint_{[\cs{\varepsilon}]}a\cdot\langle \cs{\varepsilon}|\Ba_k^{t_1}|\cs{\varepsilon}\rangle\Bigg|
&\le \begin{multlined}[t]
\Bigg|\sum_{S_1\in\mathscr S_r}\sum_{[\cs{\varepsilon}]\in S_1}
\fint_{S_1}a\cdot\langle \cs{\varepsilon}|\Ba_k^{t_1}|\cs{\varepsilon}\rangle\Bigg|+\\
\hspace{1cm}+\sqrt{2}\|a\|_\Lip D^{-r}\sum_{[\cs{\varepsilon}]\in\mathcal R_{k,\ell}}|\langle \cs{\varepsilon}|\Ba_k^{t_1}|\cs{\varepsilon}\rangle|.
\end{multlined}
\end{aligned}\numberthis\label{eqn:s1bound}
\end{align*}
Applying the triangle inequality to the sum over $S_1\in\mathscr S_r$ only, and using \eqref{eqn:ptrace-bound1}, gives
\begin{align*}
\Bigg|\sum_{S_1\in\mathscr S_r}\sum_{[\cs{\varepsilon}]\in S_1}
\fint_{S_1}a\cdot\langle \cs{\varepsilon}|\Ba_k^{t_1}|\cs{\varepsilon}\rangle\Bigg|&\le \sum_{S_1\in\mathscr S_r}\left|\fint_{S_1}a\right|\Bigg|\sum_{[\cs{\varepsilon}]\in S_1}
\langle \cs{\varepsilon}|\Ba_k^{t_1}|\cs{\varepsilon}\rangle\Bigg| \le \|a\|_\infty D^{2r}\mathfrak{b}.
\end{align*}
Note it was key that the absolute value is taken on the \emph{outside} of the sum over $[\cs{\varepsilon}]\in S_1$.
The remaining term in \eqref{eqn:s1bound} can be bounded using Proposition~\ref{prop:walsh-powers}(i,ii) for $t_1\not\in q(k)\Z$. If also $t_1\not\in 2k\Z$, this gives
\begin{align*}
\sqrt{2}\|a\|_\Lip D^{-r}\sum_{[\cs{\varepsilon}]\in\mathcal R_{k,\ell}}|\langle \cs{\varepsilon}|\Ba_k^{t_1}|\cs{\varepsilon}\rangle|&\le \sqrt{2}\|a\|_\Lip D^{-r}D^{\eta(t_1)}D^{-\eta(t_1)/2}\le \sqrt{2}\|a\|_\Lip D^{-r}D^{k/2}.
\end{align*}
If $D\ge3$ and $t_1\in 2k+4k\Z$, the same end bound still holds since $\sum_{[\cs{\varepsilon}]\in\mathcal R_{k,\ell}}|\langle \cs{\varepsilon}|\Ba_k^{2k}|\cs{\varepsilon}\rangle|=1$ for $D$ odd and $2^k$ for $D\ge4$ even, both of which are $\le D^{k/2}$.

(ii) We start by applying the triangle inequality twice to get,
\begin{align}
\Bigg|\fint_{[\cs{\varepsilon}]}a\fint_{[\cs{\delta}]}a-\fint_{S_1}a\fint_{S_2}a\Bigg|&\le 2\sqrt{2}\|a\|_\Lip D^{-r}\|a\|_\infty,
\end{align}
for $[\cs{\varepsilon}]\subseteq S_1$ and $[\cs{\delta}]\subseteq S_2$.
Then we can estimate
\begin{align*}
\Bigg|\sum_{[\cs{\varepsilon}],[\cs{\delta}]\in\mathcal R_{k,\ell}}\fint_{[\cs{\varepsilon}]}a&\fint_{[\cs{\delta}]}a\cdot\langle \cs{\varepsilon}|\Ba_k^{t_1}|\cs{\delta}\rangle\langle\cs{\delta}|\Ba_k^{t_2}|\cs{\varepsilon}\rangle\Bigg|\\
&\le \Bigg|\sum_{S_1,S_2\in\mathscr S_r}\sum_{[\cs{\varepsilon}]\in S_1,[\cs{\delta}]\in S_2}
\fint_{S_1}a\fint_{S_2}a\cdot\langle \cs{\varepsilon}|\Ba_k^{t_1}|\cs{\delta}\rangle\langle\cs{\delta}|\Ba_k^{t_2}|\cs{\varepsilon}\rangle\Bigg|+\\
&\hspace{1cm}+2\sqrt{2}\|a\|_\Lip D^{-r}\|a\|_\infty\sum_{[\cs{\varepsilon}],[\cs{\delta}]\in\mathcal R_{k,\ell}}|\langle \cs{\varepsilon}|\Ba_k^{t_1}|\cs{\delta}\rangle\langle\cs{\delta}|\Ba_k^{t_2}|\cs{\varepsilon}\rangle|.
\end{align*}
The first term of the right hand side is bounded by applying the triangle inequality to the sum over $S_1,S_2\in\mathscr S_r$ only, followed by \eqref{eqn:ptrace-bound2}, yielding
\begin{align*}
\sum_{S_1,S_2\in\mathscr S_r}\Bigg|\fint_{S_1}a\fint_{S_2}a \Bigg|\Bigg|\sum_{[\cs{\varepsilon}]\in S_1,[\cs{\delta}]\in S_2}\langle \cs{\varepsilon}|\Ba_k^{t_1}|\cs{\delta}\rangle\langle\cs{\delta}|\Ba_k^{t_2}|\cs{\varepsilon}\rangle\Bigg|
&\le (D^{2r})^2\|a\|_\infty^2 \mathfrak b.\numberthis
\end{align*}
For the second term, applying the Cauchy--Schwarz inequality gives 
\begin{align*}
\sum_{[\cs{\varepsilon}],[\cs{\delta}]\in\mathcal R_{k,\ell}}&|\langle \cs{\varepsilon}|\Ba_k^{t_1}|\cs{\delta}\rangle\langle\cs{\delta}|\Ba_k^{t_2}|\cs{\varepsilon}\rangle| \\
&\le \sum_{[\cs{\varepsilon}]\in\mathcal R_{k,\ell}}
\left(\sum_{[\cs{\delta}]\in\mathcal R_{k,\ell}}|\langle\cs{\varepsilon}|\Ba_k^{t_{1}}|\cs{\delta}\rangle|^2\right)^{1/2}\left(\sum_{[\cs{\delta}]\in\mathcal R_{k,\ell}}|\langle\cs{\delta}|\Ba_k^{t_2}|\cs{\varepsilon}\rangle|^2\right)^{1/2}\\
&\le D^k\cdot 1\cdot 1=D^k,\numberthis
\end{align*}
using that $\{|\cs{\delta}\rangle\}$ forms an orthonormal basis and $\Ba_k^t$ is unitary. This gives Lemma~\ref{lem:a-cancellation}(ii).
\end{proof}

\begin{proof}[Proof of Proposition~\ref{thm:trace-cancel}]
Proposition~\ref{thm:trace-cancel} follows from Lemma~\ref{lem:a-cancellation}, using Lemma~\ref{lem:trace-part} to obtain the bound $\mathfrak{b}=G_{D,k}(t)=D^{\gcd([t]_k,k)}$, and using Lemma~\ref{lem:trace-part2} to obtain the better estimates for $t_1=\pm k$ and $[t_1+t_2]_k=0$ with $t_1+t_2\ne0\;\mathrm{mod}\;q(k)$.
\end{proof}

\section{Individual variance and convergence}\label{sec:tr-ind}

In this section, we apply the phase cancellation lemmas from Section~\ref{sec:phases} to prove that for $\min(\ell,k-\ell)\ge 3\log_Dk$,
\begin{align}
\E_\omega F_j&=\begin{cases}o(1),&D\ne 4\\(-1)^{\alpha_j}\langle a_0\rangle_{\mathcal B^{(4)}_{0,2}}+o(1),&D=4\end{cases},\label{eqn:mean0}\\
\E_\omega F_j^2&=\begin{cases}V(a)+o(1),&D\ne 4\\
V(a)+\langle a_0\rangle_{\mathcal B^{(4)}_{0,2}}^2+o(1),&D=4
\end{cases}, \label{eqn:var2} 
\end{align}
which imply \eqref{eqn:center} and \eqref{eqn:e-var} of Theorem~\ref{thm:ind}. 
Once we have established the convergence of individual $\E_\omega F_j^2$, observe that if all the $F_j$ were independent, then the convergence of the quantum variance \eqref{eqn:var} would follow immediately from e.g. Markov's inequality and a constant bound on $\E_\omega F_j^4$.
Since the $F_j$ are not independent due to orthogonality requirements of the eigenvectors, we will do some explicit expectation value calculations in Section~\ref{subsec:conv-var} to prove the desired convergence of the quantum variance.

Recall from Section~\ref{subsec:weingarten} that
\begin{align*}
\E F_j &= \frac{q(k
)}{\sqrt{N}}\Tr(\Opkl(a_0)P_\jj)(1+o(1)),\\
\E F_j^2&= \frac{q(k)^2}{N}\left[(\Tr \Opkl(a_0) P_\jj)^2 + \Tr(\Opkl(a_0) P_\jj \Opkl(a_0)P_\jj)\right](1+o(1)),
\end{align*}
for $\alpha=\alpha_j$ the corresponding eigenspace index.
Then it is enough to show the following.

\begin{thm}\label{thm:trace}
Suppose
$\min(\ell,k-\ell)\ge 3\log_D k$.
Then for any $\jj=0,\ldots,q(k)-1$, as $k\to\infty$,
\begin{align}\label{eqn:trap}
\Tr(\Opkl(a_0)P_\jj)&=\begin{cases}
o\Big(\frac{\sqrt{N}}{q(k)}\Big),&D\ne 4\\
\frac{\sqrt{N}}{q(k)}(-1)^{\alpha}\langle a_0\rangle_{\mathcal B^{(4)}_{0,2}}+o\Big(\frac{\sqrt{N}}{q(k)}\Big),&D=4
\end{cases},
\end{align}
and
\begin{align}\label{eqn:trapap}
\Tr(\Opkl(a)P_\jj \Opkl(a)P_\jj)&=\frac{N}{q(k)^2}\left[\sum_{t=-\infty}^\infty \int_{\T^2}a(\x)a(B^t\x)\,d\x +\mathbf{1}_{D\ge3}\int_{\T^2}a(\x)a(B^tR\x)\,d\x+o_a(1)\right],
\end{align}
where $R$ is the map $(q,p)\mapsto(1-q,1-p)$. All error terms are uniform in the eigenspace index $\jj=0,\ldots,q(k)-1$.
\end{thm}

\begin{proof}
First we will prove \eqref{eqn:trap}. By \eqref{eqn:ppoly0} we can expand
\begin{align}
\Tr(\Opkl(a_0)P_\jj)&=\frac{1}{q(k)}\sum_{\substack{t=-q(k)/2\\t\ne0}}^{q(k)/2-1}\Tr(\Opkl(a_0)\Ba_k^t)e^{2\pi i\jj t/q(k)},
\label{eqn:tr-op}
\end{align}
where we used that $\int_{\T^2}a_0=0$ to skip the $t=0$ term.
We consider $D\ge3$, for which $q(k)=4k$, since the proof for $D=2$ will just be a simpler version of the same argument.
For $0<Q<k$ to be determined, let $\mathcal T_Q:=\intbrr{-2k+1:-2k+Q}\cup\intbrr{-Q:Q}\cup\intbrr{2k-Q:2k-1}$, which corresponds to when the matrices $\Ba_k^t$ are most sparse (Fig.~\ref{fig:dense}), though excluding $t=-2k$ which we will handle later.
Expanding the trace in the $(k,\ell)$-coherent state basis and applying a simple absolute value bound using Proposition~\ref{prop:walsh-powers} gives
\begin{align*}
\left|\frac{1}{q(k)}\sum_{\substack{t\in \mathcal T_Q\\t\ne0}}\sum_{[\cs{\varepsilon}]\in\mathcal R_{k,\ell}}\langle\cs{\varepsilon}|\Ba_k^t|\cs{\varepsilon}\rangle e^{2\pi i\jj t/q(k)}\fint_{[\cs{\varepsilon}]}a_0\right|
&\le \frac{2}{q(k)}\sum_{\substack{t=-Q\\t\ne0}}^{Q}D^{\eta(t)}D^{-\eta(t)/2} \|a_0\|_\infty\\
&\le \frac{C\|a_0\|_\infty}{q(k)}D^{Q/2}.\numberthis\label{eqn:Qbound}
\end{align*}
Our target bound is lower order than $D^{k/2}/q(k)$, so we can safely take up to e.g. $Q=3k/4$.
Outside this time, we will consider the phases of the matrix elements of $\Ba_k^t$. 

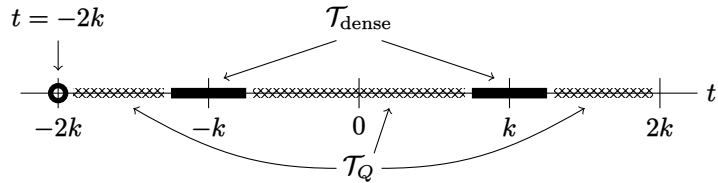
\begin{figure}[htb]
\tp{\begin{tikzpicture}
\draw(-4.5,0)--(4.5,0) node[right] {$t$};
\def\h{.2cm};
\draw (-4,\h)--(-4,-\h) node[below] {$-2k$};
\draw (4,\h)--(4,-\h) node[below] {$2k$};
\draw (-2,\h)--(-2,-\h) node[below] {$-k$};
\draw (2,\h)--(2,-\h) node[below] {$k$};
\draw (0,\h)--(0,-\h) node[below] {$0$};
\draw[line width=2] (-4,0) circle (1mm);
\draw[line width=4] (-2.5,0)--(-1.5,0);
\draw[line width=4] (1.5,0)--(2.5,0);
\node at (0,1) {$\mathcal T_\mathrm{dense}$};
\draw[->] (.3,.7)--(1.8,.2);
\draw[->] (-.3,.7)--(-1.8,.2);
\node at (-4,1) {$t=-2k$};
\draw[->] (-4,.7)--(-4,.3);
\def\sh{.6mm}
\fill[pattern=crosshatch] (-3.8,-\sh) rectangle (-2.6,\sh);
\fill[pattern=crosshatch] (-1.4,-\sh) rectangle (1.4,\sh);
\fill[pattern=crosshatch] (2.6,-\sh) rectangle (3.9,\sh);
\node at (0,-1) {$\mathcal T_Q$};
\draw[->] (.2,-.8)--(.4,-.2);
\draw[->] (.3,-1) to[out=0, in=210] (3,-.2);
\draw[->] (-.3,-1) to[out=180, in=330] (-3,-.2);
\end{tikzpicture}}
\caption{Three regions of times $t$: the set $\mathcal T_Q$ when $\Ba_k^t$ is sparse (and $t\ne -2k$), the set $\mathcal T_\mathrm{dense}$ where $\Ba_k^t$ is dense, and the time $t=-2k$.}\label{fig:dense}
\end{figure}

It remains to consider the times $t\in\intbrr{-5k/4:-3k/4}\cup\intbrr{3k/4:5k/4}=:\mathcal T_\mathrm{dense}$, along with $t=-2k$ (Fig.~\ref{fig:dense}). The times $\mathcal T_\mathrm{dense}$ are where $\Ba_k^t$ is a ``dense'' matrix.
We want to show that for $t\in\mathcal T_\mathrm{dense}$, the phases cancel enough to make their contribution to \eqref{eqn:tr-op} small.
For this, we will use Proposition~\ref{thm:trace-cancel}, which states for $t\not\in q(k)\Z$ and $r\in\intbrr{0:\min(\ell,k-\ell)}$,
\begin{align}\label{eqn:thm42}
|\Tr(\Opkl(a)\Ba_k^{t})|\le
\|a\|_\infty D^{2r}G_{D,k}(t)+\|a\|_\mathrm{Lip}\sqrt{2}D^{k/2-r},
\end{align}
for
\begin{align}
G_{D,k}(t)&=\begin{cases}
D^{\operatorname{gcd}([t]_k,k)},&\text{ any }t\not\in q(k)\Z\\
D^{k/4},&t=\pm k
\end{cases}.
\end{align}
\begin{itemize}
\item For $t\in\intbrr{-5k/4:-k-1}\cup\intbrr{3k/4:k-1}$, we have $[t]_k\in\intbrr{3k/4:k-1}$. Then $\gcd([t]_k,k)\le [t]_k/3$, since the gcd cannot be $[t]_k$ or $[t]_k/2$, as they are too large: if the gcd were $x=[t]_k/2$, then $2x=[t]_k<k$, but $3x\ge 9k/8>k$, so $x$ does not divide $k$. Similarly if $x=[t]_k$, then $2x\ge 3k/2>k$.

\item For $t\in\intbrr{-k+1:-3k/4}\cup\intbrr{k+1:5k/4}$, we have $[t]_k\in\intbrr{1:k/4}$, and so $\gcd([t]_k,k)\le [t]_k\le k/4$.
\end{itemize}
Thus for all $t\in\mathcal T_\mathrm{dense}$,
\begin{align*}
G_{D,k}(t)&\le D^{k/3},
\end{align*}
which by \eqref{eqn:thm42} gives the bound
\begin{align}
\left|\frac{1}{q(k)}\sum_{t\in\mathcal T_\mathrm{dense}}\Tr(\Opkl(a_0)\Ba_k^t)e^{2\pi i\jj t/q(k)}\right|&\le \frac{C}{q(k)}k\left[\|a\|_\infty D^{2r+k/3}+\|a\|_\Lip D^{k/2-r}\right]=o\left(\frac{D^{k/2}}{q(k)}\right),
\end{align}
if we take $r=3\log_Dk$.

The only remaining time for $D\ge3$ is $t=-2k$. (In the case $D=2$, the proof is the same as the above, but only considering times $t=-q(k)/2=-k$ to $t=q(k)/2-1=k-1$, which are split into $\mathcal T_Q$ and $\mathcal T_\mathrm{dense}$ but with no need to consider the separate time $t=-2k$.) Equation~\eqref{eqn:tr-op} for $D\ge3$ at this point has become
\begin{align}\label{eqn:tr2k}
\Tr(\Opkl(a_0)P_\jj)&=o\left(\frac{D^{k/2}}{q(k)}\right)+\frac{1}{q(k)}\Tr(\Opkl(a_0)\Ba_k^{-2k})(-1)^\jj.
\end{align}
Using the previous evaluation of $\Tr(\Opkl(a_0)\Ba_k^{-2k})$ in \eqref{eqn:trace2k-final}, we obtain \eqref{eqn:trap}.

To prove \eqref{eqn:trapap}, we again start with the expression for $P_\jj$ in \eqref{eqn:ppoly0}, giving
\begin{align}\label{eqn:trap2-expand}
\Tr(\Opkl(a)P_\jj \Opkl(a)P_\jj)&=\frac{1}{q(k)^2}\sum_{t_1,t_2=-q(k)/2}^{q(k)/2-1}\Tr(\Opkl(a)\Ba_k^{t_1}\Opkl(a)\Ba_k^{t_2})e^{2\pi i\jj(t_1+t_2)/q(k)}.
\end{align}
By \eqref{eqn:inf-int} of Lemma~\ref{lem:intsum}, which is independent of $\jj$,
we know that the terms when $t_1+t_2=0\;\mathrm{mod}\;q(k)$ give the desired leading order term $\frac{N}{q(k)^2}V(a)$. 
Therefore, we just need to show the sum over $t_1,t_2$, with $t_1+t_2\ne0\;\mathrm{mod}\;q(k)$, is subleading order.
We note that trying a simple absolute value bound does not work: if $t_1=t_2=k$, then $\sum_{[\cs{\varepsilon}],[\cs{\delta}]}|\langle\cs{\varepsilon}|\Ba_k^k|\cs{\delta}\rangle|^2=D^k$ is already too large to be subleading order.

We instead apply Proposition~\ref{thm:trace-cancel}. The function $G_{D,k}(t_1+t_2)$ in \eqref{eqn:tr-bound2} of Proposition~\ref{thm:trace-cancel} can be taken as
\begin{align}
G_{D,k}(t_1+t_2)&\le \begin{cases}
D^{\operatorname{gcd}([t_1+t_2],k)},&t_1+t_2\ne 0\;\mathrm{mod}\;k\\
D^{k/2},&t_1+t_2=0\;\mathrm{mod}\;k,\;t_1+t_2\ne 0\;\mathrm{mod}\;q(k)
\end{cases},
\end{align}
since \eqref{eqn:tr-boundp2} is always $\le D^{k/2}$.
Additionally, if $[t_1+t_2]_k\ne 0$, then $\gcd([t_1+t_2]_k,k)<k$ and so must be $\le k/2$. 
Proposition~\ref{thm:trace-cancel} then implies
\begin{align}\label{eqn:tr2-error}
\frac{1}{q(k)^2}\sum_{\substack{t_1,t_2=-q(k)/2\\ t_1+t_2\ne 0\;\mathrm{mod}\;q(k)}}^{q(k)/2-1}|\Tr(\Opkl(a)\Ba_k^{t_1}\Opkl(a)\Ba_k^{t_2})|&\le C\|a\|_\infty^2D^{4r}D^{k/2}+C\|a\|_\mathrm{Lip}\|a\|_\infty D^{k-r}.
\end{align}
Taking $r=3\log_Dk$ and using that $q(k)=2k$ or $4k$, we see \eqref{eqn:tr2-error} is $o(D^k/q(k)^2)$.
\end{proof}

\subsection{Quantum variance convergence}\label{subsec:conv-var}

To show convergence in probability of the (scaled) quantum variance to $V(a)+\oneb_{D=4}\langle a_0\rangle_{\mathcal B^{(4)}_{0,2}}^2$, note that by \eqref{eqn:var2} and Markov's inequality,
\begin{align*}
\P_\omega\Bigg[\bigg|\frac{1}{N}\sum_{j=1}^N|F_j|^2&-\Big(V(a)+\oneb_{D=4}\langle a_0\rangle_{\mathcal B^{(4)}_{0,2}}^2\Big)\bigg|>\epsilon\Bigg]\\
&\le \frac{1}{\epsilon^2}\E_\omega\left[\Bigg(\frac{1}{N}\sum_{j=1}^N|F_j|^2-\Big(V(a)+\oneb_{D=4}\langle a_0\rangle_{\mathcal B^{(4)}_{0,2}}^2\Big)\Bigg)^2\right]\\
&=\frac{1}{\epsilon^2}\Bigg[\frac{1}{N^2}\sum_{\substack{j,k=1\\j\ne k}}^N\E_\omega[|F_j|^2|F_k|^2]+
\frac{1}{N^2}\sum_{j=1}^N\E_\omega[|F_j|^4]-\Big(V(a)+\oneb_{D=4}\langle a_0\rangle_{\mathcal B^{(4)}_{0,2}}^2\Big)^2+o_a(1)\Bigg].\numberthis\label{eqn:markov}
\end{align*}
Since $F_j$ and $F_k$ are not independent if they correspond to eigenvectors in the same eigenspace, we compute $\E_\omega[|F_j|^2|F_k|^2]$ 
using the Weingarten calculus for integration on the unitary group $\mathcal{U}(d)$. 

Recalling \eqref{eqn:F-unitary}, if $F_j$ and $F_k$ correspond to eigenvectors in $E_\jj$, which has dimension $d_\jj$, then we can write,
\begin{align}\label{eqn:wf}
F_j&\dsim\sqrt{N}\langle u^{(1)}|\Lambda_\jj^\dagger \Opkl(a_0)\Lambda_\jj|u^{(1)}\rangle,\quad
F_k\overset{d}{\sim}\sqrt{N}\langle u^{(2)}|\Lambda_\jj^\dagger\Opkl(a_0)\Lambda_\jj|u^{(2)}\rangle,
\end{align}
where $u^{(1)}$ and $u^{(2)}$ are two columns from a $d_\jj\times d_\jj$ Haar random unitary matrix and $\Lambda_\jj$ is an $N\times d_\jj$ matrix whose columns form an orthonormal basis of the eigenspace $E_\jj$.
As a consequence of Theorem~\ref{thm:trace}, we have:
\begin{cor}\label{lem:trM}
For an eigenspace index $\jj$, let $M=\Lambda_\jj^\dagger \Opkl(a_0)\Lambda_\jj$. Suppose $\min(\ell,k-\ell)\ge 3\log_D k$. Then for any $D\ge2$, as $k\to\infty$,
\begin{align}
\Tr M&=\begin{cases}
o\left(\frac{\sqrt{N}}{q(k)}\right),&D\ne 4\\
\frac{\sqrt{N}}{q(k)}(-1)^{\alpha_j}\langle a_0\rangle_{\mathcal B^{(4)}_{0,2}}+o\left(\frac{\sqrt{N}}{q(k)}\right),&D=4
\end{cases},\label{eqn:trM1}\\
\Tr M^2&=\frac{N}{q(k)^2}(V(a)+o(1)),\label{eqn:trM2}\\
\Tr M^p&=o(N^{p/2}/q(k)^{p}),\quad p=3,4,\ldots,
\end{align}
with error terms uniform over the eigenspace index $\jj$.
\end{cor}
\begin{proof}
Equations \eqref{eqn:trM1} and \eqref{eqn:trM2} follow immediately from Theorem~\ref{thm:trace} and cyclicity of the trace.

For $\Tr M^p$ with $p\ge3$, in view of Section~\ref{sec:off-diag} we will consider slightly more general products, where we have projections onto different eigenspaces indexed by $\jj(i)\in\intbrr{0:q(k)-1}$. We have
\begin{align}\label{eqn:tr-prod}
\Tr\left[\prod_{i=1}^p(\Opkl(a)P_{\jj(i)})\right]=O\left(\frac{N}{q(k)}\right),
\end{align}
since for a positive operator $P\ge0$ and matrix $B$, 
\begin{align*}
|\Tr(BP)|&\le (\Tr P) \|B\|, 
\end{align*}
where $\|B\|$ is the operator norm. Applying this once to \eqref{eqn:tr-prod} to pull out a projection $P_\jj$, and using $\|\Opkl(a)\|\le\|a\|_\infty$ and $\Tr P_{\jj(i)}=\frac{N}{q(k)}(1+o(1))$, gives the bound $\frac{N}{q(k)}(1+o(1))\|a_0\|_\infty^m$, which is
 $o(N^{p/2}/q(k)^p)$ for $p\ge3$. 
\end{proof}

\begin{proof}[Proof of Theorem~\ref{thm:mat}'s Eq.~\eqref{eqn:var} and \eqref{eqn:var4}]
Let $M=\Lambda_\jj^\dagger \Opkl(a_0)\Lambda_\jj$ as in the above lemma.
Recall we are bounding \eqref{eqn:markov}. 
Let $d=d_\jj$, which we recall from \eqref{eqn:dim} is $\frac{N}{q(k)}(1+o(1))$. Using \eqref{eqn:wf} we have for $j\ne k$,
\begin{align}\label{eqn:22expansion}
\E_\omega[|F_j|^2|F_k|^2] &= N^2\sum_{i_1,i_2,i_3,i_4=1}^{d_\jj}\sum_{i_1'i_2'i_3'i_4'=1}^{d_\jj}M_{i_1'i_1}\ol{M}_{i_2i_2'}M_{i_3'i_3}\ol{M}_{i_4i_4'}\E\Big[\ol{u^{(1)}_{i_1'}}u^{(1)}_{i_1}u^{(1)}_{i_2}\ol{u^{(1)}_{i_2'}}\ol{u^{(2)}_{i_3'}}u^{(2)}_{i_3}u^{(2)}_{i_4}\ol{u^{(2)}_{i_4'}}\Big].
\end{align}
To compute the expectation, we can use the Weingarten calculus Theorem~\ref{thm:weingarten} with $n=4$.
We only care about leading order in $d$ for which \eqref{eqn:weingarten-asymptotics} will be sufficient.
It will turn out there is only one pair $(\sigma,\tau)$ for $D\ne4$, and two pairs for $D=4$, which can contribute to the leading order.
First, in terms of the indices in \eqref{eqn:weingarten}, we have $j_1=j_2=j_1'=j_2'=1$, which is a different value from $j_3=j_4=j_3'=j_4'=2$. Therefore when considering the permutation $\tau\in S_4$, the only nonzero contributions come from $\tau$ that map $\{1,2\}$ to itself, and $\{3,4\}$ to itself. This gives only 4 possible values of $\tau$ (corresponding to $S_2\times S_2$); written in cycle notation these are $\tau\in\{\operatorname{Id}, (12),(34),(12)(34)\}$.

Second, we can also reduce the number of possible $\sigma$ by considering the contributions from $M$. 
Using that $M$ is self-adjoint since $a_0$ is real, evaluating the expectation in \eqref{eqn:22expansion} using Theorem~\ref{thm:weingarten} gives
\begin{align}\label{eqn:22exp2}
\E_\omega[|F_j|^2|F_k|^2] &=N^2\sum_{\tau\in S_2\times S_2}\sum_{\sigma\in S_4}\left(\sum_{i_1,i_2,i_3,i_4=1}^{d_\jj}M_{i_{\sigma^{-1}(1)}i_1}M_{i_{\sigma^{-1}(2)}i_2}M_{i_{\sigma^{-1}(3)}i_3}M_{i_{\sigma^{-1}(4)}i_4}\right)\Wg(\tau\sigma^{-1}).
\end{align}
The cycle shape of $\sigma$ determines the value of the sum over entries of $M$, which is enclosed above in large parentheses. The possible values are $\Tr(M)^4$ for $\sigma$ with cycle shape $[1,1,1,1]$ (the identity), $\Tr(M)^2\Tr(M^2)$ for $[2,1,1]$, $(\Tr M^2)^2$ for $[2,2]$, $\Tr(M)\Tr(M^3)$ for $[3,1]$, and $\Tr(M^4)$ for $[4]$. 
By Corollary~\ref{lem:trM}, for $D\ne4$ all of the terms except for the one from shape $[2,2]$ are $o(N^2/q(k)^4)$.
For $[2,2]$, the term is $(\Tr M^2)^2=\frac{N^2}{q(k)^4}(V(a)^2+o(1))$.
For $D=4$, we must consider $\sigma$ with shapes $[2,2]$, $[1,1,1,1]$, or $[2,1,1]$ as well.

Next, by \eqref{eqn:weingarten-asymptotics} the Weingarten function $\Wg(\tau\sigma^{-1})$ has largest order $1/d^4$ only for $\tau\sigma^{-1}=\operatorname{Id}$; the rest are $O(d^{-5})$.
So for $D\ne4$, the largest order term over all $(\tau,\sigma)$ in \eqref{eqn:22exp2} occurs when $\sigma$ has cycle shape $[2,2]$ and $\tau=\sigma$. This occurs only when $\sigma=\tau=(12)(34)$.
For $D=4$, the largest order terms occur when $\sigma$ has shape $[2,2]$, $[1,1,1,1]$, or $[2,1,1]$, and $\tau=\sigma$. This occurs for any of the four choices of $\sigma=\tau\in S_2\times S_2$.
Thus we have
\begin{align*}
\E_\omega[|F_j|^2|F_k|^2]&= N^2\left[\frac{(\Tr M^2)^2}{d^4}+\frac{\oneb_{D=4}\left[(\Tr M)^4+2(\Tr M)^2\Tr(M^2)\right]}{d^4}+o\left(\frac{N^2}{q(k)^4d^4}\right)\right]\\
&=V(a)^2+\oneb_{D=4}\Big[\dfoura^4+2\dfoura^2 V(a)\Big]+ o(1)\\
&=\Big(V(a)+\oneb_{D=4}\dfoura^2\Big)^2+o(1),\numberthis\label{eqn:fcov}
\end{align*}
using Corollary~\ref{lem:trM} and the eigenspace dimensions \eqref{eqn:dim} for the second equality.

For \eqref{eqn:markov}, it remains to evaluate $\E[|F_j|^4]$. From similar considerations using Weingarten calculus and Corollary~\ref{lem:trM} as above (this time with any $\tau\in S_4$ and $\sigma=\tau$), we obtain
\begin{align}
\E_\omega[|F_j|^4] &= O\left(N^2\frac{N^2}{q(k)^4}\frac{1}{d^4}\right)\le C. 
\end{align}
Using \eqref{eqn:fcov} and the above inequality in \eqref{eqn:markov} yields
\begin{multline}
\P_\omega\left[\bigg|\frac{1}{N}\sum_{j=1}^N|F_j|^2-\Big(V(a)+\oneb_{D=4}\langle a_0\rangle_{\mathcal B^{(4)}_{0,2}}^2\Big)\bigg|>\epsilon\right]\\
\le\frac{1}{\epsilon^2}\left[\Big(1-\frac{1}{N}\Big)\Big(V(a)+\oneb_{D=4}\langle a_0\rangle_{\mathcal B^{(4)}_{0,2}}^2\Big)^2+\frac{1}{N}\E[|F_j|^4]-\Big(V(a)+\oneb_{D=4}\langle a_0\rangle_{\mathcal B^{(4)}_{0,2}}^2\Big)^2+o(1)\right],
\end{multline}
which goes to zero as $N\to\infty$.
This finishes the proof of \eqref{eqn:var} and \eqref{eqn:var4}.
\end{proof}

\section{Higher moments and asymptotic distribution}\label{sec:exp}

In this section, we finish the proof of Theorem~\ref{thm:ind}, and use it to prove the weak convergence in probability to a Gaussian distribution of Theorem~\ref{thm:mat}. For the higher moments it will be easier to work with the recentered matrix element fluctuations
\begin{align*}
\tilde F_j^{(k)}:=\begin{cases}F_j^{(k)},&D\ne4\\F_j^{(k)}-(-1)^{\jj_j}\langle a_0\rangle_{\mathcal{B}^{(4)}_{0,2}},&D=4\end{cases},
\end{align*} 
as defined in \eqref{eqn:tilde}.

\subsection{Higher moments}\label{subsec:moments}
In order to compute the higher moments of $\tilde F_j^{(k)}$, as usual we will use the Weingarten calculus \eqref{eqn:weingarten-asymptotics}. 
Since we work with $\tilde F_j^{(k)}$ instead of $F_j^{(k)}$, it will be useful to define for a chosen eigenspace index $\jj$ the $d_\jj\times d_\jj$ matrix
\begin{align}\label{eqn:M}
\tilde M&=\begin{cases}
\Lambda_\jj^\dagger \Opkl(a_0)\Lambda_\jj,&D\ne4\\
\Lambda_\jj^\dagger \Opkl(a_0)\Lambda_\jj-\frac{1}{\sqrt{N}}(-1)^{\jj}\langle a_0\rangle_{\mathcal{B}^{(4)}_{0,2}},&D=4
\end{cases},
\end{align}
where $\frac{1}{\sqrt{N}}(-1)^{\jj_j}\langle a_0\rangle_{\mathcal{B}^{(4)}_{0,2}}$ denotes a scalar multiple of the $d_\jj\times d_\jj$ identity matrix.
Since $a_0$ is real-valued, $\tilde M$ is self-adjoint.
Then $\tilde F_j^{(k)}$ is distributed as
\begin{align*}
\tilde F_j^{(k)}\dsim \sqrt{N}\langle u|\tilde M|u\rangle,
\end{align*}
for $u$ a Haar random vector.

\begin{proof}[Proof of the remainder of Theorem~\ref{thm:ind}]

For $\jj=\jj_j$ the eigenspace index corresponding to $j$, let $\tilde M$ be the matrix defined in \eqref{eqn:M}. 
By Weingarten calculus,
\begin{align*}
\E_\omega[(\tilde F_j^{(k)})^p]&= N^{p/2}\sum_{\substack{i_1,\ldots,i_p=1\\i_1',\ldots,i_p'=1}}^d \tilde M_{i_1'i_1}\cdots \tilde M_{i_p'i_p}\E[\bar u_{i_1'}\cdots\bar u_{i_p'}u_{i_1}\cdots u_{i_p}]\\
&= N^{p/2}\!\!\!\!\sum_{\substack{\sigma\text{ cycle shape }[c_1,\ldots,c_m]\\c_1+\cdots+c_m=p}}\!\!\!\! \#\{\sigma\in S_p\text{ of cycle shape }[c_1,\ldots,c_m]\}\Tr(\tilde M^{c_1})\cdots\Tr(\tilde M^{c_m})\left(\sum_{\tau\in S_p}\Wg(\tau)\right).
\end{align*}
Using \eqref{eqn:weingarten-asymptotics} on the leading order asymptotics of the Weingarten function, we see that
\begin{align}
\sum_{\tau\in S_p}\Wg(\tau)&=\frac{1}{d^p}(1+o(1)),
\end{align}
with the leading order term coming from $\Wg(\operatorname{Id})$.
For $\E[(\tilde F_j^{(k)})^p]$ to be non-decaying, we need to identify when $\Tr(\tilde M^{c_1})\cdots\Tr(\tilde M^{c_m})$ is at least order $d^p/N^{p/2}=N^{p/2}/q(k)^p(1+o(1))$.
If $p$ is even and $c_1=\cdots=c_m=2$, then by Corollary~\ref{lem:trM}, we have
\begin{align*}
(\Tr(\tilde M^2))^{p/2}&=V(a)^{p/2}\frac{N^{p/2}}{q(k)^p}(1+o(1)).
\end{align*}
However, if any $c_j\ne 2$, then Corollary~\ref{lem:trM} shows that $\Tr(\tilde M^{c_1})\cdots\Tr(\tilde M^{c_m})$ must be $o(N^{p/2}/q(k)^p)$, using the constraint $c_1+\cdots+c_m=p$.
Therefore if $p$ is odd, then $\E[(\tilde F_j^{(k)})^p]=o(1)$, and if $p$ is even,
\begin{align*}
\E_\omega[(\tilde F_j^{(k)})^p]&=\#\{\tau\in S_p\text{ of cycle type }[2,\ldots,2]\}V(a)^{p/2}+o(1).
\end{align*}
The number of cycles in $S_p$ of shape $[2,\ldots,2]$ is
\begin{align*}
\frac{p!}{2^{p/2}(p/2)!}=(p-1)!!,
\end{align*}
which is exactly the $p$th moment (for $p$ even) of a standard normal random variable, proving \eqref{eqn:moments}.

Equations~\eqref{eqn:clt} and \eqref{eqn:clt4} then follow by the method of moments (e.g. \cite[\S3.3]{Durrett}).
\end{proof}

\subsection{Asymptotic Gaussian distribution}\label{subsec:gaussian}
To show weak convergence in probability of the empirical distribution $\mu_k=\frac{1}{N}\sum_{j=1}^N\delta_{F_j}$, instead of further Weingarten calculus we can show that for each $t\in\R$ the characteristic function $\operatorname{chf}(t)$ of $\mu_k$ converges in probability to that of the appropriate limiting distribution.
For $D\ne4$, the desired limiting distribution is $g\sim\RN(0,V(a))$, which has characteristic function $e^{-t^2V(a)/2}$.
For $D=4$, it will be easiest to consider convergence of the empirical measures for the recentered matrix element fluctuations $\tilde F_j$, separated by parity of their eigenspace index $\alpha_j$. Then for $D=4$ we will show the weak convergences in probability
\begin{align}\label{eqn:tildeFjhalf}
\frac{1}{|\{j:\alpha_j\text{ even}\}|}\sum_{\substack{j=1\\\alpha_j\text{ even}}}^N\delta_{\tilde F_j}&\xrightarrow{w,\P}\RN(0,V(a)),\quad\text{and}\quad
\frac{1}{|\{j:\alpha_j\text{ odd}\}|}\sum_{\substack{j=1\\\alpha_j\text{ odd}}}^N\delta_{\tilde F_j}\xrightarrow{w,\P}\RN(0,V(a)).
\end{align}
By separating into these two sets of eigenspaces, $\tilde F_j$ is just a constant shift of $F_j$ within each set of eigenspaces. Therefore, \eqref{eqn:tildeFjhalf} implies convergence of the empirical distribution for $\{F_j\}$ with $\alpha_j$ even to $\RN(\langle a_0\rangle_{\mathcal{B}^{(4)}_{0,2}},V(a))$, and that for $\{F_j\}$ with $\alpha_j$ odd to $\RN(-\langle a_0\rangle_{\mathcal{B}^{(4)}_{0,2}},V(a))$.
Since the dimension of each eigenspace has the same leading order $N/q(k)$ by \eqref{eqn:dim}, then \eqref{eqn:tildeFjhalf} will imply that $\mu_k$ converges to the mixture of the two Gaussians \eqref{eqn:d4pdf}.
We will thus mainly consider empirical distributions of $\tilde F_j$ rather than of $F_j$ in this section.

By the moment calculations of Theorem~\ref{thm:ind}, in which the error terms can be taken uniform in the subscript $j$ in $\tilde F_j^{(k)}$, the sequence  $\tilde F_1^{(k)},\ldots, \tilde F_{D^k}^{(k)},\tilde F_1^{(k+1)},\ldots,\tilde F_{D^{k+1}}^{(k+1)},\ldots$ converges (with respect to $\P_\omega$) in distribution to $\RN(0,V(a))$. 
In order to handle both $D=4$ and $D\ne4$ together, let $J_k\subseteq\intbrr{1:N}$ be a set of indices. It will the whole set for $D\ne 4$, and all indices $j$ with a certain eigenspace index parity $\alpha_j$ for $D=4$.
Then as $k\to\infty$, the expectation value of the characteristic function $\operatorname{chf}_{J_k}(t)$ of the empirical distribution of $\{\tilde F_j^{(k)}\}_{j\in J_k}$ satisfies
\begin{align}\label{eqn:exp-chf}
\E_\omega[\operatorname{chf}_{J_k}(t)]&=
\frac{1}{|J_k|}\sum_{j\in J_k}\E_\omega[e^{it\tilde F_j^{(k)}}]\to e^{-t^2V(a)/2}.
\end{align}
In anticipation of applying Markov's inequality to show convergence in probability, we next compute 
$\E_\omega[|\operatorname{chf}_{J_k}(t)|^2]=\frac{1}{|J_k|^2}\sum_{i,j\in J_k}\E_\omega[e^{it(\tilde F_j-\tilde F_i)}]$.
The only issue with $\E_\omega[e^{it(\tilde F_j-\tilde F_i)}]$ is when $\tilde F_j$ and $\tilde F_i$ correspond to eigenvectors chosen from the same eigenspace, since then they are not independent. 
We will show
\begin{lem}\label{lem:chf-split}
As $k\to\infty$, for any $i,j\in\intbrr{1:D^k}$ with $i\ne j$,
\begin{align}\label{eqn:chf-split}
\E_\omega[e^{it(\tilde F_j-\tilde F_i)}]&= \E_\omega[e^{it \tilde F_j}]\E_\omega[e^{-it \tilde F_i}]+o_t(1),
\end{align}
where $o_t(1)$ denotes that the $o(1)$ term is allowed to depend on $t$.
\end{lem}
\begin{proof}
If $\tilde F_j$ and $\tilde F_i$ correspond to different eigenspaces, then they are independent and \eqref{eqn:chf-split} holds with no remainder term.
If $\tilde F_j$ and $\tilde F_i$ come from the same eigenspace $E_\jj$, then let $d=d_\jj=\frac{N}{q(k)}(1+o(1))$ be the eigenspace dimension, and let $u,v\in\C^d$ be two independent Haar-random unit vectors.
Recall the definition of $\tilde M$ from \eqref{eqn:M}.
The distributions of $\tilde F_j$ and $\tilde F_i$ in terms of $u,v$ are given as
\begin{align*}
(\tilde F_j,\tilde F_i)\dsim(\sqrt{N}\langle u|\tilde M|u\rangle,\sqrt{N}\langle v'|\tilde M|v'\rangle),\quad
\text{where }v'=\frac{v-\langle u|v\rangle u}{\|v-\langle u|v\rangle u\|_2},
\end{align*}
using the first steps of the Gram--Schmidt orthonormalization procedure.
Then
\begin{multline*}
\E_\omega[e^{it(\tilde F_j-\tilde F_i)}]=\\
\E\bigg[\exp\bigg[it\sqrt{N}(\langle u|\tilde M|u\rangle-\langle v|\tilde M|v\rangle)
+it\sqrt{N}\langle v|\tilde M|v\rangle\Big(1-\frac{1}{\|v-\langle u|v\rangle u\|_2^2}\Big)\\
+2it\sqrt{N}\frac{\Re[\langle v|\tilde M|u\rangle\langle u|v\rangle]}{\|v-\langle u|v\rangle u\|_2^2}+it\sqrt{N}\frac{|\langle u|v\rangle|^2\langle u|\tilde M|u\rangle}{\|v-\langle u|v\rangle u\|_2^2}\bigg]\bigg].\numberthis\label{eqn:chf-error}
\end{multline*}
The last three terms are generally small, since $\langle u|v\rangle$ is typically small in high dimension $d$.
More precisely, for $0<\epsilon<1/8$, consider the event
\begin{align*}
\Omega_{k,\epsilon}:=\left\{|\langle u|v\rangle|\le d^{-1/2+\epsilon}\text{ and }\max(|\langle u|\tilde M|u\rangle|,|\langle v|\tilde M|v\rangle|,|\langle v|\tilde M|u\rangle|,|\langle u|\tilde M|v\rangle|)\le N^{-1/2+\epsilon} \right\}.
\end{align*}
We can bound $\P[\Omega_{k,\epsilon}^c]=o(1)$ as follows.
The random vectors $u$ and $v$ are independent and distributed as $g/\|g\|_2$ for $g=(g_1,\ldots,g_d)\sim\CN(0,I_d)$. By rotational invariance of $u$ and $v$, as $d\to\infty$,
\begin{align}
\P[|\langle u|v\rangle|>d^{-1/2+\epsilon}]&=\P[|u_1|> d^{-1/2+\epsilon}]=\P[|g_1|>\|g\|_2d^{-1/2+\epsilon}]=o(1),
\end{align}
since the norm $\|g\|_2$ concentrates near $\sqrt{d}$, i.e.
\begin{align*}
\P[|\|g\|_2-\sqrt{d}|>t]\le 2e^{-ct^2},
\end{align*}
see e.g. \cite[Theorem 3.1.1]{Vershynin2018book}, recalling that $\|g\|_2$ can be viewed as the norm of $2d$ independent real Gaussian random variables each with variance $1/2$.
Using Markov's inequality, Weingarten calculus like in \eqref{eqn:variance}, and Corollary~\ref{lem:trM}, we can estimate
\begin{align}
\P[|\langle u|\tilde M|u\rangle|>N^{-1/2+\epsilon}]&\le\E[|\langle u|\tilde M|u\rangle|^2]N^{1-2\epsilon}=\left[(\Tr \tilde M)^2+\Tr(\tilde M^2)\right]\frac{(1+o(1))}{d^2}N^{1-2\epsilon}\le CN^{-2\epsilon}.
\end{align}
This similarly holds for $|\langle v|\tilde M|v\rangle|$, and for $|\langle v|\tilde M|u\rangle|$ and $|\langle u|\tilde M|v\rangle|$,
recalling $u$ and $v$ are independent. 
Thus $\P[\Omega_{k,\epsilon}^c]=o(1)$.

On $\Omega_{k,\epsilon}$, we have $|\langle u|v\rangle|\le d^{-1/2+\epsilon}$, and so
\begin{align}
\begin{aligned}
\|v-\langle u|v\rangle u\|_2^2&= 1-|\langle u|v\rangle|^2\ge1-d^{-1+2\epsilon},\\
\text{and}\quad\sqrt{N}|\langle w|\tilde M|w'\rangle| &\le N^\epsilon,\quad \text{for any }w,w'\in\{u,v\}.
\end{aligned}
\end{align}
Thus on $\Omega_{k,\epsilon}$, the last three terms in \eqref{eqn:chf-error} are
\begin{align}
it\sqrt{N}\langle v|\tilde M|v\rangle\bigg(1-\frac{1}{\|v-\langle u|v\rangle u\|_2^2}\bigg)+2it\sqrt{N}\frac{\Re[\langle v|\tilde M|u\rangle\langle u|v\rangle]}{\|v-\langle u|v\rangle u\|_2^2}+it\sqrt{N}\frac{|\langle u|v\rangle|^2\langle u|\tilde M|u\rangle}{\|v-\langle u|v\rangle u\|_2^2}&= t\,O(N^\epsilon d^{-1/2+\epsilon}),\label{eqn:omega-errorterms}
\end{align}
with implicit constant uniform over $\Omega_{k,\epsilon}$. 
Equation~\eqref{eqn:omega-errorterms} is $o(1)$ since $\epsilon<1/8$ and $d=N/q(k)(1+o(1))$ with $q(k)$ of order $k=\log_D N$. Then considering $\Omega_{k,\epsilon}$ and its complement, from \eqref{eqn:chf-error} we get
\begin{align*}
\E_\omega[e^{it(\tilde F_j-\tilde F_i)}]&=\E[e^{it\sqrt{N}\langle u|\tilde M|u\rangle}]\E[e^{-it\sqrt{N}\langle v|\tilde M|v\rangle}]+o_t(1)\\
&=\E_\omega[e^{it\tilde F_j}]\E_\omega[e^{-it \tilde F_i}]+o_t(1),\numberthis
\end{align*}
where $o_t(1)$ denotes that the $o(1)$ term is allowed to depend on $t$.
\end{proof}

\begin{proof}[Proof of the remainder of Theorem~\ref{thm:mat}]
This is the proof of the empirical distribution convergence in Theorem~\ref{thm:mat}.
The set $J_k\subseteq\intbrr{1:N}$ is to be taken as $\intbrr{1:N}$ for $D\ne4$, and one of $\{j:\alpha_j\text{ even}\}$ or $\{j:\alpha_j\text{ odd}\}$ for $D=4$. In any case, $|J_k|\to\infty$.
Applying Lemma~\ref{lem:chf-split}, we obtain,
\begin{align*}
\E_\omega[|\operatorname{chf}_{J_k}(t)|^2] =\frac{1}{|J_k|^2}\sum_{i,j\in J_k}\E_\omega[e^{it(\tilde F_j-\tilde F_i)}]
&=\frac{1}{|J_k|^2}\sum_{\substack{i,j\in J_k\\i\ne j}}(\E_\omega[e^{it \tilde F_j}]\E_\omega[e^{-it \tilde F_i}]+o_t(1))+\frac{1}{|J_k|}\\
&=\left|\E_\omega\operatorname{chf}_{J_k}(t)\right|^2+o_t(1).\numberthis
\end{align*}
Then for each $t\in\R$, Markov's inequality gives
\begin{align}\label{eqn:chf-markov}
\P_\omega[|\operatorname{chf}_{J_k}(t)-\E_\omega\operatorname{chf}_{J_k}(t)|>\epsilon]&\le \frac{\E_\omega[|\operatorname{chf}_{J_k}(t)|^2]-|\E_\omega[\operatorname{chf}_{J_k}(t)]|^2}{\epsilon^2}=\frac{o_t(1)}{\epsilon^2},
\end{align}
since $|J_k|\to\infty$.
Since we have $\E_\omega\operatorname{chf}_{J_k}(t)=e^{-t^2V(a)/2}+o_t(1)$ from \eqref{eqn:exp-chf}, a standard argument (see e.g. \cite[Lemma 2.2]{DiaconisFreedman1984}, \cite[Theorem 6.3]{Kallenberg-book}) 
shows that the convergence in probability \eqref{eqn:chf-markov} of the characteristic function implies that the measure $\rho_{J_k}:=\frac{1}{|J_k|}\sum_{j\in J_k}\delta_{\tilde F_j}$ converges weakly in probability to $\RN(0,V(a))$, i.e. for $g\sim\RN(0,V(a))$ and any bounded continuous $f:\R\to\C$,
\begin{align}\label{eqn:wpJk}
\P_\omega\left[\left|\E_{\rho_{J_k}}[f]-\E_g[f]\right|>\epsilon\right]\to0.
\end{align}
For $D\ne4$, $F_j=\tilde F_j$ and taking $J_k=\intbrr{1:N}$ proves the desired convergence \eqref{eqn:conv}. 
For $D=4$, define $Y$ to be the distribution given by \eqref{eqn:d4pdf}, let $g_\pm\sim\RN(\pm\dfoura,V(a))$, and $J_+=\{j:\alpha_j\text{ even}\}$ and let $J_-=\{j:\alpha_j\text{ odd}\}$. Recalling $\mu_k=\frac{1}{N}\sum_{j=1}^N\delta_{F_j}$, letting $\nu_{J_\pm}:=\frac{1}{|J_\pm|}\sum_{j\in J_\pm}\delta_{F_j}$, and using that $|J_+|$ and $|J_-|$ are both leading order $N/2$ by the eigenspace dimensions \eqref{eqn:dim}, we have
\begin{align*}
\P_\omega\left[\left|\E_{\mu_k}[f]-\E_Y[f]\right|>\epsilon\right]&=\P_\omega\Bigg[\Bigg|\frac{1}{N}\sum_{j\in J_+}f(F_j)+\frac{1}{N}\sum_{j\in J_-}f(F_j)-\frac{1}{2}(\E_{g_+}[f]+\E_{g_-}[f])\Bigg|>\epsilon\Bigg]\\
&\le \begin{multlined}[t]
\P_\omega\left[\left|(1+o(1))\E_{\nu_{J_+}}[f]-\E_{g_+}[f]\right|>\epsilon\right]
+\P_\omega\left[\left|(1+o(1))\E_{\nu_{J_-}}[f]-\E_{g_-}[f]\right|>\epsilon\right].
\end{multlined}\numberthis\label{eqn:wpJksplit}
\end{align*}
Using that $f$ is bounded and applying \eqref{eqn:wpJk} with the functions $x\mapsto f(x-\dfoura)$ and $x\mapsto f(x+\dfoura)$, implies \eqref{eqn:wpJksplit} tends to zero as $k\to\infty$.
This finishes the proof of Theorem~\ref{thm:mat}.
\end{proof}

\section{Off-diagonal matrix elements}\label{sec:off-diag}

\subsection{Random off-diagonal element}
To prove Theorem~\ref{thm:eth}, we will prove
\begin{thm}[Random ETH statement]\label{thm:eth-random}
Consider the setting of Theorem~\ref{thm:eth}.
For $i,j\in\intbrr{1:D^k}$, let $A_{ij}^{(k)}:=\langle\varphi^{(i)}|\Opkl(a)|\varphi^{(j)}\rangle$ and $\langle A\rangle:=\int_{\T^2}a(\x)\,d\x$. Viewing $A_{ij}^{(k)}$ as a random variable with respect to Haar measure in each of the eigenspaces $E_\jj$ of $\varphi^{(j)}$ and $E_\jjtwo$ of $\varphi^{(i)}$ (which may be the same eigenspace, in which case if $i\ne j$ they are always taken orthogonal), we have the following weak convergences in probability as $k\to\infty$:
\begin{align}\label{eqn:eth-random}
\frac{\sqrt{N}}{\sqrt{V(a)}f_a(\jj,\jjtwo)}(A^{(k)}_{ij}-\langle A\rangle\delta_{ij}) \xrightarrow{\P,w} R_{ij},
\end{align}
where $V(a)$ and $f_a(\jj,\jjtwo)$ are as in Theorem~\ref{thm:eth}, $R_{ij}$ for $i\ne j$ is a standard complex Gaussian random variable, and 
\begin{align*}
R_{ii}\dsim\begin{cases}
\RN(0,1),&D\ne 4\\
\RN\Big((-1)^{\jj_j}\dfoura/\sqrt{V(a)},1\Big),&D=4
\end{cases},
\end{align*}
provided for $D=4$ the $\jj_j$ are eventually a fixed parity.

For any $D\ge2$, the sequence
\begin{align}\label{eqn:eth-random-seq}
\left(\frac{\sqrt{D^k}}{\sqrt{V(a)}f_a(\jj,\jjtwo)}A^{(k)}_{ij}\right)_{\substack{i,j=1\\i\ne j}}^{D^k},
\left(\frac{\sqrt{D^{k+1}}}{\sqrt{V(a)}f_a(\jj,\jjtwo)}A^{(k+1)}_{ij}\right)_{\substack{i,j=1\\i\ne j}}^{D^{k+1}},\ldots
\end{align}
also converges weakly in probability to a standard complex normal distribution $\CN(0,1)$.
\end{thm}
Equation~\eqref{eqn:eth-random} is a generalization of the limiting individual fluctuations in QUE statement \eqref{eqn:clt}, but including off-diagonal matrix entries.
The off-diagonal matrix entries have zero expectation value whether $D=4$ or $D\ne4$, and their behavior ends up the same in both cases, unlike in the on-diagonal case.
We will use the statement \eqref{eqn:eth-random-seq} to finish the proof of Theorem~\ref{thm:eth} in Section~\ref{subsec:eth-conv}.

The diagonal statement $i=j$ of \eqref{eqn:eth-random} is already proved in \eqref{eqn:clt} and \eqref{eqn:clt4}. 
Thus we only need to prove the convergence for off-diagonal elements $\langle\varphi^{(i)}|\Opkl(a)|\varphi^{(j)}\rangle$, where $i\ne j$.
Since $i\ne j$, we have $\langle\varphi^{(i)}|\Opkl(a)|\varphi^{(j)}\rangle=\langle\varphi^{(i)}|\Opkl(a_0)|\varphi^{(j)}\rangle$ by orthogonality of the eigenbasis.
Also, $\E[\langle\varphi^{(i)}|\Opkl(a)|\varphi^{(j)}\rangle]=0$ immediately from either independence or Theorem~\ref{thm:weingarten}.

The characteristic function of a complex random variable $X$ is
$\chf_X(t)=\E[e^{i\Re(\bar{t}X)}]$ for $t\in\C$.
So instead of the usual moments $\E[X^p]$, we look at 
\begin{align}\label{eqn:rmoments}
\E[(\Re(\bar{t}X))^p]&= \frac{1}{2^p}\E[(tX+\bar{t}\bar{X})^p]
=\frac{1}{2^p}\sum_{n=0}^p\binom{p}{n}t^n\bar{t}^{p-n}\E[X^n\bar{X}^{p-n}],
\end{align}
for $X=\langle\varphi^{(i)}|\Opkl(a_0)|\varphi^{(j)}\rangle$.
By the Cram\'er--Wold theorem and moment method (see e.g. \cite[\S3.3.5 and \S3.10]{Durrett}), to show the convergences \eqref{eqn:eth-random} and \eqref{eqn:eth-random-seq}, it will be sufficient to show that for every $t\in\C\cong\R^2$ and $p\in\N$, a properly scaled version of \eqref{eqn:rmoments} converges to $\E[(\Re(\bar{t}Z))^p]$ where $Z\sim\CN(0,1)$, as $k\to\infty$ and with rates uniform over $i,j\in\intbrr{1:D^k}$.

\subsubsection{Same eigenspace}
First suppose $\varphi^{(i)}$ and $\varphi^{(j)}$ come from the same eigenspace, with index $\jj$. 
Let $M=\Lambda_\jj^\dagger\Opkl(a_0)\Lambda_\jj$ which is a $d_\jj\times d_\jj$ matrix.
Then $\langle\varphi^{(i)}|\Opkl(a_0)|\varphi^{(j)}\rangle$ is distributed as $\langle u^{(1)}|M|u^{(2)}\rangle$ for $u^{(1)},u^{(2)}$ the first two columns of a Haar random unitary matrix $U=(u^{(1)}\cdots u^{(d_\jj)})\in\mathcal U(d_\jj)$.
Since $M$ is self-adjoint, we look at
\begin{multline*}
\E[\langle u^{(1)}|M|u^{(2)}\rangle^n\langle u^{(2)}|M|u^{(1)}\rangle^{p-n}]\\
=\sum_{\substack{i_1,\ldots i_{p}=1\\i_1',\ldots,i_{p}'=1}}^{d_\jj}
M_{i_1'i_1}\cdots M_{i_{p}'i_{p}}\E[\overline{u^{(1)}_{i_1'}}\cdots \overline{u^{(1)}_{i_n'}}u^{(2)}_{i_{1}}\cdots u^{(2)}_{i_n}
\overline{u^{(2)}_{i_{n+1}'}}\cdots\overline{u^{(2)}_{i_{p}'}}u^{(1)}_{i_{n+1}}\cdots u^{(1)}_{i_{p}}].
\end{multline*}
By Theorem~\ref{thm:weingarten}, the expectation is zero if $n\ne p-n$, since to be nonzero we must have $\tau\in S_p$ map all the column indices $j_1=\cdots=j_n=2$ to the set of indices $j_{n+1}'=\cdots j_{p}'=2$, but this is only possible if $n=p-n$.
So we must have $p$ even and $n=p/2$ to obtain a nonzero expectation value.
Now with $n=p/2$ and by the same reasoning,
we see $\tau\in S_p$ must map $\{1,\ldots,p/2\}$ to $\{p/2+1,\ldots,p\}$ and vice versa, in order to match the column indices to have $\delta_{j_\ell j_{\tau(\ell)}'}\ne0$. Denoting this by $\tau:\{1,\ldots,p/2\}\leftrightarrow\{p/2+1,\ldots,p\}$, we have for $p$ even,
\begin{align}\label{eqn:offdiag-same}
\E[\langle u^{(1)}|M|u^{(2)}\rangle^{p/2}\langle u^{(2)}|M|u^{(1)}\rangle^{p/2}]
&=\sum_{\substack{\sigma\in S_p\\\text{shape }[c_1,\ldots,c_m]}}\Tr(M^{c_1})\cdots\Tr(M^{c_m})\sum_{\substack{\tau:\{1,\ldots,p/2\}\\\leftrightarrow\{p/2+1,\ldots,p\}}}\Wg(\tau\sigma^{-1}).
\end{align}
We first consider $D\ne4$. Then as in the proof of Theorem~\ref{thm:ind} given in Section~\ref{subsec:moments}, by Corollary~\ref{lem:trM} the leading order trace term
comes from when $c_1=\cdots=c_m=2$, as long as there is an allowable $\tau$ and a $\sigma$ with shape $[2,\ldots,2]$ so that $\tau\sigma^{-1}=\operatorname{Id}$, i.e. $\tau=\sigma$.
Since $\tau$ has to exchange $\{1,\ldots,p/2\}$ with $\{p/2+1,\ldots,p\}$, and since a permutation of shape $[2,\ldots,2]$ is always its own inverse, this happens exactly when $\tau=\sigma=\sigma^{-1}$ is of the form
\begin{align}\label{eqn:sigmaform}
(1\,s_1)(2\,s_2)\cdots(\tfrac{p}{2}\,s_{\frac{p}{2}}),
\end{align}
for $\{s_j\}_j=\{\frac{p}{2}+1,\ldots,p\}$.
For $D=4$, one initially has to consider $c_1,\ldots,c_m\in\{1,2\}$ by Corollary~\ref{lem:trM}. However, because $\tau$ exchanges two disjoint sets, any $\sigma$ with fixed points, that is $c_i=1$, cannot be equal to $\tau$ and so will produce a subleading order term due to smallness of  $\Wg(\tau\sigma^{-1})$. Therefore the leading order term still only comes from $c_1=\cdots=c_m=2$, just as in the $D\ne4$ case.

Returning to \eqref{eqn:sigmaform}, there are $(p/2)!$ choices for the sequence $(s_j)$, so there are a total of $(p/2)!$ such permutations $\tau=\sigma=\sigma^{-1}$, yielding
\begin{align*}
\E[\langle u^{(1)}|M|u^{(2)}\rangle^{p/2}\langle u^{(2)}|M|u^{(1)}\rangle^{p/2}]&=(\Tr M^2)^{p/2}\frac{(p/2)!}{d_\jj^p}(1+o(1)),
\end{align*}
and
\begin{align*}
\E[(\Re(\bar t\langle\varphi^{(i)}|\Opkl(a)|\varphi^{(j)}\rangle ))^p]&=\frac{\mathbf{1}_{p\in2\N}}{2^p}\binom{p}{p/2}|t|^p(p/2)!N^{-p/2}V(a)^{p/2}(1+o(1)),\numberthis\label{eqn:re-p}
\end{align*}
using Corollary~\ref{lem:trM} and that $d_\jj=\frac{N}{q(k)}(1+o(1))$.
Since the characteristic function of a standard complex Gaussian random variable is $e^{-|t|^2/4}$, we can check, for example by noting that
\begin{align*}
\sum_{p=0}^\infty \frac{i^p}{p!}\frac{\mathbf{1}_{p\in2\N}}{2^p}\binom{p}{p/2}|t|^p(p/2)!
&=\sum_{r=0}^\infty\frac{(-1)^r}{(2r)!}\frac{1}{4^r}\frac{(2r)!}{r!r!}|t|^{2r}r!=e^{-|t|^2/4},
\end{align*}
that for any $t\in\C$ and $p\in\N$, as $k\to\infty$,
\begin{align}\label{eqn:re-conv}
\left(\frac{N}{V(a)}\right)^{p/2}\E[(\Re(\bar t\langle\varphi^{(i)}|\Opkl(a)|\varphi^{(j)}\rangle ))^p]\to \E[(\Re(\bar{t}Z))^p],
\end{align}
for $Z\sim\CN(0,1)$.

\subsubsection{Different eigenspaces}\label{subsubsec:diff-eigenspace}
Now consider when $\varphi^{(i)}$ and $\varphi^{(j)}$ come from different eigenspaces, say those with indices $\jj$ and $\jjtwo$. In this case $\langle\varphi^{(i)}|\Opkl(a_0)|\varphi^{(j)}\rangle$ is distributed as $\langle u,\mathcalboondox Mv\rangle$ for $\mathcalboondox M=\Lambda_\jj^\dagger\Opkl(a_0)\Lambda_\jjtwo$ a (not necessarily square) $d_\jj\times d_\jjtwo$ matrix, and $u\in\C^{d_\jj}$, $v\in\C^{d_\jjtwo}$ independent Haar random vectors. Recall from \eqref{eqn:F-unitary} that $\Lambda_\jj$ (resp. $\Lambda_\jjtwo$) is an $N\times d_\jj$ ($N\times d_\jjtwo$) matrix whose columns form an orthonormal basis for the eigenspace $E_\jj$ ($E_\jjtwo$).
We then proceed similarly as in the case of the same eigenspace, taking care to differentiate $d_\jj$ from $d_\jjtwo$. We will add a subscript $(d)$ to the Weingarten function $\Wg$ in Theorem~\ref{thm:weingarten} to identify the dimension. For \eqref{eqn:rmoments}, we compute
\begin{multline}
\E[\langle u|\mathcalboondox M|v\rangle^n\overline{\langle u|\mathcalboondox M|v\rangle^{p-n}}] =\E[\langle u|\mathcalboondox M|v\rangle^n\langle v|\mathcalboondox M^\dagger|u\rangle^{p-n}] \\
= \sum_{\substack{i_1',\ldots i'_n=1\\i_1,\ldots,i_{p-n}=1}}^{d_\jj}\sum_{\substack{\ell_1',\ldots,\ell_{p-n}'=1\\\ell_1,\ldots,\ell_n=1}}^{d_\jjtwo}\mathcalboondox M_{i_1'\ell_1}\cdots \mathcalboondox M_{i_n'\ell_n}\mathcalboondox M^\dagger_{\ell_1'i_1}\cdots\mathcalboondox M^\dagger_{\ell_{p-n}'i_{p-n}}
\E[\bar u_{i_1'}\cdots \bar u_{i_n'}u_{i_1}\cdots u_{i_{p-n}}]\\
\times \E[v_{\ell_1}\cdots v_{\ell_n}\bar v_{\ell_1'}\cdots \bar v_{\ell_{p-n}'}].
\end{multline}
By \eqref{eqn:weingarten-diff}, we see we must have $n=p/2$ and $p$ even to have  nonzero expectation values.
Applying the rest of Theorem~\ref{thm:weingarten}, we have for $p$ even,
\begin{multline}
\E[\langle u|\mathcalboondox M|v\rangle^{p/2}\langle v|\mathcalboondox M^\dagger|u\rangle^{p/2}]=\\
\sum_{\sigma_u\in S_{p/2}}
\sum_{\sigma_v\in S_{p/2}}
\sum_{i_1',\ldots,i_{p/2}'=1}^{d_\jj}\sum_{\ell_1',\ldots,\ell_{p/2}'=1}^{d_\jjtwo}
\mathcalboondox M_{i_1'\ell'_{\sigma_v(1)}}\cdots \mathcalboondox M_{i_{p/2}'\ell'_{\sigma_v(p/2)}}\mathcalboondox M^\dagger_{\ell_1'i_{\sigma_u(1)}'}\cdots\mathcalboondox M^\dagger_{\ell_{p/2}'i_{\sigma_u(p/2)}'}\\
\times\sum_{\tau\in S_{p/2}}\Wg_{(d_\jj)}(\tau)\sum_{\tau\in S_{p/2}}\Wg_{(d_\jjtwo)}(\tau). 
\end{multline}
To simplify the above expression, we can combine the indexing over $i'_m$ and $\ell'_m$ by defining $i'_{p/2+m}:=\ell_m'$, $m=1,\ldots,p/2$. With this indexing, we can also combine $\sigma_u$ and $\sigma_v$ by defining a new permutation $\sigma\in S_p$ via $\sigma(x):=\sigma_v(x)+p/2$ if $x\in\intbrr{1:p/2}$, and $\sigma(p/2+x):=\sigma_u(x)$.
Then we can replace every $i'_{\sigma_u(x)}$ with $i'_{\sigma(p/2+x)}$, and every $\ell'_{\sigma_v(x)}$ with $i'_{\sigma(x)}$, giving
\begin{multline}
\sum_{\sigma\in S_p}
\sum_{i_1',\ldots,i_{p/2}'=1}^{d_\jj}\sum_{i_{p/2+1}',\ldots,\ell_{p}'=1}^{d_\jjtwo}
\mathcalboondox M_{i_1'i'_{\sigma(1)}}\cdots \mathcalboondox M_{i_{p/2}'i'_{\sigma(p/2)}}\mathcalboondox M^\dagger_{i_{p/2+1}'i_{\sigma(p/2+1)}'}\cdots\mathcalboondox M^\dagger_{i'_{p}i'_{\sigma(p)}}\frac{1}{d_\jj^{p/2}d_\jjtwo^{p/2}}(1+o(1)).\label{eqn:2p}
\end{multline}
Note that $\sigma$ sends $\{1,\ldots,p/2\}$ to $\{p/2+1,\ldots,p\}$ and vice versa. 
Because of this requirement, every cycle in $\sigma$ must have even length, and so the cycle shape $[c_1,\ldots,c_m]$ of $\sigma$ only has even $c_1,\ldots,c_m$. 
Additionally, for a cycle $(s_1\cdots s_{c_j})$ of length $c_j$ in $\sigma$, which without loss of generality starts with an $s_1\le p/2$, summing over the corresponding indices and matrix elements in \eqref{eqn:2p} yields
\begin{align*}
\sum_{i'_{s_1},\ldots,i'_{s_{c_j}}=1}^{d_\jj,d_\jjtwo}\mathcalboondox M_{i'_{s_1}i'_{s_2}}\mathcalboondox M^\dagger_{i'_{s_2}i'_{s_3}}\cdots\mathcalboondox M_{i'_{s_{c_j-1}}i'_{s_{c_j}}}\mathcalboondox M^\dagger_{i'_{s_{c_j}}i_{s_1}}&=\Tr((\mathcalboondox M\mathcalboondox M^\dagger)^{c_j/2}),
\end{align*}
where we always alternate $\mathcalboondox M$ and $\mathcalboondox M^\dagger$ due to the requirement that $\sigma$ swaps $\{1,\ldots,p/2\}$ and $\{p/2+1,\ldots,p\}$, combined with the observation that the first index of $\mathcalboondox M$ in \eqref{eqn:2p} is always $\le p/2$ and the first index of $\mathcalboondox M^\dagger$ is always $\ge p/2+1$. 
Thus \eqref{eqn:2p} becomes 
\begin{align}\label{eqn:2p-2}
\E[\langle u|\mathcalboondox M|v\rangle^{p/2}\langle v|\mathcalboondox M^\dagger|u\rangle^{p/2}]=\sum_{\substack{\sigma:\{1,\ldots,p/2\}\\\;\;\;\leftrightarrow\{p/2+1,\ldots,p\},\\\text{shape }[c_1,\ldots,c_m]}}\Tr((\mathcalboondox M\mathcalboondox M^\dagger)^{c_1/2})\cdots\Tr((\mathcalboondox M\mathcalboondox M^\dagger)^{c_m/2})\frac{1}{d_\jj^{p/2} d_\jjtwo^{p/2}}(1+o(1)).
\end{align}
We can identify the largest order for the trace term. 
For this, note that for any $p\in\N$, $\Tr\left[(\mathcalboondox M\mathcalboondox M^\dagger)^p\right]=\Tr\left[(P_\jj\Opkl(a_0)P_\jjtwo\Opkl(a_0))^p\right]$; in particular,
\begin{align*}
\Tr(\mathcalboondox M\mathcalboondox M^\dagger)=\Tr(\Lambda_\jj^\dagger\Opkl(a_0)\Lambda_\jjtwo\Lambda_\jjtwo^\dagger\Opkl(a_0)\Lambda_\jj)=\Tr(P_\jj\Opkl(a_0)P_\jjtwo\Opkl(a_0)),
\end{align*}
so we need a version of \eqref{eqn:trapap} for different $P_\jj$, $P_\jjtwo$.
\begin{prop}\label{prop:var-ind-2}
For $\min(\ell,k-\ell)\ge 3\log_D k$ and any $\jj,\jjtwo\in\intbrr{0:q(k)-1}$,
\begin{align}\label{eqn:var-ind-2}
\Tr(\Opkl(a_0)P_\jj\Opkl(a_0)P_\jjtwo)
=\frac{N}{q(k)^2}\left[\widetilde V(a,\jj,\jjtwo)+o(1)\right],
\end{align}
where
\begin{align}\label{eqn:tv}
\widetilde V(a,\jj,\jjtwo):=\sum_{t=-\infty}^\infty e^{2\pi it(\jj-\jjtwo)/q(k)}\left(\int_{\T^2}a_0(\x)a_0(B^t\x)\,d\x+\mathbf{1}_{D\ge3}\int_{\T^2}a_0(\x)a_0(B^tR\x)\,d\x\right).
\end{align}
\end{prop}
\begin{proof}
The proof of \eqref{eqn:var-ind-2} is the same as for the single-eigenspace counterpart \eqref{eqn:trapap}, using the cancellation lemmas from Section~\ref{sec:phases}. 
By \eqref{eqn:ppoly0}, the desired trace is
\begin{align}
\Tr(\Opkl(a_0)P_\jj \Opkl(a_0)P_\jjtwo)&=\frac{1}{q(k)^2}\sum_{t_1,t_2=-q(k)/2}^{q(k)/2-1}\Tr(\Opkl(a_0)\Ba_k^{t_1}\Opkl(a_0)\Ba_k^{t_2})e^{2\pi i(\jj t_1+\jjtwo t_2)/q(k)},
\end{align}
with the only difference from \eqref{eqn:trap2-expand} being that the phase factor at the end is $e^{2\pi i(\jj t_1+\jjtwo t_2)/q(k)}$.
The phase factor never played a role in showing the terms where $t_1+t_2\ne0\;\mathrm{mod}\;q(k)$ were subleading order in \eqref{eqn:tr2-error}; they are thus still $o(D^k/q(k)^2)$.
The terms where $t_1+t_2=0\;\mathrm{mod}\;q(k)$, or $t_2=-t_1$, give the leading order
\begin{multline*}
\frac{1}{q(k)^2} \sum_{t_1=-q(k)/2}^{q(k)/2-1}e^{2\pi i(\jj t_1-\jjtwo t_1)/q(k)}\Tr(\Opkl(a_0)\Ba_k^{t_1}\Opkl(a_0)\Ba_k^{-t_1})\\
= \frac{N}{q(k)^2}\left[\sum_{t=-\infty}^\infty e^{2\pi it(\jj-\jjtwo)/q(k)}\left(\int_{\T^2}a_0(\x)a_0(B^t\x)\,d\x+\mathbf{1}_{D\ge 3}\int_{\T^2}a_0(\x)a_0(B^tR\x)\,d\x\right)\right]+o\bigg(\frac{N}{q(k)^2}\bigg),
\end{multline*}
using the same argument as in Lemma~\ref{lem:intsum}. 

Finally, \eqref{eqn:tv} is real by considering $t\mapsto -t$ in the sum, since $B^t$ and $R$ are measure-preserving and invertible, and $R=R^{-1}$ commutes with $B$. It is eventually nonnegative since $\tilde V(a,\alpha,\beta)+o(1)=\frac{q(k)^2}{N}\Tr(\mathcalboondox M\mathcalboondox M^\dagger)\ge0$.
\end{proof}

Applying Proposition~\ref{prop:var-ind-2} with \eqref{eqn:tr-prod} from Corollary~\ref{lem:trM} for higher order traces, we see the leading order terms in \eqref{eqn:2p-2} are those with $c_1=\cdots=c_m=2$, giving terms of order $\frac{N^{p/2}}{q(k)^p}\frac{\widetilde V(a,\jj,\jjtwo)^{p/2}}{d_\jj^{p/2}d_\jjtwo^{p/2}}=O(N^{-p/2})$.
The number of $\sigma\in S_p$ mapping $\{1,\ldots,p/2\}\leftrightarrow\{p/2+1,\ldots,p\}$ with cycle shape $[2,\ldots,2]$ is $(p/2)!$ as before. 
This gives the analog of \eqref{eqn:re-p} and \eqref{eqn:re-conv} in this case where $\varphi^{(i)}$ and $\varphi^{(j)}$ are from different eigenspaces $E_\jj$ and $E_\jjtwo$:
\begin{align}
\E[(\Re(\bar t\langle\varphi^{(i)}|\Opkl(a)|\varphi^{(j)}\rangle ))^p]&=\frac{\mathbf{1}_{p\in2\N}}{2^p}\binom{p}{p/2}|t|^p(p/2)!N^{-p/2}\widetilde V(a,\jj,\jjtwo)^{p/2}(1+o(1)),\numberthis
\end{align}
and
\begin{align}
\left(\frac{N}{\widetilde V(a,\jj,\jjtwo)}\right)^{p/2}\E[(\Re(\bar t\langle\varphi^{(i)}|\Opkl(a)|\varphi^{(j)}\rangle ))^p]\to \E[(\Re(\bar{t}Z))^p],
\end{align}
for $Z\sim\CN(0,1)$ and $\widetilde V(a,\jj,\jjtwo)$ as defined in \eqref{eqn:tv}.

Letting $f_a(\jj,\jjtwo)=\left[\frac{\tilde V(a,\jj,\jjtwo)}{V(a)}\right]^{1/2}$ as in \eqref{eqn:f}, this finishes the proof of Theorem~\ref{thm:eth-random}. \qed

\subsection{Convergence of the empirical distribution}\label{subsec:eth-conv}

To show the empirical distribution convergence in Theorem~\ref{thm:eth}, we proceed similarly as in Section~\ref{subsec:gaussian}.
For $k\to\infty$, we have 
\begin{align}
\E_\omega[\operatorname{chf}_{\nu^{(k)}_{\jj\jjtwo}}(t)]&=\frac{1}{\#\{(j,i)\in J_\jj\times J_\jjtwo:j\ne i\}}\sum_{\substack{j\in J_\jj,i\in J_\jjtwo\\ j\ne i}}\!\!\!\E_\omega\Big[e^{i\frac{\sqrt{N}}{\sqrt{V(a)}f_a(\jj,\jjtwo)}\Re(\bar t\langle\varphi^{(i)}|\Opkl(a_0)|\varphi^{(j)}\rangle)}\Big]\to e^{-|t|^2/4}
\end{align}
by \eqref{eqn:eth-random-seq} of Theorem~\ref{thm:eth-random}.
Following the argument in Section~\ref{subsec:gaussian}, adapted for measures on $\C$,
Theorem~\ref{thm:eth} follows if we prove convergence in probability $\operatorname{chf}_{\nu^{(k)}_{\jj\jjtwo}}(t)\overset{p}{\to}e^{-|t|^2/4}$ for $t\in\C$, which will follow from Markov's inequality and the following analogue of Lemma~\ref{lem:chf-split}.
\begin{lem}\label{lem:chf-split-2}
Let $\varphi^{(i_1)},\varphi^{(j_1)},\varphi^{(i_2)},\varphi^{(j_2)}$ be distinct orthonormal random eigenvectors, where $\varphi^{(j_1)},\varphi^{(j_2)}$ come from the same eigenspace, and $\varphi^{(i_1)},\varphi^{(i_2)}$ come from the same eigenspace (which may the same eigenspace as that of the $j_1,j_2$'s).
Then in the setting of Theorem~\ref{thm:eth}, for $t\in\C$,
\begin{multline}\label{eqn:offdiag-exp}
\E_\omega\left[\exp\left({i\sqrt{N}\Re\big(\bar t\langle\varphi^{(i_1)}|\Opkl(a_0)|\varphi^{(j_1)}\rangle-\bar t\langle\varphi^{(i_2)}|\Opkl(a_0)|\varphi^{(j_2)}\rangle\big)}\right)\right]\\
=\E_\omega\left[\exp\left(i\sqrt{N}\Re\big(\bar t\langle\varphi^{(i_1)}|\Opkl(a_0)|\varphi^{(j_1)}\rangle\right)\right]\\
\times\E_\omega\left[\exp\left(-i\sqrt{N}\Re\big(\bar t\langle\varphi^{(i_2)}|\Opkl(a_0)|\varphi^{(j_2)}\rangle\big)\right)\right]+o_t(1).
\end{multline}
\end{lem}

\begin{proof}
There are two cases, when either all $\varphi^{(i_1)},\varphi^{(j_1)},\varphi^{(i_2)},\varphi^{(j_2)}$ come from the same eigenspace $E_{\jj}$, or when $\varphi^{(j_1)},\varphi^{(j_2)}$ come from the eigenspace $E_\jj$, and $\varphi^{(i_1)},\varphi^{(i_2)}$ come from a different eigenspace $E_\jjtwo$.

\textit{Same eigenspaces}.
Suppose $i_1,j_1,i_2,j_2$ are all distinct, and that the eigenvectors $\varphi^{(i_1)},\varphi^{(j_1)},\varphi^{(i_2)},\varphi^{(j_2)}$ come from the same eigenspace $E_\jj$ of dimension $d=d_\jj$.
Let $u_1,u_2,v_1,v_2$ be independent random unit vectors each chosen from Haar measure on $\mathcal U(d)$. 
Applying the Gram--Schmidt procedure, we can define
\begin{align}
\begin{aligned}
u_2'&=\frac{u_2-\langle u_1|u_2\rangle u_1}{\|u_2-\langle u_1|u_2\rangle u_1\|_2},
\\
v_1'&=\frac{v_1-\langle u_1|v_1\rangle u_1-\langle u_2'|v_1\rangle u_2'}{\|v_1-\langle u_1|v_1\rangle u_1-\langle u_2'|v_1\rangle u_2'\|_2},
\quad
v_2'=\frac{v_2-\langle u_1|v_2\rangle u_1-\langle u_2'|v_2\rangle u_2'-\langle v_1'|v_2\rangle v_1'}{\|v_2-\langle u_1|v_2\rangle u_1-\langle u_2'|v_2\rangle u_2'-\langle v_1'|v_2\rangle v_1'\|_2},
\end{aligned}
\end{align}
so that $u_1,u_2',v_1',v_2'$ are now orthonormal random unit vectors. They have the same distribution as four columns of a random $d\times d$ Haar unitary matrix,
and so the left side of expression \eqref{eqn:offdiag-exp} is equal to
\begin{align}\label{eqn:exp-uv}
\E\left[\exp\left({i\sqrt{N}\Re\big(\bar t\langle u_1| M|u_2'\rangle-\bar t\langle v_1'| M|v_2'\rangle\big)}\right)\right],
\end{align}
where $M=\Lambda_\jj^\dagger\Opkl(a_0)\Lambda_\jj$ as usual.
For $0<\epsilon<1/8$, consider the event
\begin{multline}
\Omega_{k,\epsilon}=\big\{|\langle y,z\rangle|\le d^{-1/2+\epsilon}\;\text{for }y\ne z\in\{u_1,u_2,v_1,v_2\},\\
\text{ and }|\langle w|M|v\rangle|\le N^{-1/2+\epsilon}\text{ for }w,v\in\{u_1,u_2,v_1,v_2\}\big\}.
\end{multline}
As in the proof of Lemma~\ref{lem:chf-split}, we have $\P[\Omega_{k,\epsilon}]=o(1)$.

On $\Omega_{k,\epsilon}$, similar to the estimates in Lemma~\ref{lem:chf-split}, we have
\begin{align}
\begin{aligned}
\sqrt{N}\langle u_1|M|u_2'\rangle&=\sqrt{N}\langle u_1|M|u_2\rangle+O(d^{-1/2+\epsilon}N^\epsilon),\\
\sqrt{N}\langle v_1'|M|v_2'\rangle&=\sqrt{N}\langle v_1|M|v_2\rangle+O(d^{-1/2+\epsilon}N^\epsilon).
\end{aligned}
\end{align}
Thus considering $\Omega_{k,\epsilon}$ and its complement, \eqref{eqn:exp-uv} is
\begin{align*}
\E\Big[\exp\Big(i\sqrt{N}&\Re\big(\bar t\langle u_1| M|u_2'\rangle-\bar t\langle v_1'| M|v_2'\rangle\big)\Big)\Big]\\
&=\E\Big[\exp\Big(i\sqrt{N}\Re\big(\bar t\langle u_1| M|u_2\rangle-\bar t\langle v_1| M|v_2\rangle\big)\Big)e^{O(|t|d^{-1/2+\epsilon}N^\epsilon)}\oneb_{\Omega_{k,\epsilon}}\Big]+O(\P[\Omega_{k,\epsilon}^c])\\
&=\E\left[\exp\left(i\sqrt{N}\Re(\bar t\langle u_1| M|u_2'\rangle)\right)\right]\E\left[\exp\left(-i\sqrt{N}\Re(\bar t\langle v_1'| M|v_2'\rangle)\right)\right]+o_t(1),
\end{align*}
which gives \eqref{eqn:offdiag-exp}.

\textit{Different eigenspaces}.
Suppose $i_1,i_2,j_1,j_2$ are all distinct, and that $\varphi^{(j_1)},\varphi^{(j_2)}$ come from $E_\jj$ while $\varphi^{(i_1)},\varphi^{(i_2)}$ come from $E_\jjtwo$. 
As in Section~\ref{subsubsec:diff-eigenspace}, let $\mathcalboondox M=\Lambda_\jj^\dagger\Opkl(a_0)\Lambda_\jjtwo$ which is a rectangular $d_\jj\times d_\jjtwo$ matrix.
Let $u_1,v_1\in\C^{d_\jjtwo}$ and $u_2,v_2\in\C^{d_\jj}$ be independent Haar unit vectors.
Then the left side of expression \eqref{eqn:offdiag-exp} is equal to
\begin{align}
\E\left[\exp\left({i\sqrt{N}\Re\big(\bar t\langle u_1|\mathcalboondox M|u_2\rangle-\bar t\langle v_1'|\mathcalboondox M|v_2'\rangle\big)}\right)\right],
\end{align}
where
\begin{align}
\begin{aligned}
v_1'&=\frac{v_1-\langle u_1|v_1\rangle u_1}{\|v_1-\langle u_1|v_1\rangle u_1\|_2},
\quad v_2'=\frac{v_2-\langle u_2|v_2\rangle u_2}{\|v_2-\langle u_2|v_2\rangle u_2\|_2}.
\end{aligned}
\end{align}
For $0<\epsilon<1/8$, we consider 
\begin{multline*}
\Omega_{k,\epsilon}=\{|\langle u_1|v_1\rangle|\le d^{-1/2+\epsilon},\;|\langle u_2|v_2\rangle|\le d^{-1/2+\epsilon},\\
\text{ and }|\langle w|\mathcalboondox M|y\rangle|\le N^{-1/2+\epsilon}\text{ for }w\in\{u_1,v_1\},\,y\in\{u_2,v_2\}\}.
\end{multline*}
The same argument as in the proof of Lemma~\ref{lem:chf-split}, this time using Weingarten calculations from Section~\ref{subsubsec:diff-eigenspace}, shows that $\P[\Omega_{k,\epsilon}^c]=o(1)$. 
Once again considering $\Omega_{k,\epsilon}$ and its complement, we obtain
\begin{multline*}
\E\left[\exp\left(i\sqrt{N}\Re\big(\bar t\langle u_1|\mathcalboondox M|u_2\rangle-\bar t\langle v_1'|\mathcalboondox M|v_2'\rangle\big)\right)\right]\\
=\E\left[\exp\left(i\sqrt{N}\Re\big(\bar t\langle u_1|\mathcalboondox M|u_2\rangle\big)\right)\right]\E\left[\exp\left(-i\sqrt{N}\Re\big(\bar t\langle v_1'|\mathcalboondox M|v_2'\rangle\big)\right)\right]+o_t(1),
\end{multline*}
which gives \eqref{eqn:offdiag-exp}.
\end{proof}

\begin{proof}[Proof of Theorem~\ref{thm:eth}]
We only need to consider the convergence of the off-diagonal matrix elements, since the on-diagonal convergence is already proved in Theorems~\ref{thm:mat} and \ref{thm:ind}.
By Markov's inequality,
\begin{align}\label{eqn:offdiag-markov1}
\P_\omega[|\chf_{\nu^{(k)}_{\jj\jjtwo}}(t)-\E\chf_{\nu^{(k)}_{\jj\jjtwo}}(t)|>\epsilon]&\le \frac{\E_\omega[|\chf_{\nu^{(k)}_{\jj\jjtwo}}(t)|^2]-|\E_\omega[\operatorname{chf}_{\nu^{(k)}_{\jj\jjtwo}}(t)]|^2}{\epsilon^2}.
\end{align}
Let $n_1(k)=\max(|J_\jj(k)|,|J_\jjtwo(k)|)$ and $n_2(k)=\min(|J_\jj(k)|,|J_\jjtwo(k)|)$, which both $\to\infty$ by assumption.
Using Lemma~\ref{lem:chf-split-2},
\begin{align*}
&\E_\omega|\chf_{\nu^{(k)}_{\jj\jjtwo}}(t)|^2\\
&=\begin{multlined}[t]
\frac{1}{\#\{(j,i)\in J_\jj\times J_\jjtwo:j\ne i\}^2}\\
\sum_{\substack{j_1,j_2\in J_\jj,i_1,i_2\in J_\jjtwo\\ j_1\ne i_1,j_2\ne i_2}}\!\!\!\E_\omega\left[\exp\left(\frac{i\sqrt{N}}{\sqrt{V(a)}f_a(\jj,\jjtwo)}\Re\left(\bar t\langle\varphi^{(i_1)}|\Opkl(a_0)|\varphi^{(j_1)}\rangle-\bar t\langle\varphi^{(i_2)}|\Opkl(a_0)|\varphi^{(j_2)}\rangle\right)\right)\right]
\end{multlined}\\
&=\begin{multlined}[t]
\frac{1}{\#\{(j,i)\in J_\jj\times J_\jjtwo:j\ne i\}^2}\\
\Bigg(\sum_{\substack{j_1,j_2\in J_\jj,i_1,i_2\in J_\jjtwo\\ j_1,j_2,i_1,i_2\text{ distinct}}}\!\!\!\E_\omega\left[e^{\frac{i\sqrt{N}}{\sqrt{V(a)}f_a(\jj,\jjtwo)}\Re\left(\bar t\langle\varphi^{(i_1)}|\Opkl(a_0)|\varphi^{(j_1)}\rangle\right)}\right]\E_\omega\left[e^{-\frac{i\sqrt{N}}{\sqrt{V(a)}f_a(\jj,\jjtwo)}\Re\left(\bar t\langle\varphi^{(i_2)}|\Opkl(a_0)|\varphi^{(j_2)}\rangle\right)}\right]\Bigg)\\
+o(1)+\frac{O(n_1(k)n_2(k)^2)}{(n_1(k)n_2(k)-O(n_1(k)))^2}
\end{multlined}\\
&=|\E_\omega\chf_{\nu^{(k)}_{\jj\jjtwo}}(t)|^2+ o(1).
\end{align*}
Therefore \eqref{eqn:offdiag-markov1} is $o(1)$, so $\chf_{\nu^{(k)}_{\jj\jjtwo}}(t)-\E\chf_{\nu^{(k)}_{\jj\jjtwo}}(t)\xrightarrow{\P}0$, where the convergence is in probability. Combining with Theorem~\ref{thm:eth-random}, which implies $\E\chf_{\nu^{(k)}_{\jj\jjtwo}}(t)\to e^{-|t|^2/4}$, we have convergence in probability of the characteristic function $\chf_{\nu^{(k)}_{\jj\jjtwo}}(t)$ to $e^{-|t|^2/4}$.
The continuity theorem/tightness argument discussed at the start of Section~\ref{subsec:eth-conv} (see e.g. \cite[Lemma 2.2]{DiaconisFreedman1984}, \cite[Theorem 6.3]{Kallenberg-book}) then implies the weak convergence in probability, $\nu^{(k)}_{\jj\jjtwo}\xrightarrow{w,\P}\CN(0,1)$ as $k\to\infty$.
\end{proof}

\appendix

\section{QUE decay rates}\label{sec:que}

In this section, we provide the details for \eqref{eqn:ind-que} and \eqref{eqn:therm}. 
First we bound the on-diagonal terms. 
Let $(\varphi^{(k,j)})_{j=1}^{D^k}$ be an orthonormal basis of eigenvectors chosen randomly according to Haar measure within each eigenspace $E_\jj$, $\jj=0,\ldots,q(k)-1$, and let $\Lambda_\jj$ be an $N\times d_\jj$ matrix whose columns form an orthonormal basis for $E_\jj$.
Let $P_\jj=\Lambda_\jj\Lambda_\jj^\dagger$ be the spectral projection matrix onto $E_\jj$, and let $d_\jj=\dim E_\jj=\Tr(P_\jj)$.
For a random $\varphi^{(k,j)}$ from the eigenspace $E_\jj$, we have
\begin{align*}
\langle \varphi^{(k,j)}|\Opkl(a)|\varphi^{(k,j)}\rangle&\dsim\frac{1}{\|g\|_2^2}\langle g|\Lambda_\jj^\dagger \Opkl(a)\Lambda_\jj|g\rangle,
\end{align*}
where $g\dsim\CN(0,I_{d_\jj})$. By concentration of $\|g\|_2^2$ from Bernstein's inequality, see e.g. \cite[Eq.~(3.1)]{Vershynin2018book},
\begin{align}\label{eqn:normc}
\P\left[\left|\left(\frac{1}{\|g\|_2^2}-\frac{1}{d_\jj}\right)\langle g|\Lambda_\jj^\dagger\Opkl(a)\Lambda_\jj|g\rangle\right|>t\|a\|_\infty\right]&\le \P\left[\left|\|g\|_2^2-d_\jj\right|>td_\jj\right]\le 2\exp\left[-cd_\jj\min(t^2,t)\right],
\end{align}
for $t\ge0$, and so we will work with the quadratic form $\frac{1}{d_\jj}\langle g|\Lambda_\jj^\dagger \Opkl(a)\Lambda_\jj|g\rangle$.
Applying the Hanson--Wright inequality from \cite[Theorem 1.1]{RudelsonVershynin2013} or Bernstein inequality \cite[Theorem 2.8.1]{Vershynin2018book} 
gives for $t\ge0$,
\begin{align*}
\P\big[\big|\langle g|\Lambda_\jj^\dagger \Opkl(a)\Lambda_\jj|g\rangle&-\E\langle g|\Lambda_\jj^\dagger \Opkl(a)\Lambda_\jj|g\rangle\big|>td_\jj \|a\|_\infty\big]\\
&\le 2\exp\left[-c\min\left(\frac{t^2d_\jj^2\|a\|_\infty^2}{\|\Lambda_\jj^\dagger \Opkl(a)\Lambda_\jj\|^2_\mathrm{HS}},\frac{td_\jj\|a\|_\infty}{\|\Lambda_\jj^\dagger \Opkl(a)\Lambda_\jj\|}\right)\right]\\
&\le 2\exp\left[-c'd_\jj \min(t^2,t)\right],\numberthis\label{eqn:hw}
\end{align*}
using that the spectral norm is bounded as $\|\Lambda_\jj^\dagger \Opkl(a)\Lambda_\jj\|\le \|a\|_\infty$, and the Hilbert--Schmidt norm as
\begin{align*}
\|\Lambda_\jj^\dagger \Opkl(a)\Lambda_\jj\|_\mathrm{HS}^2=\Tr(P_\jj\Opkl(a)P_\jj\Opkl(a))\le \Tr(P_\jj)\|\Opkl(a)P_\jj\Opkl(a)\|\le d_\jj\|a\|_\infty^2. 
\end{align*} 

Next we need to relate $\frac{1}{d_\jj}\E\langle g|\Lambda_\jj^\dagger \Opkl(a)\Lambda_\jj|g\rangle$ to $\int_{\T^2} a(\x)\,d\x$. Since evaluation shows
$
\frac{1}{d_\jj}\E\langle g|\Lambda_\jj^\dagger \Opkl(a)\Lambda_\jj|g\rangle=\frac{1}{d_\jj}\Tr(P_\jj\Opkl(a)),$ we can estimate
\begin{align*}
\left|\frac{1}{d_\jj}\Tr(P_\jj\Opkl(a))-\int_{\T^2} a(\x)\,d\x\right|&=\left|\frac{1}{d_\jj}\Tr(P_\jj\Opkl(a_0))\right| \\
&\le \frac{1}{d_\jj q(k)}\left|\sum_{\substack{t=-q(k)/2\\t\ne0}}^{q(k)/2-1}e^{2\pi i\jj t/q(k)}\sum_{[\cs{\varepsilon}]\in\mathcal R_{k,\ell}}\langle\cs{\varepsilon}|\Ba_k^t|\cs{\varepsilon}\rangle\langle\cs{\varepsilon}|\Opkl(a_0)|\cs{\varepsilon}\rangle \right|,
\intertext{
where we have used \eqref{eqn:ppoly0} and skipped the $t=0$ term in the sum since $\Tr(\Opkl(a_0)) =0$. Continuing, using the triangle inequality followed by Proposition~\ref{prop:walsh-powers} (Proposition 8.4 of \cite{wbaker}) gives the bound
}
&\le \frac{\|a_0\|_\infty}{d_\jj q(k)}\Bigg(\sum_{\substack{t=-q(k)/2\\t\ne0,-2k}}^{q(k)/2-1}D^{\eta_k(t)}D^{-\eta_k(t)/2}+\begin{cases}1,&D\text{ odd}\\2^k,&D\ge4\text{ even}\end{cases}\Bigg)\\
&\le \frac{\|a_0\|_\infty}{d_\jj q(k)}\left[CD^{k/2}+2^k\oneb_{D\ge4\text{ even}}\right]\\
&=:r_k\|a\|_\infty=O(N^{-1/2}\|a\|_\infty).\numberthis\label{eqn:trbound}
\end{align*}
Combining with \eqref{eqn:normc} and \eqref{eqn:hw}, we obtain for $\epsilon_k\in(2r_k,1)$,
\begin{align*}
\P\bigg[\Big|\langle\varphi^{(k,j)}|&\Opkl(a)|\varphi^{(k,j)}\rangle-\int_{\T^2}a(\x)\,d\x\Big|>\epsilon_k\|a\|_\infty\bigg]\\
&\le
\begin{multlined}[t]
\P\left[\left|\frac{1}{d_\jj}\langle g|\Lambda_\jj^\dagger\Opkl(a)\Lambda_\jj|g\rangle-\int_{\T^2}a(\x)\,d\x\right|>\frac{\epsilon_k}{2}\|a\|_\infty\right]\\
\qquad\qquad+\P\left[\left|\frac{1}{d_\jj}\langle g|\Lambda_\jj^\dagger\Opkl(a)\Lambda_\jj|g\rangle-\langle\varphi^{(k,j)}|\Opkl(a)|\varphi^{(k,j)}\rangle\right|>\frac{\epsilon_k}{2}\|a\|_\infty\right] 
\end{multlined}\\
&\le
\begin{multlined}[t]
\P\left[\left|\frac{1}{d_\jj}\langle g|\Lambda_\jj^\dagger\Opkl(a)\Lambda_\jj|g\rangle-\frac{1}{d_\jj}\E\langle g|\Lambda_\jj^\dagger\Opkl(a)\Lambda_\jj|g\rangle\right|>\left(\frac{\epsilon_k}{2}-r_k\right)\|a\|_\infty\right]
+2e^{-cd_\jj\epsilon_k^2}
\end{multlined}\\
&\le C\exp\left[-c\big(\epsilon_k-O(N^{-1/2})\big)^2\frac{N}{q(k)}\right].
\end{align*}
Choosing $\epsilon_k=N^{-1/2+\delta}$ for any $0<\delta<1/2$ gives the bound $C\exp\left[-cN^{2\delta}/q(k)\right]$, with $q(k)\propto\log N$.
A union bound then gives \eqref{eqn:ind-que} with $\epsilon_k=N^{-1/2+\delta}$.

For the off-diagonal bound \eqref{eqn:therm}, we can apply a union bound followed by Markov's inequality to obtain for any $p\in\N$,
\begin{align}\label{eqn:offdiag-markov}
\P\bigg[\max_{\substack{i,j\in\intbrr{1:N}\\i\ne j}}\left|\langle\varphi^{(k,i)}|\Opkl(a)|\varphi^{(k,j)}\rangle\right|>\frac{1}{N^{\frac{1}{2}-\delta}}\bigg] &\le N^2\max_{i\ne j}\E[|\langle\varphi^{(k,i)}|\Opkl(a)|\varphi^{(k,j)}\rangle|^{p} ]N^{\frac{p}{2}-p\delta}.
\end{align}
Since $\langle\varphi^{(k,i)}|\varphi^{(k,j)}\rangle=0$, we may consider $\Opkl(a_0)$ instead of $\Opkl(a)$.
For $\varphi^{(k,i)}$ from $E_\jj$ and $\varphi^{(k,j)}$ from $E_\jjtwo$, by similar Weingarten calculus as in Section~\ref{sec:off-diag},
specifically \eqref{eqn:offdiag-same} for $\alpha=\beta$ or \eqref{eqn:2p-2} for $\alpha\ne\beta$,
we have for fixed $p$ even,
\begin{align}\label{eqn:2pbound}
\E[|\langle\varphi^{(k,i)}|\Opkl(a_0)|\varphi^{(k,j)}\rangle|^{p}] &\le O\left(\|a_0\|_\infty^{p}N^{p/2}\frac{q(k)^p}{N^p}\right),
\end{align}
where instead of using the cancellation lemmas in Section~\ref{sec:phases} or more precise classical evolution in Section~\ref{sec:lemmas}, it is enough to just use the simple triangle inequality/absolute value inequalities, for $M=\Lambda_\jj^\dagger\Opkl(a_0)\Lambda_\jj$ and $\mathcalboondox M=\Lambda_\jj^\dagger\Opkl(a_0)\Lambda_\jjtwo$,
\begin{align*}
&|\Tr M|=|\Tr(\Opkl(a_0)P_\jj)|\le O\big(\sqrt{N}\|a\|_\infty\big)\quad\text{by \eqref{eqn:trbound}},\\
&|\Tr M^m|\le N\|(\Opkl(a_0)P_\jj)^m\|\le N \|a_0\|_\infty^m,\\
&|\Tr((\mathcalboondox M\mathcalboondox M^\dagger)^m)|\le N\|(\Opkl(a_0)P_\jj \Opkl(a)P_\jjtwo)^m\|\le N\|a_0\|_\infty^{2m},
\end{align*}
for $m\in\N$. Using \eqref{eqn:2pbound} with even $p>2/\delta$ in \eqref{eqn:offdiag-markov} shows
\begin{align}
\P\bigg[\max_{\substack{i,j\in\intbrr{1:N}\\i\ne j}}\left|\langle\varphi^{(k,i)}|\Opkl(a)|\varphi^{(k,j)}\rangle\right|>\frac{1}{N^{\frac{1}{2}-\delta}}\bigg] &\to0,
\end{align}
as $k\to\infty$, which is the off-diagonal statement of \eqref{eqn:therm}. \qed

\subsection*{Acknowledgments}
The author would like to thank Amit Vikram for insightful discussions on the Eigenstate Thermalization Hypothesis.

\bibliographystyle{custom_url_alpha}
\bibliography{fluctuations.bib}

\end{document}